\documentclass[12pt, onecolumn]{IEEEtran} 

\usepackage[utf8]{inputenc} 
\usepackage[T1]{fontenc}
\usepackage{url}
\usepackage{ifthen}
\usepackage{cite}
\usepackage[cmex10]{amsmath} 
\usepackage{amssymb}
\usepackage{algorithm}
\usepackage{algorithmic}
\usepackage{subfigure}
\usepackage{tikz}
\usepackage{algorithm}
\usepackage{algorithmic}
\usepackage{amsmath}
\usepackage{amsthm}
\usepackage{amssymb}
\usepackage{ecltree}
\usepackage{enumerate}
\usepackage{mathrsfs}
\usepackage{xcolor}
\usepackage{comment}
\usepackage{tikz}
\usepackage{wrapfig}
\usepackage{mathtools}
\usetikzlibrary{trees}

\usepackage{here}
\usepackage{cases}
\usepackage{stackengine}
\usepackage{longtable}

\newtheorem{definition}{Definition}
\newtheorem{lemma}{Lemma}
\newtheorem{theorem}{Theorem}

\newtheorem{corollary}{Corollary}
\newtheorem{remark}{Remark}
\newtheorem{example}{Example}

\def\trans{\tau}
\def\PREF{\mathcal{P}}

\def\hdec{\text{-}\mathrm{dec}}
\def\ext{\mathrm{ext}}
\def\irr{\mathrm{irr}}
\def\reg{\mathrm{reg}}
\def\comp{\mathrm{comp}}

\def\hopt{\text{-}\mathrm{opt}}

\def\suff{\mathrm{suff}}
\def\kernel{\mathcal{R}}
\def \prefset{\mathscr{P}}
\newcommand{\quot}[2]{\frac{#1}{#2}}

\newcommand{\argmin}{\mathop{\rm arg~min}\limits}

\newcommand{\eqlab}[2]{\overset{(\mathrm{#1})}{#2}}

\begin{document}

\title{Reduction of Sufficient Number of Code Tables of\\ $k$-Bit Delay Decodable Codes} 
\author{Kengo Hashimoto, Ken-ichi Iwata \thanks{The authors are with University of Fukui, Japan. E-mail: \{khasimot, k-iwata\}@u-fukui.ac.jp}}

\maketitle

\begin{abstract}
A $k$-bit delay decodable code-tuple is a lossless source code that can achieve a smaller average codeword length than Huffman codes by using a finite number of code tables and allowing at most $k$-bit delay for decoding.
It is known that there exists a $k$-bit delay decodable code-tuple with at most $2^{(2^k)}$ code tables that attains the optimal average codeword length among all the $k$-bit delay decodable code-tuples for any given i.i.d.~source distribution.
Namely, it suffices to consider only the code-tuples with at most $2^{(2^k)}$ code tables to accomplish optimality.
In this paper, we propose a method to dramatically reduce the number of code tables to be considered in the theoretical analysis, code construction, and coding process.
\end{abstract}


\section{Introduction}
\label{sec:introduction}

Huffman codes \cite{Huffman1952} achieve the optimal average codeword length in the class of instantaneous (i.e., uniquely decodable without decoding delay) codes with a single code table.
The class of $k$-bit delay decodable code-tuples \cite{JSAIT2022, IEICE2023} can achieve a smaller average codeword length than Huffman codes by using a finite number of code tables and allowing at most $k$-bit delay for decoding.

Algorithms to find a \emph{$k$-bit delay optimal code-tuple}, a $k$-bit delay decodable code-tuple that achieves the minimum average codeword length among all $k$-bit delay decodable code-tuples for a given source distribution, have been proposed for $k \leq 2$.
More specifically, it is shown by \cite{JSAIT2022} that a $1$-bit delay optimal code-tuple is obtained as a Huffman code,
and it is shown by \cite{Hashimoto2021, Hashimoto2023} that a $2$-bit delay optimal code-tuple is obtained as an optimal AIFV-$2$ code \cite{Yamamoto2015, Hu2017}, which can be constructed by the methods in \cite{Yamamoto2015, GolinH23, Golin2021}.

It is known that there exists a $k$-bit delay optimal code-tuple with at most $2^{(2^k)}$ code tables \cite{IEICE2023}.
In other words, we have an upper bound $2^{(2^k)}$ of the sufficient number of code tables that we have to consider to obtain a $k$-bit delay optimal code-tuple.
Improving this upper bound makes constructing a $k$-bit delay optimal code-tuple and other related problems more tractable in theoretical and practical aspects.

In this paper, we propose a method to dramatically reduce the number of code tables to be considered to discuss $k$-bit delay optimal code-tuples in the theoretical analysis, code construction, and coding process using a code-tuple.

\subsection{Contributions and Organization}
The contribution of this paper is as follows.
\begin{itemize} 
\item For an integer $k \geq 0$, we propose a \emph{reduced-code-tuple} (\emph{RCT} for short) consisting of an extremely smaller number $a_k$ of code tables than the upper bound $2^{(2^k)}$ in \cite{IEICE2023} as shown in Table \ref{tab:orbit-num}.
\item Then we show that the construction of a $k$-bit delay optimal code-tuple can be reduced to the construction of an RCT satisfying certain conditions.
\item Further, we present a coding procedure using an RCT equivalent to the coding procedure using a code-tuple.
\end{itemize}
Therefore, our results dramatically reduce the number of code tables to be considered in the theoretical analysis, code construction, and coding process by replacing code-tuples with RCTs. 

This paper is organized as follows.
In Section \ref{sec:preliminary},  we define some notation and introduce the basic definitions in \cite{JSAIT2022, IEICE2023}, including code-tuples.
Then our main results are presented in Section \ref{sec:main}.
Lastly, we conclude this paper in Section \ref{sec:conclusion}.
To clarify the flow of the discussion, we relegate most of the proofs to the appendices.
The main notation is listed in Appendix \ref{sec:notation}.

\subsection{Related Work}
AIFV-$k$ codes \cite{Yamamoto2015, Hu2017} are source codes that can achieve a shorter average codeword length than Huffman codes by using $k$ code tables and allowing at most $k$-bit decoding delay.\footnote
{In the original proposal of AIFV-$k$ codes \cite{Hu2017}, $m$ is used to represent the number of code tables and the length of allowed decoding delay instead of $k$.}
The set of all AIFV-$k$ codes can be viewed as a proper subset of the set of all $k$-bit delay decodable code-tuples.

There are algorithms to construct an optimal AIFV-$k$ code for a given source distribution using an algorithm called the \emph{iterative algorithm},
which reduces the global optimization of the entire code-tuple to an iteration of local optimization on the individual code tables.
Algorithms to find an optimal AIFV-$k$ code by the iterative algorithm via dynamic programming are given by \cite{Yamamoto2015, GolinH23, Golin2021} for $k = 2$ and given by \cite{Fujita2018, Golin2022} for a general integer $k \geq 2$.

The literature \cite{Fujita2019} proposes an iterative algorithm under more general conditions than the construction of an optimal AIFV-$k$ code and proves its validity.
The literature \cite{Golin2024} gives a variant of the iterative algorithm based on a binary search with theoretically bounded time complexity in polynomial time.

The literature \cite{Sugiura2022} presents another generalization of AIFV-$k$ codes and \cite{Sugiura2023} constructs a sub-optimal code in the generalized AIFV-$k$ codes by the iterative algorithm via integer programming.

\section{Preliminaries}
\label{sec:preliminary}

First, we define some notation based on \cite{JSAIT2022, IEICE2023}.
Let $|\mathcal{A}|$ denote the cardinality of a finite set $\mathcal{A}$.
Let $\mathcal{A} \times \mathcal{B}$ denote the Cartesian product of $\mathcal{A}$ and $\mathcal{B}$, that is, $\mathcal{A} \times \mathcal{B} \coloneqq \{(a, b) : a \in \mathcal{A}, b \in \mathcal{B}\}$.
Let $\mathcal{A}^k$ (resp.~$\mathcal{A}^{\leq k}$, $\mathcal{A}^{\geq k}$, $\mathcal{A}^{\ast}$, $\mathcal{A}^{+}$) denote the set of all sequences of length $k$ (resp.~of length less than or equal to $k$, of length greater than or equal to $k$, of finite length, of finite positive length) over a set $\mathcal{A}$.
Thus, $\mathcal{A}^{+} = \mathcal{A}^{\ast} \setminus \{\lambda\}$, where $\lambda$ denotes the empty sequence.
The length of a sequence $\pmb{x}$ is denoted by $|\pmb{x}|$, in particular, $|\lambda| = 0$.
For a sequence $\pmb{x}$ and an integer $1 \leq i \leq |\pmb{x}|$, the $i$-th letter of $\pmb{x}$ is denoted by $x_i$.
For a sequence $\pmb{x}$ and a set $\mathcal{A}$ of sequences, we define
$\pmb{x}\mathcal{A} \coloneqq \{\pmb{x}\pmb{y} : \pmb{y} \in \mathcal{A}\}$.
For a sequence $\pmb{x}$ and an integer $0 \leq k \leq |\pmb{x}|$, we define $[\pmb{x}]_k \coloneqq x_1x_2\ldots x_k$.
Namely, $[\pmb{x}]_k$ is the prefix of length $k$ of $\pmb{x}$, in particular, $[\pmb{x}]_0 = \lambda$.
For a set $\mathcal{A}$ of sequences and an integer $k \geq 0$, we define $[\mathcal{A}]_k \coloneqq \{[\pmb{x}]_k : \pmb{x} \in \mathcal{A}, |\pmb{x}| \geq k\}$.
For an integer $k \geq 0$ and a sequence $\pmb{x} = x_1x_2\ldots x_{n}$ such that $n \geq k$, we define $\suff^k(\pmb{x}) \coloneqq x_{k+1}\ldots x_{n-1}x_{n}$.
Namely, $\suff^k(\pmb{x})$ is the sequence obtained by deleting the first $k$ letters from $\pmb{x}$.
We write $\suff^1(\pmb{x})$ simply as $\suff(\pmb{x})$.
We write $\pmb{x} \preceq \pmb{y}$ if $\pmb{x}$ is a prefix of $\pmb{y}$, that is, there exists a sequence $\pmb{z}$, possibly $\pmb{z} = \lambda$, such that $\pmb{y} = \pmb{x}\pmb{z}$.
Also, we write $\pmb{x} \prec \pmb{y}$ if $\pmb{x} \preceq \pmb{y}$ and $\pmb{x} \neq \pmb{y}$.
If $\pmb{x} \preceq \pmb{y}$, then $\pmb{x}^{-1}\pmb{y}$ denotes the unique sequence $\pmb{z}$ such that $\pmb{x}\pmb{z} = \pmb{y}$.
Note that a notation $\pmb{x}^{-1}$ behaves like the ``inverse element'' of $\pmb{x}$ as stated in the following statements (i)--(iii).
\begin{itemize}
\item[$(\mathrm{i})$] For any $\pmb{x}$, we have $\pmb{x}^{-1}\pmb{x} = \lambda$.
\item[$(\mathrm{ii})$] For any $\pmb{x}$ and $\pmb{y}$ such that $\pmb{x} \preceq \pmb{y}$, we have $\pmb{x}\pmb{x}^{-1}\pmb{y} =\pmb{y}$.
\item[$(\mathrm{iii})$] For any $\pmb{x}, \pmb{y}$, and $\pmb{z}$ such that $\pmb{x}\pmb{y} \preceq \pmb{z}$, we have $(\pmb{x}\pmb{y})^{-1}\pmb{z} = \pmb{y}^{-1}\pmb{x}^{-1}\pmb{z}$.
\end{itemize}
For two mappings $\phi \colon \mathcal{B} \to \mathcal{C}$ and $\psi \colon \mathcal{A} \to \mathcal{B}$, let $\phi \circ \psi \colon \mathcal{A} \to \mathcal{C}$ denote the composite $\phi(\psi(\cdot))$ of $\phi$ and $\psi$.
For a mapping $\phi \colon \mathcal{A} \to \mathcal{B}$ and a subset $\mathcal{X} \subseteq \mathcal{A}$, let $\phi(\mathcal{X})$ denote the image of $\mathcal{X}$ under $\phi$, that is, $\phi(\mathcal{X}) \coloneqq \{\phi(a) : a \in \mathcal{X}\}$.
The main notation used in this paper is listed in Appendix \ref{sec:notation}.

In this paper, we consider binary source coding from a finite source alphabet $\mathcal{S}$ to the binary coding alphabet $\mathcal{C} \coloneqq\{0, 1\}$ with a coding system consisting of a source, an encoder, and a decoder.
We consider an independent and identical distribution (i.i.d.)\ source: each symbol of the source sequence $\pmb{x} \in \mathcal{S}^{\ast}$ is determined independently by a fixed probability distribution $\mu \colon \mathcal{S} \rightarrow (0, 1) \subseteq \mathbb{R}$ such that $\sum_{s \in \mathcal{S}} \mu(s) = 1$.
From now on, we fix the probability distribution $\mu$ arbitrarily and omit the notation $\mu$ even if a value depends on $\mu$ because the discussion in this paper holds for any $\mu$.
Note that we exclude the case where $\mu(s) = 0$ for some $s \in \mathcal{S}$ without loss of generality.
Also, we assume $|\mathcal{S}| \geq 2$.

\subsection{Code-Tuples}
\label{subsec:treepair}

We introduce the notion of code-tuples\cite{JSAIT2022, IEICE2023}, which consists of some code tables $f_i$ and mappings $\trans_i$ to determine which code table to use for each symbol.

\begin{definition} [{\cite[Definition 1]{JSAIT2022}, \cite[Definition 1]{IEICE2023}}]
  \label{def:treepair}
Let $m$ be a positive integer.
An \emph{$m$-code-tuple} $F=(f_0, f_1, \ldots, f_{m-1}, \allowbreak \trans_0, \trans_1, \ldots, \trans_{m-1})$ is a tuple of
$m$ mappings $f_0, f_1, \ldots, f_{m-1}\colon \mathcal{S} \rightarrow \mathcal{C}^{\ast}$ and $m$ mappings $\trans_0, \trans_1, \ldots, \trans_{m-1} \colon \mathcal{S} \rightarrow \{0, 1, 2, \ldots, m-1\}$.
The set of all $m$-code-tuples is denoted by $\mathscr{F}^{(m)}$.
A \emph{code-tuple} is an element of $\mathscr{F} \coloneqq  \mathscr{F}^{(1)} \cup \mathscr{F}^{(2)} \cup \mathscr{F}^{(3)} \cup \cdots$.
\end{definition}

We write a code-tuple $F=(f_0, f_1, \ldots, f_{m-1}, \trans_0, \trans_1, \allowbreak \ldots, \trans_{m-1})$ also as $F=(f, \trans)$ or $F$ for simplicity.
The number of code tables of $F$ is denoted by $|F|$, that is, $|F| \coloneqq m$ for $F \in \mathscr{F}^{(m)}$.
Also, let $[F]$ represent the set $\{0, 1, 2, \ldots, |F|-1\}$ of the indices of $F$, and we refer to $[F]$ as the \emph{domain} of $F$.

\begin{table}
\caption{An example $F = (f, \trans)$ of a code-tuple}
\label{tab:code-tuple}
\centering
\begin{tabular}{c | lclclc}
\hline
$s \in \mathcal{S}$ & $f_0$ & $\trans_0$ & $f_1$ & $\trans_1$ & $f_2$ & $\trans_2$\\
\hline
a & 01 & 0 & 00 & 1 & 1100 & 1\\
b & 10 & 1 & $\lambda$ & 0 & 1110 & 2\\
c & 0100 & 0 & 00111 & 1 & 111000 & 2\\
d & 01 & 2 & 00111 & 2 & 110 & 2\\
\hline
\end{tabular}
\end{table}

\begin{example}
\label{ex:encode}
Table \ref{tab:code-tuple} shows an example of a $3$-code-tuple 
for $\mathcal{S} = \{\mathrm{a}, \mathrm{b}, \mathrm{c}, \mathrm{d}\}$.
\end{example}

Next, we state the coding procedure specified by a code-tuple $F=(f, \trans)$.
First, the encoder and decoder share the used code-tuple $F$ and the index $i_1\in [F]$ of the code table used for the first symbol $x_1$ of the source sequence in advance.
Then the encoding and decoding process with $F$ is described as follows.

\begin{itemize}
\item Encoding: The encoder reads the source sequence $\pmb{x} = x_1x_2\ldots x_n \in \mathcal{S}^{\ast}$ symbol by symbol from the beginning of $\pmb{x}$ and encodes them according to the code tables.
The first symbol $x_1$ is encoded with the code table $f_{i_1}$.
For $x_2, x_3, \ldots, x_n$, we determine which code table to use to encode them according to the mappings $\trans_0, \trans_1, \ldots, \trans_{m-1}$: if the previous symbol $x_{i-1}$ is encoded by the code table $f_j$, then the current symbol $x_i$ is encoded by the code table $f_{\trans_j(x_{i-1})}$.

\item Decoding: 
The decoder reads the codeword sequence $\pmb{c}$ bit by bit from the beginning of $\pmb{c}$.
Each time the decoder reads a bit, the decoder recovers as long prefix of $\pmb{x}$ as the decoder can uniquely identify from the prefix of $\pmb{c}$ already read.
\end{itemize}

\begin{example}[{\cite[Example 2]{IEICE2023}}]
\label{ex:encode}
We consider encoding of a source sequence $\pmb{x} = x_1x_2x_3x_4 \coloneqq \mathrm{badb}$ with the code-tuple $F=(f, \trans)$ in Table \ref{tab:code-tuple}.
If $x_1 = \mathrm{b}$ is encoded with the code table $f_0$, then the encoding process is as follows.
\begin{itemize}
\item $x_1 = \mathrm{b}$ is encoded to $f_0(\mathrm{b}) = 10$. The index of the next code table is $\trans_0(\mathrm{b}) = 1$.
\item $x_2 = \mathrm{a}$ is encoded to $f_1(\mathrm{a}) = 00$. The index of the next code table is $\trans_1(\mathrm{a}) = 1$.
\item $x_3 = \mathrm{d}$ is encoded to $f_1(\mathrm{d}) = 00111$. The index of the next code table is $\trans_1(\mathrm{d}) = 2$.
\item $x_4 = \mathrm{b}$ is encoded to $f_2(\mathrm{b}) = 1110$. The index of the next code table is $\trans_2(\mathrm{b}) = 2$.
\end{itemize}
As the result, we obtain a codeword sequence $\pmb{c} \coloneqq f_0(\mathrm{b})f_1(\mathrm{a})f_1(\mathrm{d})f_2(\mathrm{b}) = 1000001111110$.

The decoding process of $\pmb{c} = 1000001111110$ is as follows.
\begin{itemize}
\item After reading the prefix $10$ of $\pmb{c}$, the decoder can uniquely identify $x_1 = \mathrm{b}$ and $10 = f_0(\mathrm{b})$. The decoder can also know that $x_2$ should be decoded with $f_{\trans_0(\mathrm{b})} = f_1$.
\item After reading the prefix $1000 = f_0(\mathrm{b})f_1(\mathrm{a})$ of $\pmb{c}$, the decoder still cannot uniquely identify  $x_2 = \mathrm{a}$ because there remain three possible cases: the case $x_2 = \mathrm{a}$, the case $x_2 = \mathrm{c}$, and the case $x_2 = \mathrm{d}$.
\item After reading the prefix $10000$ of $\pmb{c}$, the decoder can uniquely identify $x_2 = \mathrm{a}$ and $10000 = f_0(\mathrm{b})f_1(\mathrm{a})0$. The decoder can also know that $x_3$ should be decoded with $f_{\trans_1(\mathrm{a})} = f_1$.
\item After reading the prefix $100000111 = f_0(\mathrm{b})f_1(\mathrm{a})\allowbreak f_1(\mathrm{d})$ of $\pmb{c}$, the decoder still cannot uniquely identify $x_3 = \mathrm{d}$ because there remain two possible cases: the case $x_3 = \mathrm{c}$ and the case $x_3 = \mathrm{d}$.
\item After reading the prefix $10000011111$ of $\pmb{c}$, the decoder can uniquely identify $x_3 = \mathrm{d}$ and $10000011111 \allowbreak = f_0(\mathrm{b})f_1(\mathrm{a})f_1(\mathrm{d})11$.
The decoder can also know that $x_4$ should be decoded with $f_{\trans_1(\mathrm{d})} = f_2$.
\item After reading the prefix $\pmb{c} = 1000001111110$, the decoder can uniquely identify $x_4 = \mathrm{b}$ and $1000001111\allowbreak110 = f_0(\mathrm{b})f_1(\mathrm{a})f_1(\mathrm{d})f_2(\mathrm{b})$.
\end{itemize}
As the result, the decoder recovers the original sequence $\pmb{x} = \mathrm{badb}$.
\end{example}

Let $i \in [F]$ and $\pmb{x} \in \mathcal{S}^{\ast}$.
Then $f^{\ast}_i(\pmb{x}) \in \mathcal{C}^{\ast}$ is defined as the codeword sequence in the case where $x_1$ is encoded with $f_i$.
Also, $\trans^{\ast}_i(\pmb{x}) \in [F]$ is defined as the index of the code table used next after encoding $\pmb{x}$ in the case where $x_1$ is encoded with $f_i$.
The formal definitions are given in the following Definition \ref{def:f_T} as recursive formulas.

\begin{definition}[{\cite[Definition 2]{JSAIT2022}}, {\cite[Definition 2]{IEICE2023}}]
 \label{def:f_T}
For $F=(f, \trans) \in \mathscr{F}$ and $i \in [F]$, the mapping $f_i^{\ast} \colon \mathcal{S}^{\ast} \rightarrow \mathcal{C}^{\ast}$ and the mapping $\trans_i^{\ast} \colon \mathcal{S}^{\ast} \rightarrow [F]$ are defined as
\begin{equation}
\label{eq:fstar}
f_i^{\ast}(\pmb{x}) = 
\begin{cases}
\lambda &\,\,\text{if}\,\, \pmb{x} = \lambda,\\
f_i(x_1)f_{\trans_i(x_1)}^{\ast}(\suff(\pmb{x})) &\,\,\text{if}\,\, \pmb{x} \neq \lambda,\\
\end{cases}
\end{equation}
\begin{equation}
\label{eq:tstar}
\trans_i^{\ast}(\pmb{x}) = 
\begin{cases}
i &\,\,\text{if}\,\, \pmb{x} = \lambda,\\
\trans^{\ast}_{\trans_i(x_1)}(\suff(\pmb{x})) &\,\,\text{if}\,\, \pmb{x} \neq \lambda\\
\end{cases}
\end{equation}
for $\pmb{x} = x_1 x_2 \ldots x_{n} \in \mathcal{S}^{\ast}$.
\end{definition}

The next Lemma \ref{lem:f_T} follows from Definition \ref{def:f_T}.
\begin{lemma}[{\cite[Lemma 1]{JSAIT2022}}, {\cite[Lemma 1]{IEICE2023}}]
\label{lem:f_T}
For any $F=(f, \trans) \in \mathscr{F}$, $i \in [F]$, and $\pmb{x}, \pmb{y} \in \mathcal{S}^{\ast}$, 
the following statements (i)--(iii) hold.
\begin{enumerate}[(i)]
\item $f_i^{\ast}(\pmb{x} \pmb{y}) = f_i^{\ast}(\pmb{x}) f^{\ast}_{\trans_i^{\ast}(\pmb{x})}(\pmb{y})$. 
\item $\trans_i^{\ast}(\pmb{x} \pmb{y}) = \trans^{\ast}_{\trans^{\ast}_i(\pmb{x})}(\pmb{y})$.
\item If $\pmb{x} \preceq \pmb{y}$, then $f^{\ast}_i(\pmb{x}) \preceq f^{\ast}_i(\pmb{y})$.
\end{enumerate}
\end{lemma}

\subsection{$k$-bit Delay Decodable Code-tuples}

A code-tuple is said to be \emph{$k$-bit delay decodable} if the decoder can always uniquely identify each source symbol by reading the additional $k$ bits of the codeword sequence.
To state the formal definition of a $k$-bit delay decodable code-tuple,
we introduce the following Definitions \ref{def:pref}.

\begin{definition}[{\cite[Definitions 3 and 4]{IEICE2023}}]
\label{def:pref}
For an integer $k \geq 0$, $F=(f, \trans) \in \mathscr{F}, i \in [F]$, and $\pmb{b} \in \mathcal{C}^{\ast}$, we define 
\begin{equation}
\label{eq:pref1}
\PREF^k_{F, i}(\pmb{b}) \coloneqq \{\pmb{c} \in \mathcal{C}^k : \exists \pmb{x} \in \mathcal{S}^{+} \,\,\,\mathtt{s.t.}\,\,  (f^{\ast}_i(\pmb{x}) \succeq \pmb{b}\pmb{c}, f_i(x_1) \succeq \pmb{b})  \},
\end{equation}
\begin{equation}
\label{eq:pref2}
\bar{\PREF}^k_{F, i}(\pmb{b}) \coloneqq \{\pmb{c} \in \mathcal{C}^k : \exists\pmb{x} \in \mathcal{S}^{+} \,\,\,\mathtt{s.t.}\,\, (f^{\ast}_i(\pmb{x}) \succeq \pmb{b}\pmb{c}, f_i(x_1) \succ \pmb{b})  \}.
\end{equation}
Namely, $\PREF^k_{F, i}(\pmb{b})$ (resp.~$\bar{\PREF}^k_{F, i}(\pmb{b})$) is the set of all $\pmb{c} \in \mathcal{C}^k$ such that there exists $\pmb{x} \in \mathcal{S}^{+}$ satisfying $f^{\ast}_i(\pmb{x}) \succeq \pmb{b}\pmb{c}$ and $f_i(x_1) \succeq \pmb{b}$ (resp.~$f_i(x_1) \succ \pmb{b}$).
\end{definition}
We write $\PREF^k_{F, i}(\lambda)$ (resp.~$\bar{\PREF}^k_{F, i}(\lambda)$) as $\PREF^k_{F, i}$ (resp.~$\bar{\PREF}^k_{F, i}$) for simplicity.
Note that
\begin{eqnarray}
\PREF^k_{F, i} = \{\pmb{c} \in \mathcal{C}^k : {\exists}\pmb{x} \in \mathcal{S}^{+} \,\,\mathtt{s.t.}\,\ f^{\ast}_i(\pmb{x}) \succeq \pmb{c}\} 
= \{\pmb{c} \in \mathcal{C}^k : {\exists}\pmb{x} \in \mathcal{S}^{\ast} \,\,\mathtt{s.t.}\,\ f^{\ast}_i(\pmb{x}) \succeq \pmb{c}\}.\label{eq:pref3}
\end{eqnarray}

The following lemma gives other representation forms of $\PREF^k_{F, i}(\pmb{b})$ and $\bar{\PREF}^k_{F, i}(\pmb{b})$.
See Appendix \ref{subsec:proof-pref} for the proof of Lemma \ref{lem:pref}.

\begin{lemma}
\label{lem:pref}
For an integer $k \geq 0$, $F=(f, \trans) \in \mathscr{F}$, and $i \in [F]$, 
the following statements (i) and (ii) hold.
\begin{enumerate}[(i)]
\item \begin{equation}
\label{eq:9fhqriquk4ht}
\PREF^k_{F, i} = \bigcup_{s \in \mathcal{S}} \left[ f_i(s) \PREF^k_{F, \trans_i(s)} \right]_k.
\end{equation}
\item 
For any $\pmb{b} \in \mathcal{C}^{\ast}$, we have
\begin{equation}
\label{eq:gfo3oig06zin}
\bar{\PREF}^k_{F, i}(\pmb{b}) = \bigcup_{\substack{s \in \mathcal{S},\\ f_i(s) \succ \pmb{b}}} \left[ \pmb{b}^{-1}f_i(s) \PREF^k_{F, \trans_i(s)} \right]_k.
\end{equation}
\end{enumerate}
\end{lemma}

By using Definition \ref{def:pref}, the condition for a code-tuple to be decodable in at most $k$-bit delay is given as follows.
Refer to \cite{IEICE2023} for detailed discussion.

 \begin{definition}[{\cite[Definition 5]{IEICE2023}}]
  \label{def:k-bitdelay}
 Let $k \geq 0$ be an integer. 
A code-tuple $F=(f, \trans)$ is said to be \emph{$k$-bit delay decodable} if the following conditions (a) and (b) hold.
\begin{enumerate}[(a)]
\item For any $i \in [F]$ and $s \in \mathcal{S}$, it holds that $\PREF^k_{F, \trans_i(s)} \cap \bar{\PREF}^k_{F, i}(f_i(s)) = \emptyset$.
\item For any $i \in [F]$ and $s, s' \in \mathcal{S}$, if $s \neq s'$ and $f_i(s) = f_i(s')$, then $\PREF^k_{F, \trans_i(s)} \cap \PREF^k_{F, \trans_i(s')} =  \emptyset$.
\end{enumerate}
 For an integer $k \geq 0$, we define $\mathscr{F}_{k\hdec}$ as the set of all $k$-bit delay decodable code-tuples, that is, 
$\mathscr{F}_{k\hdec} \coloneqq \{F \in \mathscr{F} : F \text{ is } k \text{-bit delay decodable} \}$.
\end{definition}

\begin{remark}
\label{rem:not-unique}
In a precise sense, a $k$-bit delay decodable code-tuple $F = (f, \trans)$ is not necessarily uniquely decodable,
that is, the mappings $f^{\ast}_0, f^{\ast}_1, \ldots, f^{\ast}_{|F|-1}$ are not necessarily injective. 
For example, for $F \in \mathscr{F}_{2\hdec}$ in Table \ref{tab:code-tuple}, we have ${f_0}^{\ast}(\mathrm{bc}) = 1000111 = {f_0}^{\ast}(\mathrm{bd})$.
In general, it is possible that the decoder cannot uniquely recover the last few symbols of the original source sequence when the rest of the codeword sequence is less than $k$ bits.
Therefore, we should append additional $k$ bits to recover the tail of the codeword sequence in practice.
More specifically, when encoding $\pmb{x} \in \mathcal{S}^{\ast}$ starting from $f_i$, it suffices to append an arbitrarily chosen $\pmb{c} \in \PREF^k_{F, \trans^{\ast}_i(\pmb{x})}$ to the codeword sequence $f^{\ast}_i(\pmb{x})$.
This additional $k$ bits can be an overhead for a short source sequence.
However, in this paper, we consider an asymptotic evaluation of the average codeword length for a sufficiently long source sequence,
and thus we do not explicitly discuss this additional $k$ bits, which does not affect the asymptotic performance.
\end{remark}

Let $F = (f, \trans)$ be a code-tuple in which all code tables map every symbol to the empty sequence, that is, $f_i(s) = \lambda$ holds for any $i \in [F]$ and $s \in \mathcal{S}$.
Then $\PREF^k_{F, i} = \bar{\PREF}^k_{F, i}(f_i(s)) = \lambda$ holds for any $i \in [F]$ and $s \in \mathcal{S}$, and thus $F$ is $k$-bit delay decodable for any $k \geq 1$ by Definition \ref{def:k-bitdelay}.
However, this code-tuple $F$ is obviously useless and should be excluded from our consideration.
To exclude such abnormal and useless code-tuples, we introduce a set $\mathscr{F}_{\ext}$ in the following Definition \ref{def:F_ext}.

\begin{definition}[{\cite[Definition 6]{IEICE2023}}]
\label{def:F_ext}
A code-tuple $F$ is said to be \emph{extendable} if $\PREF^1_{F, i} \neq \emptyset$ for any $i \in [F]$.
The set of all extendable code-tuples is denoted by $\mathscr{F}_{\ext}$, that is,
$\mathscr{F}_{\ext} \coloneqq \{F \in \mathscr{F} : {\forall}i \in [F], \PREF^1_{F, i} \neq \emptyset\}$.
\end{definition}

By the following Lemma \ref{lem:F_ext}, for an extendable code-tuple $F=(f, \trans)$, we can ``extend'' the length of $f^{\ast}_i(\pmb{x})$ up to an arbitrary integer by extending $\pmb{x}$ appropriately.

\begin{lemma} [{\cite[Lemma 3]{IEICE2023}}]
\label{lem:F_ext}
A code-tuple $F=(f, \trans)$ is extendable if and only if for any $i \in [F]$ and integer $l \geq 0$, 
there exists $\pmb{x} \in \mathcal{S}^{\ast}$ such that $|f^{\ast}_i(\pmb{x})| \geq l$.
\end{lemma}

\subsection{Average Codeword Length of Code-Tuple}
\label{subsec:evaluation}

In this subsection, we introduce the average codeword length $L(F)$ of a code-tuple $F$.
First, we state the definitions of the transition probability matrix and stationary distributions of a code-tuple in the following Definitions \ref{def:transprobability} and \ref{def:stationary}.

 \begin{definition}[{\cite[Definition 6]{JSAIT2022}}, {\cite[Definition 7]{IEICE2023}}]
\label{def:transprobability}
For $F=(f, \trans) \in \mathscr{F}$ and $i, j \in [F]$, the \emph{transition probability} $Q_{i,j}(F)$ is defined as
\begin{equation*}
 Q_{i,j}(F) \coloneqq \sum_{s \in \mathcal{S}, \trans_i(s) = j} \mu(s).
 \end{equation*}
The \emph{transition probability matrix} $Q(F)$ is defined as the following $|F|\times|F|$ matrix:
 \begin{equation*}
  \left[
    \begin{array}{cccc}
      Q_{0,0}(F) & Q_{0,1}(F) & \cdots & Q_{0, |F|-1}(F) \\
       Q_{1,0}(F) &  Q_{1,1}(F) & \cdots &  Q_{1, |F|-1}(F) \\
      \vdots & \vdots & \ddots & \vdots \\
      Q_{|F|-1, 0}(F) &  Q_{|F|-1, 1}(F) & \cdots &  Q_{|F|-1, |F|-1}(F) 
    \end{array}
  \right].
 \end{equation*}
  \end{definition}

 \begin{definition}[{\cite[Definition 7]{JSAIT2022}}{\cite[Definition 9]{IEICE2023}}]
\label{def:stationary}
For $F \in \mathscr{F}$, a solution $\pmb{\pi} = (\pi_0, \allowbreak \pi_1, \ldots, \pi_{|F|-1}) \in \mathbb{R}^{|F|}$ of the following simultaneous equations (\ref{eq:stationary1}) and (\ref{eq:stationary2})
is called a \emph{stationary distribution of $F$}:
\begin{numcases}{}
\pmb{\pi}Q(F) = \pmb{\pi}\label{eq:stationary1},\\
 \sum_{i \in [F]} \pi_i = 1. \label{eq:stationary2}
 \end{numcases}
\end{definition}

It is guaranteed by {\cite[Lemma 6]{IEICE2023}} that every code-tuple has at least one stationary distribution $\pmb{\pi} = (\pi_0, \pi_1, \ldots, \allowbreak \pi_{|F|-1})$ such that $\pi_i \geq 0$ for any $i \in [F]$.
A code-tuple with a unique stationary distribution is said to be \emph{regular} as the next Definition \ref{def:regular}.
The average codeword length is defined in Definition \ref{def:evaluation} only for regular code-tuples by using the unique stationary distribution.
 
 \begin{definition}[{\cite[Definition 7]{JSAIT2022}}{\cite[Definition 9]{IEICE2023}}]
 \label{def:regular}
A code-tuple $F$ is said to be \emph{regular} if $F$ has a unique stationary distribution.
 The set of all regular code-tuples is denoted by $\mathscr{F}_{\reg}$, that is,
$\mathscr{F}_{\reg} \coloneqq \{F \in \mathscr{F} : F \text{ is regular}\}.$
 For $F \in \mathscr{F}_{\reg}$, the unique stationary distribution of $F$ is denoted by $\pmb{\pi}(F) = (\pi_0(F), \pi_1(F), \ldots, \pi_{|F|-1}(F))$.
 \end{definition}

Then $\mathscr{F}_{\reg}$ is characterized by the following lemma.
 
 \begin{lemma}[{\cite[Lemma 8]{IEICE2023}}]
\label{lem:kernel}
For any $F \in \mathscr{F}$, the following statements (i) and (ii) hold:
\begin{enumerate}[(i)]
\item $\mathscr{F}_{\reg} = \{F \in \mathscr{F} : \kernel_F \neq \emptyset\}$,
\item if $F \in \mathscr{F}_{\reg}$, then $\kernel_F = \{ i \in [F] : \pi_i(F) > 0\}$,
\end{enumerate}
where
\begin{equation*}
\label{eq:7tuxvj14yeno}
\kernel_F \coloneqq \{i \in [F] : {\forall}j \in [F], \exists\pmb{x} \in \mathcal{S}^{\ast} \,\,\mathtt{s.t.}\,\, \trans^{\ast}_j(\pmb{x}) = i\}.
\end{equation*}
\end{lemma}
 
\begin{definition} [{\cite[Definition 8]{JSAIT2022}}, {\cite[Definition 10]{IEICE2023}}]
 \label{def:evaluation}
 For $F=(f, \trans) \in \mathscr{F}$ and $i \in [F]$, the \emph{average codeword length $L_i(F)$ of the single code table} $f_i : \mathcal{S} \rightarrow \mathcal{C}^{\ast}$ is defined as
  \begin{equation*}
 L_i(F) \coloneqq \sum_{s \in \mathcal{S}} |f_i(s)| \mu(s).
  \end{equation*}
For $F \in \mathscr{F}_{\reg}$,
the \emph{average codeword length $L(F)$ of the code-tuple $F$} is defined as 
 \begin{equation}
 \label{eq:evaluation}
 L(F) \coloneqq \sum_{i \in [F]} \pi_i(F)L_i(F).
 \end{equation}
 \end{definition}
 
 The code tables $f_i$ of $F \in \mathscr{F}_{\reg}$ such that $\pi_i(F) = 0$, equivalently $i \not\in \kernel_{F}$ by Lemma \ref{lem:kernel} (ii), do not contribute to the average codeword length $L(F)$ by (\ref{eq:evaluation}).
A code-tuple is said to be \emph{irreducible} if it does not have such non-essential code tables.

\begin{definition}[{\cite[Definition 13]{IEICE2023}}]
\label{def:F_irr}
A code-tuple $F$ is said to be \emph{irreducible} if $\kernel_F = [F]$.
We define $\mathscr{F}_{\irr}$ as the set of all irreducible code-tuples, that is,
$\mathscr{F}_{\irr} \coloneqq \{F \in \mathscr{F} :\kernel_F = [F]\}.$
\end{definition}
Note that 
$\mathscr{F}_{\irr}
= \{F \in \mathscr{F} : \kernel_{F} = [F]\} 
\subseteq  \{F \in \mathscr{F} : \kernel_{F} \neq \emptyset\} 
= \mathscr{F}_{\reg}$.

For use in later proof, we here state the following Definition \ref{def:closed} and Lemma \ref{lem:closed} on irreducible code-tuples.
See Appendix \ref{subsec:proof-closed} for the proof of Lemma \ref{lem:closed}.
\begin{definition}
\label{def:closed}
For $F=(f, \trans) \in \mathscr{F}$ and $\mathcal{I} \subseteq [F]$, we introduce the following definitions (i) and (ii).
\begin{enumerate}[(i)]
\item The set $\mathcal{I}$ is called a \emph{closed set of $F$} if
for any $i \in \mathcal{I}$ and $\pmb{x} \in \mathcal{S}^{\ast}$, it holds that $\trans^{\ast}_{i}(\pmb{x}) \in \mathcal{I}$.
\item The set $\mathcal{I}$ is called a \emph{minimal non-empty closed set of $F$} if $\mathcal{I}$ satisfies the following two conditions:
\begin{itemize}
\item $\mathcal{I}$ is a non-empty closed set of $F$;
\item  for any $\emptyset \subsetneq \mathcal{J} \subsetneq \mathcal{I}$, the set $\mathcal{J}$ is not a closed set of $F$.
\end{itemize}
In other words, $\mathcal{I}$ does not have a non-empty proper subset that is a closed set of $F$.
\end{enumerate}
\end{definition}

\begin{lemma}
\label{lem:closed}
For any $F \in \mathscr{F}$, if $[F]$ is a minimal non-empty closed set of $F$, then $F \in \mathscr{F}_{\irr}$.
\end{lemma}

\subsection{$k$-Bit Delay Optimal Code-Tuples}
\label{subsec:F_opt}
 
A $k$-bit delay optimal code-tuple is a code-tuple with the optimal average codeword length in $\mathscr{F}_{\reg} \cap \mathscr{F}_{\ext} \cap \mathscr{F}_{k\hdec}$ as defined below.

\begin{definition}
\label{def:optimalset}
Let $k \geq 0$ be an integer.
A code-tuple $F \in \mathscr{F}_{\reg} \cap \mathscr{F}_{\ext} \cap \mathscr{F}_{k\hdec}$ is said to be \emph{$k$-bit delay optimal}
if for any $F' \in \mathscr{F}_{\reg} \cap \mathscr{F}_{\ext} \cap \mathscr{F}_{k\hdec}$, it holds that $L(F) \leq L(F')$.
We define $\mathscr{F}_{k\hopt}$ as the set of all $k$-bit delay optimal code-tuples.
\end{definition}

\begin{remark}
\label{rem:transitionprobability}
Note that $Q(F)$, $L(F)$, $\pmb{\pi}(F)$, and $\mathscr{F}_{k\hopt}$ depend on $\mu$ and we are now discussing a fixed $\mu$.
On the other hand, $\mathscr{F}_{\reg}$ and $\mathscr{F}_{\irr}$ are determined independently of $\mu$ since $\kernel_{F}$ does not depend on $\mu$.
\end{remark} 

The following result on $k$-bit delay optimal code-tuples is known \cite{IEICE2023}.

\begin{lemma}[{\cite[Theorem 1]{IEICE2023}}]
\label{thm:differ}
For any integer $k \geq 0$ and $F \in \mathscr{F}_{\reg} \cap \mathscr{F}_{\ext} \cap \mathscr{F}_{k\hdec}$, there exists $F^{\dagger} \in \mathscr{F}$ satisfying the following conditions (a)--(d), where $\prefset^k_F \coloneqq \{\PREF^k_{F, i} : i \in [F]\}$.
\begin{enumerate}[(a)]
\item $F^{\dagger}  \in \mathscr{F}_{\irr} \cap \mathscr{F}_{\ext} \cap \mathscr{F}_{k\hdec}$.
\item $L(F^{\dagger}) \leq L(F)$.
\item $\prefset^k_{F^{\dagger}} \subseteq \prefset^k_F$.
\item $|\prefset^k_{F^{\dagger}}| = |F^{\dagger}|$.
\end{enumerate}
\end{lemma}

Lemma \ref{thm:differ} implies that it suffices to consider only code-tuples satisfying the conditions (a)--(d) to obtain a $k$-bit delay optimal code-tuple.
In particular, the condition (d) implies $|F^{\dagger}| = |\prefset^k_{F^{\dagger}}| \leq |\mathscr{C}^k| = 2^{(2^k)}$, where $\mathscr{C}^k$ denotes the power set of $\mathcal{C}^k$.
Therefore, it suffices to consider only code-tuples with at most $2^{(2^k)}$ code tables.
In other words, we have an upper bound $2^{(2^k)}$ of the sufficient number of code tables to obtain a $k$-bit delay optimal code-tuple.

Note that $|\prefset^k_{F^{\dagger}}| = |F^{\dagger}|$ means that $\PREF^k_{F^{\dagger}, 0}, \PREF^k_{F^{\dagger}, 1}, \ldots, \PREF^k_{F^{\dagger}, |F^{\dagger}|-1}$ are distinct.
Intuitively, code tables $f_i$ with the same sets $\PREF^k_{F, i}$ play the same role, and having multiple code tables with the same role is redundant.

\begin{remark}
\label{rem:differ}
More precisely, since $\PREF^k_{F, i} \neq \emptyset$ must hold for any $F \in \mathscr{F}_{\ext}$ and $i \in [F]$, 
we have an upper bound $|F^{\dagger}| = |\prefset^k_{F^{\dagger}}| \leq |\mathscr{C}^k \setminus \{\emptyset\}| = 2^{(2^k)} - 1$.
Additionally, it is known by \cite{JSAIT2022} that there exists a $k$-bit delay optimal code-tuple $F$ such that
$\PREF^1_{F, i} = \{0, 1\}$ (equivalently, $\PREF^k_{F, i} \not\subseteq 0\mathcal{C}^{k-1}$ and $\PREF^k_{F, i} \not\subseteq 1\mathcal{C}^{k-1}$) holds for any $i \in [F]$. This fact slightly improves the upper bound to $2^{(2^k)} - 2^{(2^{k-1} + 1)} + 1$.
\end{remark}

\section{Main Results}
\label{sec:main}

In this section, we propose a method to reduce the number of code tables to be considered when discussing $k$-bit delay decodable code-tuples as the main results.
We first describe the basic idea of the reduction in the following paragraphs.

As stated in the previous subsection, a code-tuple does not need to have multiple code tables sharing the same sets $\PREF^k_{F, i}$.
Now, we consider whether code tables sharing ``symmetric'' sets $\PREF^k_{F, i}$ can be composed ``symmetrically.''
For example, let $i, j \in [F]$ satisfy $\PREF^k_{F, i} = \{00001, 10101, 11111\}$ and $\PREF^k_{F, j} = \{11110, 01010, 00000\}$, where $k = 5$.
Then it holds that $\PREF^k_{F, j} = \phi(\PREF^k_{F, i})$, where $\phi \colon \mathcal{C}^{\ast} \to \mathcal{C}^{\ast}$ is the mapping that flips every bit of the sequence.
Namely, the roles of these two code tables are symmetric in the sense that flipping all the bits of $\PREF^k_{F, i}$ yields $\PREF^k_{F, j}$ and vice versa.
Therefore, it is expected that imposing the constraint that these two code tables must be composed symmetrically does not worsen the average codeword length.

Developing this idea, we introduce a class $\Phi_k$ of mappings $\phi \colon \mathcal{C}^{\ast} \to \mathcal{C}^{\ast}$ satisfying certain conditions,
and we consider $\PREF^k_{F, i}$ and $\PREF^k_{F, j}$ to be ``symmetric'' if $\PREF^k_{F, j} = \phi(\PREF^k_{F, i})$ for some $\phi \in \Phi_k$.
This symmetry forms an equivalence relation of $\mathscr{C}^k$ and divides $\mathscr{C}^k$ into equivalence classes.
We impose the constraint that once only the code table for the representative of each equivalence class is defined,
the other code tables must be constructed symmetrically from the corresponding representative.

This constraint reduces the number of code tables to be considered when discussing $k$-bit delay code-tuples
because it suffices to consider only the code tables for the representatives in effect.
In this section, we formalize our idea as a \emph{reduced-code-tuple (RCT)}, in which a code table is associated only with each representative of the equivalence class.
Then we mathematically prove that even if we impose the constraint that we first construct an RCT and then determine the code-tuple symmetrically and automatically based on the RCT, the average codeword length does not get worse, and our idea works well.

This section consists of three subsections.
In Subsection \ref{subsec:Phi}, we first introduce the set $\Phi_k$ of all mappings $\phi \colon \mathcal{C}^{\ast} \to \mathcal{C}^{\ast}$ satisfying certain conditions
and establish an equivalence relation $\sim_k$ of the power set $\mathscr{C}^k = 2^{\mathcal{C}^k}$ using $\Phi_k$.
Then we classify the power set $\mathscr{C}^k$ of $\mathcal{C}^k$ into equivalence classes and
give the number of the equivalence classes as Theorem \ref{thm:orbit-num}.
In Subsection \ref{subsec:semi}, we propose an RCT, in which each equivalence class is associated with at most one code table;
thus, the number of code tables of an RCT does not exceed the number of equivalence classes.
Then throughout Subsections \ref{subsec:semi} and \ref{subsec:expand}, we prove Theorem \ref{thm:equiv} that finding a $k$-bit delay optimal code-tuple can be reduced to finding a ``$k$-bit delay optimal'' RCT.
Further, we also describe a coding procedure using an RCT instead of a code-tuple in Subsection \ref{subsec:expand}.

\subsection{The Set $\Phi$ of Mappings}
\label{subsec:Phi}

We define $\Phi_k$ as the set of all mappings that preserve the length, prefix relations, and the suffix after the $k$-th letter of sequences.
The formal definition of $\Phi_k$ is as follows.

\begin{definition}
\label{def:Phi}
For an integer $k \geq 0$, we define $\Phi_k$ as the set of all mappings $\phi \colon \mathcal{C}^{\ast} \to \mathcal{C}^{\ast}$ satisfying the following conditions (a)--(c).
\begin{enumerate}[(a)]
\item For any $\pmb{b} \in \mathcal{C}^{\ast}$, it holds that $|\pmb{b}| = |\phi(\pmb{b})|$.
\item For any $\pmb{b}, \pmb{b}' \in \mathcal{C}^{\ast}$, the following equivalence holds:
$\pmb{b} \preceq \pmb{b}' \iff \phi(\pmb{b}) \preceq\phi(\pmb{b}')$.
\item For any $\pmb{b} \in \mathcal{C}^{\geq k}$, it holds that $\suff^k(\pmb{b}) = \suff^k(\phi(\pmb{b}))$.
\end{enumerate}
\end{definition}

The third condition (c) in Definition \ref{def:Phi} is essentially unnecessary.
For the length $k$ of allowed decoding delay, the behavior of a mapping $\phi$ toward $\suff^k(\pmb{x})$ for each $\pmb{b} \in \mathcal{C}^{\geq k}$ does not affect our discussion.
Hence, we consider only mappings that do not change $\suff^k(\pmb{x})$ for simplicity.
As described later, this additional condition makes $\Phi_k$ a finite set and allows a mapping $\phi \in \Phi_k$ to be represented with a finite number of bits in the coding procedure.

Directly from Definition \ref{def:Phi}, we can see the following inclusion:
\begin{equation*}
\Phi_0 \subsetneq \Phi_1 \subsetneq \Phi_2 \subsetneq \cdots.
\end{equation*}

All mappings of $\Phi_k$ are bijections as follows. See Appendix \ref{subsec:proof-phi-bijection} for the proof of Lemma \ref{lem:phi-bijection}.
\begin{lemma}
\label{lem:phi-bijection}
All mappings $\phi \in \Phi_k$ are bijective.
\end{lemma}

By the bijectivity, for any integer $k \geq 0$, $\phi \in \Phi_k$, and $A \in \mathscr{C}^k$, we have
\begin{equation}
\label{eq:mpd5l92t1cji}
|A| = |\phi(A)|.
\end{equation}
Also, for any $\pmb{b} \in \mathcal{C}^{\geq k}$, we have
\begin{equation}
\label{eq:29zjeuudznan}
 \phi( \left[ \pmb{b} \right]_k) = \left[ \phi(\pmb{b}) \right]_k
\end{equation}
because
$\phi([\pmb{b}]_k) \preceq \phi(\pmb{b})$ is implied by $[\pmb{b}]_k \preceq \pmb{b}$;
and $|\phi([\pmb{b}]_k)| = |[\pmb{b}]_k| = k$ by Definition \ref{def:Phi} (a).
This yields
\begin{equation}
\label{eq:s9qdkaq7m7ap}
\phi( \left[ A \right]_k) = \left[ \phi(A) \right]_k.
\end{equation}

We also define a mapping $\quot{\phi}{\pmb{d}} \colon \mathcal{C}^{\ast} \to \mathcal{C}^{\ast}$ as the following Definition \ref{def:phi-quot}.

\begin{definition}
\label{def:phi-quot}
For an integer $k \geq 0$, $\phi \in \Phi_k$, and $\pmb{d} \in \mathcal{C}^{\ast}$, we define a mapping $\quot{\phi}{\pmb{d}} \colon \mathcal{C}^{\ast} \to \mathcal{C}^{\ast}$ as
\begin{equation}
\label{eq:r6xuob4upqw5}
\quot{\phi}{\pmb{d}}(\pmb{b}) \coloneqq \phi(\pmb{d})^{-1}\phi(\pmb{d}\pmb{b})
\end{equation}
for $\pmb{b} \in \mathcal{C}^{\ast}$.
Note that the right-hand side $\phi(\pmb{d})^{-1}\phi(\pmb{d}\pmb{b})$ in (\ref{eq:r6xuob4upqw5}) is indeed defined since $\phi(\pmb{d}) \preceq \phi(\pmb{d}\pmb{b})$ is implied by
$\pmb{d} \preceq \pmb{db}$ and Definition \ref{def:Phi} (b).
\end{definition}

Then the next Lemma \ref{lem:phi-quot} holds, the proofs of which is given in Appendix \ref{subsec:proof-phi-quot}.

\begin{lemma}
\label{lem:phi-quot}
For any integer $k \geq 0$ and $\phi \in \Phi_k$, the following statements (i) and (ii) hold.
\begin{enumerate}[(i)]
\item For any $\pmb{b}, \pmb{b}' \in \mathcal{C}^{\ast}$, we have $\phi(\pmb{b}\pmb{b}') = \phi(\pmb{b})\quot{\phi}{\pmb{b}}(\pmb{b}')$.
\item For any $\pmb{d} \in \mathcal{C}^{\ast}$, we have $\quot{\phi}{\pmb{d}} \in \Phi_{K}$, where $K \coloneqq \max\{0, k-|\pmb{d}|\}$.
In particular, $\quot{\phi}{\pmb{d}} \in \Phi_k$.
\end{enumerate}
\end{lemma}

Also, the following Definition \ref{def:phi-star} and Lemma \ref{lem:phi-represent} are useful for the representation of a mapping $\phi \in \Phi$.
See Appendix \ref{subsec:proof-phi-represent} for the proof of Lemma \ref{lem:phi-represent}.

\begin{definition}
\label{def:phi-star}
For an integer $k \geq 0$ and $\phi \in \Phi_k$,
we define a mapping $\phi^{\ast} : \mathcal{C}^{\ast} \to \mathcal{C}$ as
\begin{equation}
\label{eq:phi-star}
\phi^{\ast}(\pmb{b}) \coloneqq \quot{\phi}{\pmb{b}}(0)
\end{equation}
for $\pmb{b} \in \mathcal{C}^{\ast}$.
\end{definition}

\begin{lemma}
\label{lem:phi-represent}
For any integer $k \geq 0$ and $\phi \in \Phi_k$, the following statements (i) and (ii) hold,
where $\oplus$ denotes the addition modulo $2$, that is, $0 \oplus 0 = 1 \oplus 1 = 0, 0 \oplus 1 = 1 \oplus 0 = 1$.
\begin{enumerate}[(i)]
\item For any $\pmb{b} \in \mathcal{C}^{\ast}$, we have
\begin{equation}
\label{eq:uheypssietwo}
\phi(\pmb{b}) = (b_1\oplus\phi^{\ast}([\pmb{b}]_0))(b_2\oplus\phi^{\ast}([\pmb{b}]_1))\ldots (b_{|\pmb{b}|}\oplus\phi^{\ast}([\pmb{b}]_{|\pmb{b}|-1})),
\end{equation}
where the right-hand side represents the sequence of length of $|\pmb{b}|$ obtained by concatenating the following $|\pmb{b}|$ bits: $(b_1\oplus\phi^{\ast}([\pmb{b}]_0)), (b_2\oplus\phi^{\ast}([\pmb{b}]_1)), \ldots, (b_{|\pmb{b}|}\oplus\phi^{\ast}([\pmb{b}]_{|\pmb{b}|-1}))$.
\item For any $\pmb{b} \in \mathcal{C}^{\geq k}$, we have $\phi^{\ast}(\pmb{b}) = 0$.
\end{enumerate}
\end{lemma}

Lemma \ref{lem:phi-represent} (i) implies that a mapping $\phi \in \Phi_k$ is completely determined if the mapping $\phi^{\ast}$ is determined.
Further, $\phi^{\ast}(\pmb{b})$ for all $\pmb{b} \in \mathcal{C}^{\geq k}$ are automatically determined to be $0$ by Lemma \ref{lem:phi-represent} (ii).
Therefore, a mapping $\phi \in \Phi_k$ is specified from the information of $\phi^{\ast}(\pmb{b})$ for all $\pmb{b} \in \mathcal{C}^{\leq k-1}$ and can be represented by a bit sequence of the length of $|\mathcal{C}^{\leq k-1}| = 2^k-1$.
In particular, the cardinality of $\Phi_k$ is finite for any $k \geq 0$.

Notice that for $A \in \mathscr{C}^k$ and $\phi \in \Phi_k$,
we have $\phi(A) = \{\phi(\pmb{b}) : \pmb{b} \in A\} \in \mathscr{C}^k$ by Definition \ref{def:Phi} (a).
We now define a binary relation $\sim_k$ of $\mathscr{C}^k$ as follows:
for any $A,  B \in \mathscr{C}^k$, it holds that $A \sim_k B$ if and only if there exists $\phi \in \Phi_k$ such that $\phi(A) =  B$.
Then the binary relation $\sim_k$ is an equivalence relation of $\mathscr{C}^k$: the reflexivity (resp.~symmetry, transitivity) is confirmed from the following Lemma \ref{lem:Phi-group} (i) (resp.~(ii), (iii)), which is shown directly from Definition \ref{def:Phi}.

\begin{lemma}
\label{lem:Phi-group}
For any integer $k \geq 0$, the following statements (i)--(iii) hold.
\begin{enumerate}[(i)]
\item The identity mapping $\mathrm{id}_{\mathcal{C}^{\ast}} \colon \mathcal{C}^{\ast} \to \mathcal{C}^{\ast}$ is in $\Phi_k$.
\item For any $\phi \in \Phi_k$, the inverse mapping $\phi^{-1}$ of $\phi$ is in $\Phi_k$.
\item For any $\phi, \phi' \in \Phi_k$, the composite $\phi \circ \phi'$ is in $\Phi_k$.
\end{enumerate}
\end{lemma}

Let $\mathscr{C}^k / {\sim_k}$ denote the set of all the equivalence classes of $\mathscr{C}^k$ by $\sim_k$.
For $A \in \mathscr{C}$, let $\langle  A \rangle$ denote the equivalence class that $A$ belongs to, and for $\mathcal{I} \subseteq \mathscr{C}^k$, we define
 \begin{equation}
 \label{eq:1k1d4tzz0kr5}
 \langle \mathcal{I} \rangle \coloneqq \bigcup_{A \in \mathcal{I}} \langle A \rangle.
 \end{equation}

 \begin{example}
 The set $\mathscr{C}^2 / {\sim}_2$ consists of the following $6$ equivalence classes: 
 \begin{itemize}
 \item $\{\emptyset\}$;
 \item $\{\{00\}, \{01\}, \{10\}, \{11\}\}$;
 \item $\{\{00, 01\}, \{10, 11\}\}$;
 \item $\{\{00, 10\}, \{00, 11\}, \{01, 10\}, \{01, 11\}\}$;
  \item $\{\{00, 01, 10\}, \{00, 01, 11\}, \{00, 10, 11\}, \{01, 10, 11\}\}$;
 \item $\{\{00, 01, 10, 11\}\}$.
 \end{itemize}
 \end{example}

The number of equivalence classes is given by a recurrence relation as the following Theorem \ref{thm:orbit-num}.
See Appendix \ref{subsec:proof-orbit-num} for the proof of Theorem \ref{thm:orbit-num}.

\begin{theorem}
\label{thm:orbit-num}
Let $a_k \coloneqq |\mathscr{C}^k / {\sim_k}|$ for each integer $k \geq 0$.
Then we have
\begin{equation*}
a_k =
\begin{cases}
2 &\text{if}\,\, k = 0,\\
a_{k-1}(a_{k-1}+1)/2 &\text{if}\,\, k \geq 1
\end{cases}
\end{equation*}
for any integer $k \geq 0$.
\end{theorem}

Table \ref{tab:orbit-num} shows the upper bound $|\mathscr{C}^k| = 2^{(2^k)}$ given by Lemma \ref{thm:differ} and 
the number $a_k$ of the equivalence classes for $0 \leq k \leq 7$.
The order of growth is given as $a_k \approx 2 \cdot c^{(2^k)}$ for sufficiently large $k$, where $c \approx 1.345768$ \cite{OEIS}.

  \begin{table*}
   \caption{The upper bound  $|\mathscr{C}^k| = 2^{(2^k)}$ and the number of equivalence classes $a_k$ for $0 \leq k \leq 7$ }
  \label{tab:orbit-num}
  \centering
\begin{tabular}{c|rrrrrrrr}
 \multicolumn{1}{c|}{$k$} & \multicolumn{1}{c}{0} & \multicolumn{1}{c}{1} & \multicolumn{1}{c}{2} & \multicolumn{1}{c}{3} & \multicolumn{1}{c}{4} & \multicolumn{1}{c}{5} & \multicolumn{1}{c}{6} & \multicolumn{1}{c}{7}  \\
\hline
$|\mathscr{C}^k|$ & 2 & 4 & 16 & 256 & 65536 & 4294967296 & 18446744073709551616 & 340282366920938463463374607431768211456\\
$a_k$ & 2 & 3 & 6 & 21 & 231 & 26796 & 359026206  & 64449908476890321
 \end{tabular}
  \end{table*}

\subsection{Reduced-Code-Tuples}
\label{subsec:semi}

Hereinafter, we fix the length $k$ of the allowed decoding delay,
and we consider only code-tuples with at most $|\mathscr{C}^k| = 2^{(2^k)}$ code tables because of Lemma \ref{thm:differ} unless otherwise specified.
Also, let the domain $[F]$ of a code-tuple $F$ be a subset of $\mathscr{C}^k$ instead of $\{0, 1, 2, \ldots, m-1\}$ for an integer $m \geq 1$;
that is, each code table is subscripted by a set $A \in \mathscr{C}^k$.
Namely, we suppose that a code-tuple is determined by the following three components (i)--(iii): 
\begin{enumerate}[(i)]
\item the domain $\mathcal{I} \subseteq \mathscr{C}^k$,
\item $|\mathcal{I}|$ mappings $f_{A} \colon \mathcal{S} \to \mathcal{C}^{\ast}$ for $A \in \mathcal{I}$,
\item $|\mathcal{I}|$ mappings $\trans_{A} \colon \mathcal{S} \to \mathcal{I}$ for $A \in \mathcal{I}$.
\end{enumerate}
Then let $[F]$ represent the domain $\mathcal{I}$ of $F$, and let $|F|$ denote the cardinality of $[F]$ (i.e., the number of code tables).
Note that this change of indices is not essential since we merely replace the identifier of each code table from non-negative integers to sets $\mathcal{A} \in \mathscr{C}^k$.
In particular, we do not require any relationship between the mappings $f_A$, $\trans_A$ and the index $A$:
for example, there is no requirement that $\PREF^k_{F, A} = A$ holds between the mappings $f_A, \trans_A$ and the index $A$.

\begin{example}
Table \ref{tab:code-tuple-kai} shows an example of a code-tuple with the domain $[F] \subseteq \mathscr{C}^k$,
where $k = 2$, $[F] = \{\emptyset, \{00, 11\}, \{01, 10, 11\}\} \subseteq \mathscr{C}^k$ is chosen.
Note that the code-tuple in Table \ref{tab:code-tuple-kai} is obtained from one in Table \ref{tab:code-tuple} by merely replacing the indices $0, 1, 2$ with $\{01, 10, 11\}, \emptyset, \{00, 11\}$, respectively.
\end{example}

\begin{table}
\caption{An example $F = (f, \trans)$ of a code-tuple with the domain in $\mathscr{C}^k$}
\label{tab:code-tuple-kai}
\centering
\begin{tabular}{c | lclclc}
\hline
$s \in \mathcal{S}$ & $f_{\{01, 10, 11\}}$ & $\trans_{\{01, 10, 11\}}$ & $f_{\emptyset}$ & $\trans_{\emptyset}$ & $f_{\{00, 11\}}$ & $\trans_{\{00, 11\}}$\\
\hline
a & 01 & $\{01, 10, 11\}$ & 00 & $\emptyset$ & 1100 & $\emptyset$\\
b & 10 & $\emptyset$ & $\lambda$ & $\{01, 10, 11\}$ & 1110 & $\{00, 11\}$\\
c & 0100 & $\{01, 10, 11\}$ & 00111 & $\emptyset$ & 111000 & $\{00, 11\}$\\
d & 01 & $\{00, 11\}$ & 00111 & $\{00, 11\}$ & 110 & $\{00, 11\}$\\
\hline
\end{tabular}
\end{table}

Further, we choose and fix a representative from each equivalence class of $\mathscr{C}^k / {\sim_k}$ and define a set $\tilde{\mathscr{C}}^k$ as the set of all the representatives.
For $A \in \mathscr{C}^k$, let $\bar{A}$ denote the representative of $\langle A \rangle$.
Notice that 
\begin{equation}
\label{eq:9k4gm8wmk5qb}
|A| = |\bar{A}|
\end{equation}
by (\ref{eq:mpd5l92t1cji}).

While code tables of a code-tuple are associated elements of $\mathscr{C}^k$,  
a \emph{reduced code-tuple} is a tuple of code tables associated with elements of $\tilde{\mathscr{C}}^k$.

\begin{definition}
\label{def:semi}
A \emph{reduced code-tuple (RCT)} $\tilde{F}=(\mathcal{I}, (\tilde{f}_A)_{A \in \mathcal{I}}, (\tilde{\trans}_A)_{A \in \mathcal{I}})$ is a tuple of the following three components (i)--(iii):
\begin{enumerate}[(i)]
\item the domain $\mathcal{I} \subseteq \tilde{\mathscr{C}^k}$;
\item $|\mathcal{I}|$ mappings $\tilde{f}_{A} \colon \mathcal{S} \to \mathcal{C}^{\ast}$ for $A \in \mathcal{I}$;
\item $|\mathcal{I}|$ mappings $\tilde{\trans}_{A} \colon \mathcal{S} \to \langle \mathcal{I} \rangle$ for $A \in \mathcal{I}$.
\end{enumerate}
We define $\tilde{\mathscr{F}}$ as the set of all RCTs.
\end{definition}

For $\tilde{F}=(\mathcal{I}, (\tilde{f}_A)_{A \in \mathcal{I}}, (\tilde{\trans}_A)_{A \in \mathcal{I}}) \in \tilde{\mathscr{F}}$, let $[\tilde{F}]$ represent the domain $\mathcal{I}$ (i.e., $[\tilde{F}] \coloneqq \mathcal{I}$), and let $|\tilde{F}|$ denote the cardinality of $[\tilde{F}]$.
We also write as $\tilde{F}=(\tilde{f}, \tilde{\trans})$ or $\tilde{F}$ more simply.

\begin{example}
Table \ref{tab:reduced} shows an example of an RCT $\tilde{F} = (\tilde{f}, \tilde{\trans})$ with domain $[\tilde{F}] = \{\{00\},  \{00, 10\}, \allowbreak \{00, 01, 10, 11\}\} \subseteq \tilde{\mathscr{C}}^k$, where $k = 2$ and $\tilde{\mathscr{C}}^k = \{ \emptyset, \{00\}, \{00, 01\}, \{00, 10\}, \{00, 01, 10\}, \allowbreak \{00, 01, 10, 11\}\}$.
\end{example}

\begin{table}
\caption{An example of an RCT}
\label{tab:reduced}
\centering
\begin{tabular}{c | lclclc}
\hline
$s \in \mathcal{S}$ & $\tilde{f}_{\{00\}}$ & $\tilde{\trans}_{\{00\}}$ & $\tilde{f}_{\{00, 10\}}$ & $\tilde{\trans}_{\{00, 10\}}$ & $\tilde{f}_{\{00, 01, 10, 11\}}$ & $\tilde{\trans}_{\{00, 01, 10, 11\}}$\\
\hline
a & 001 & $\{01, 10\}$ & 1 & $\{01\}$ & 1 & $\{01, 10\}$\\
b & 000 & \{00\} & 1001 & $\{00, 01, 10, 11\}$ & 1 & $\{11\}$\\
c & 00 & \{01\} & $\lambda$ & $\{00\}$ & 100 & $\{00, 01, 10, 11\}$\\
d & 001 & $\{00, 11\}$ & 1000 & $\{00, 01, 10, 11\}$ & 0 & $\{00, 01, 10, 11\}$\\
\hline
\end{tabular}
\end{table}

\begin{remark}
\label{rem:semi}
The only difference between the definitions of code-tuple and RCT is the target set of the mappings $\tilde{\trans}_A$ in Definition \ref{def:semi} (iii).
Namely, the mappings $\trans_A$ of a code-tuple $F = (f, \trans)$ are allowed to map $s \in \mathcal{S}$ only into the domain $[F]$ of $F$,
whereas the mappings $\tilde{\trans}_A$ of an RCT $\tilde{F} = (\tilde{f}, \tilde{\trans})$ can map $s \in \mathcal{S}$ into $\langle [F] \rangle$ not necessarily in $[F]$ of $\tilde{F}$.
Therefore, in general, an RCT is not a code-tuple and cannot be used for coding as is.
\end{remark}

As counterparts of $\PREF^k_{F, A}$ and $\bar{\PREF}^k_{F, A}(\pmb{b})$,
we define $\PREF^k_{\tilde{F}, A}$ and $\bar{\PREF}^k_{\tilde{F}, A}(\pmb{b})$ as follows.

\begin{definition}
\label{def:pref-tilde}
For $\tilde{F} \in \tilde{\mathscr{F}}$ and $A \in [\tilde{F}]$, we introduce the following definitions (i) and (ii).
\begin{enumerate}[(i)]
\item We define $\PREF^k_{\tilde{F}, A} \in \mathscr{C}^k$ as
\begin{equation}
\label{eq:x425gfz5autj}
\PREF^k_{\tilde{F}, A} \coloneqq \bigcup_{s \in \mathcal{S}} \left[ \tilde{f}_{A}(s) \tilde{\trans}_{A}(s) \right]_k.
\end{equation}
\item For $\pmb{b} \in \mathcal{C}^{\ast}$, we define $\bar{\PREF}^k_{\tilde{F}, A}(\pmb{b}) \in \mathscr{C}^k$ as
\begin{equation}
\label{eq:hh4v6kmbk61n}
\bar{\PREF}^k_{\tilde{F}, A}(\pmb{b}) \coloneqq \bigcup_{\substack{s \in \mathcal{S},\\ \tilde{f}_{A}(s) \succ \pmb{b}}} \left[ \pmb{b}^{-1}\tilde{f}_{A}(s) \tilde{\trans}_{A}(s) \right]_k.
\end{equation}
\end{enumerate}
\end{definition}
Note that $\tilde{\trans}(s) \in \langle[\tilde{F}]\rangle$ in (\ref{eq:x425gfz5autj}) and (\ref{eq:hh4v6kmbk61n}) is an element of $\mathscr{C}^k$ (i.e., a subset of $\mathcal{C}^k$), and thus $\left[ \tilde{f}_{A}(s) \tilde{\trans}_{A}(s) \right]_k$ and $ \left[ \pmb{b}^{-1}\tilde{f}_{A}(s) \tilde{\trans}_{A}(s) \right]_k$ are also elements of $\mathscr{C}^k$.

\begin{table}
\caption{The sets $\PREF^2_{\tilde{F}, A}, \bar{\PREF}^2_{\tilde{F}, A}(\tilde{f}_A(\mathrm{a})), \bar{\PREF}^2_{\tilde{F}, A}(\tilde{f}_A(\mathrm{b})), \bar{\PREF}^2_{\tilde{F}, A}(\tilde{f}_A(\mathrm{c})), \bar{\PREF}^2_{\tilde{F}, A}(\tilde{f}_A(\mathrm{d}))$ of $\tilde{F}$ in Table \ref{tab:reduced} for $A \in [\tilde{F}]$}
\centering
\label{tab:reduced-pref}
\begin{tabular}{c | ccc}
\hline
$A \in [\tilde{F}]$& $\{00\}$ & $\{00, 10\}$ & $\{00, 01, 10, 11\}$\\
\hline
$\PREF^2_{\tilde{F}, A}$ & $\{00\}$ & $\{00, 10\}$ & $\{00, 01, 10, 11\}$ \\
$\bar{\PREF}^2_{\tilde{F}, A}(\tilde{f}_A(\mathrm{a}))$ & $\emptyset$ & $\{00\}$ & $\{00\}$\\
$\bar{\PREF}^2_{\tilde{F}, A}(\tilde{f}_A(\mathrm{b}))$ & $\emptyset$ & $\emptyset$ & $\{00\}$ \\
$\bar{\PREF}^2_{\tilde{F}, A}(\tilde{f}_A(\mathrm{c}))$ & $\{00, 10, 11\}$ & $\{10\}$ & $\emptyset$ \\
$\bar{\PREF}^2_{\tilde{F}, A}(\tilde{f}_A(\mathrm{d}))$ & $\emptyset$ & $\emptyset$ & $\emptyset$ \\
\hline
\end{tabular}
\end{table}

\begin{example}
\label{ex:reduced-pref}
Table \ref{tab:reduced-pref} shows $\PREF^2_{\tilde{F}, A}$ and $\bar{\PREF}^2_{\tilde{F}, A}(\tilde{f}_A(s))$ of $\tilde{F}$ in Table \ref{tab:reduced} for $A \in [F]$ and $s \in \mathcal{S}$.
For example, $\PREF^2_{\tilde{F}, \{00, 10\}}$ is obtained as
\begin{eqnarray*}
\lefteqn{\PREF^2_{\tilde{F}, \{00, 10\}}}\\
&=& \bigcup_{s \in \mathcal{S}} \left[ \tilde{f}_{\{00, 10\}}(s) \tilde{\trans}_{\{00, 10\}}(s) \right]_2 \\
&=& \left[ \tilde{f}_{\{00, 10\}}(\mathrm{a}) \tilde{\trans}_{\{00, 10\}}(\mathrm{a}) \right]_2 \cup \left[ \tilde{f}_{\{00, 10\}}(\mathrm{b}) \tilde{\trans}_{\{00, 10\}}(\mathrm{b}) \right]_2 \cup \left[ \tilde{f}_{\{00, 10\}}(\mathrm{c}) \tilde{\trans}_{\{00, 10\}}(\mathrm{c}) \right]_2 \cup \left[ \tilde{f}_{\{00, 10\}}(\mathrm{d}) \tilde{\trans}_{\{00, 10\}}(\mathrm{d}) \right]_2\\
&=& \left[ 1 \{01\} \right]_2 \cup \left[ 1001 \{00, 01, 10, 11\} \right]_2 \cup \left[ \lambda \{00\} \right]_2 \cup \left[ 1000 \{00, 01, 10, 11\} \right]_2\\
&=& \left[\{101\} \right]_2 \cup \left[\{100100, 100101, 100110, 100111\} \right]_2 \cup \left[\{00\} \right]_2 \cup \left[\{100000, 100001, 100010, 100011\} \right]_2\\
&=& \{10\} \cup \{10\} \cup \{00\} \cup \{10\}\\
&=& \{00, 10\}.
\end{eqnarray*}
Also, $\bar{\PREF}^2_{\tilde{F}, \{00\}}(\tilde{f}_{\{00\}}(\mathrm{c}))$ is obtained as
\begin{eqnarray*}
\lefteqn{\bar{\PREF}^2_{\tilde{F}, \{00\}}(\tilde{f}_{\{00\}}(\mathrm{c}))}\\
&=& \bar{\PREF}^2_{\tilde{F}, \{00\}}(00)\\
&=& \bigcup_{\substack{s \in \mathcal{S},\\ \tilde{f}_{\{00\}}(s) \succ 00}} \left[ (00)^{-1}\tilde{f}_{\{00\}}(s) \tilde{\trans}_{\{00\}}(s) \right]_2\\
&=& \bigcup_{s \in \{\mathrm{a, b, d}\}} \left[ (00)^{-1}\tilde{f}_{\{00\}}(s) \tilde{\trans}_{\{00\}}(s) \right]_2\\
&=& \left[ (00)^{-1}\tilde{f}_{\{00\}}(\mathrm{a}) \tilde{\trans}_{\{00\}}(\mathrm{a}) \right]_2 \cup 
\left[ (00)^{-1}\tilde{f}_{\{00\}}(\mathrm{b}) \tilde{\trans}_{\{00\}}(\mathrm{b}) \right]_2 \cup \left[ (00)^{-1}\tilde{f}_{\{00\}}(\mathrm{d}) \tilde{\trans}_{\{00\}}(\mathrm{d}) \right]_2\\
&=& \left[ (00)^{-1} 001 \{01, 10\} \right]_2 \cup \left[ (00)^{-1} 000 \{00\} \right]_2 \cup \left[ (00)^{-1} 001 \{00, 11\} \right]_2\\
&=& \left[\{101, 110\}\right]_2 \cup \left[\{000\}\right]_2 \cup \left[\{100, 111\}\right]_2\\
&=& \{10, 11\} \cup \{00\} \cup \{10, 11\}\\
&=& \{00, 10, 11\}.
\end{eqnarray*}
\end{example}

As mentioned in Remark \ref{rem:semi}, an RCT $\tilde{F}=(\tilde{f}, \tilde{\trans})$ cannot be used directly 
for coding since $\tilde{\trans}_A(s)$ is not necessarily in $[\tilde{F}]$.
A natural way to recast an RCT as a code-tuple is to replace each $\tilde{\trans}_A(s)$ with $\overline{\tilde{\trans}_A(s)}$,
so that it is included in $[\tilde{F}]$.
This idea yields the following definition of the \emph{direct realization} of $\tilde{F}$.

\begin{definition}
\label{def:direct}
For $\tilde{F}=(\tilde{f}, \tilde{\trans}) \in \tilde{\mathscr{F}}$, the \emph{direct realization} of $\tilde{F}$ is the code-tuple $F=(f, \trans)$
with the domain $[F] \coloneqq [\tilde{F}]$ defined by 
\begin{eqnarray}
f_{A}(s) &\coloneqq& \tilde{f}_{A}(s), \label{eq:z3m658av3wl7}\\
\trans_{A}(s) &\coloneqq& \overline{\tilde{\trans}_{A}(s)} \label{eq:qwshe4o3q759}
\end{eqnarray}
for $A \in [F] = [\tilde{F}]$ and $s \in \mathcal{S}$.
Note that since $\tilde{\trans}_A(s) \in \langle [\tilde{F}] \rangle$ by Definition \ref{def:semi} (iii),
it holds that $\trans_A(s) = \overline{\tilde{\trans}_A(s)} \in [\tilde{F}]$ in (\ref{eq:qwshe4o3q759});
and thus Definition \ref{def:direct} defines a code-tuple indeed.
\end{definition}

\begin{example}
Table \ref{tab:direct} shows the direct realization of $\tilde{F}$ in Table \ref{tab:reduced}.
\end{example}

\begin{table}
\caption{The direct realization of $\tilde{F}$}
\label{tab:direct}
\centering
\begin{tabular}{c | lclclc}
\hline
$s \in \mathcal{S}$ & $f_{\{00\}}$ & $\trans_{\{00\}}$ & $f_{\{00, 10\}}$ & $\trans_{\{00, 10\}}$ & $f_{\{00, 01, 10, 11\}}$ & $\trans_{\{00, 01, 10, 11\}}$\\
\hline
a & 001 & $\{00, 10\}$ & 1 & $\{00\}$ & 1 & $\{00, 10\}$\\
b & 000 & \{00\} & 1001 & $\{00, 01, 10, 11\}$ & 1 & $\{00\}$\\
c & 00 & \{00\} & $\lambda$ & $\{00\}$ & 100 & $\{00, 01, 10, 11\}$\\
d & 001 & $\{00, 10\}$ & 1000 & $\{00, 01, 10, 11\}$ & 0 & $\{00, 01, 10, 11\}$\\
\hline
\end{tabular}
\end{table}

As counterparts of $\mathscr{F}_{\reg}, \mathscr{F}_{\ext}, \mathscr{F}_{k\hdec},$ and $\mathscr{F}_{k\hopt}$,
we define $\tilde{\mathscr{F}}_{\comp}, \tilde{\mathscr{F}}_{\reg}, \tilde{\mathscr{F}}_{\ext}, \allowbreak \tilde{\mathscr{F}}_{k\hdec},$ and $\tilde{\mathscr{F}}_{k\hopt}$ of RCTs as follows.

\begin{definition}
\label{def:semi-class}
We introduce the following definitions (i)--(v).
\begin{enumerate}[(i)]
\item An RCT $\tilde{F}$ is said to be \emph{compliant} if 
for any $A \in [\tilde{F}]$, it holds that $\PREF^k_{\tilde{F}, A} =  A$.
We define $\tilde{\mathscr{F}}_{\comp}$ as the set of all compliant RCTs.

\item An RCT $\tilde{F}$ is said to be \emph{regular} if the direct realization of $\tilde{F}$ is regular.
We define $\tilde{\mathscr{F}}_{\reg}$ as the set of all regular RCTs.

For $\tilde{F} \in \tilde{\mathscr{F}}_{\reg}$, we define the \emph{average codeword length $\tilde{L}(\tilde{F})$ of $\tilde{F}$} as
the average codeword length of the direct realization of $\tilde{F}$,
that is, $\tilde{L}(\tilde{F}) \coloneqq L(F)$, where $F$ is the direct realization of $\tilde{F}$.
\item An RCT $\tilde{F}$ is said to be \emph{extendable} if $[\tilde{F}] \not\owns \emptyset$.
We define $\tilde{\mathscr{F}}_{\ext}$ as the set of all extendable RCTs.
\item An RCT $\tilde{F}$ is said to be \emph{$k$-bit delay decodable} if the following conditions (a) and (b) hold.
\begin{enumerate}[(a)]
\item For any $A \in [\tilde{F}]$ and $s \in \mathcal{S}$, it holds that
$\tilde{\trans}_ A(s) \cap \bar{\PREF}^k_{\tilde{F}, A}(\tilde{f}_{A}(s)) = \emptyset$.
\item For any $A \in [\tilde{F}]$ and $s, s' \in \mathcal{S}$, if $s \neq s'$ and $\tilde{f}_{A}(s) = \tilde{f}_{A}(s')$,
then $\tilde{\trans}_ A(s) \cap \tilde{\trans}_ A(s')  = \emptyset$.
\end{enumerate}
We define $\tilde{\mathscr{F}}_{k\hdec}$ as the set of all $k$-bit delay decodable RCTs.
\item An RCT $\tilde{F}$ is said to be \emph{$k$-bit delay optimal} if for any $\tilde{F}' \in \tilde{\mathscr{F}}_{\comp} \cap \tilde{\mathscr{F}}_{\reg} \cap \tilde{\mathscr{F}}_{\ext} \cap \tilde{\mathscr{F}}_{k\hdec}$, it holds that $\tilde{L}(\tilde{F}) \leq \tilde{L}(\tilde{F}')$.
We define $\tilde{\mathscr{F}}_{k\hopt}$ as the set of all $k$-bit delay optimal RCTs.
\end{enumerate}
\end{definition}

\begin{remark}
The definitions of $\tilde{\mathscr{F}}, \tilde{\mathscr{F}}_{\comp}, \tilde{\mathscr{F}}_{\reg}, \tilde{\mathscr{F}}_{\ext}, \allowbreak \tilde{\mathscr{F}}_{k\hdec},$ and $\tilde{\mathscr{F}}_{k\hopt}$ depend on $k$, and we are discussing on a fixed $k$.
\end{remark}

\begin{example}
We confirm that $\tilde{F} = (\tilde{f}, \tilde{\trans})$ in Table \ref{tab:reduced} satisfies $\tilde{F} \in \tilde{\mathscr{F}}_{\comp} \cap
\tilde{\mathscr{F}}_{\reg} \cap \tilde{\mathscr{F}}_{\ext} \cap \tilde{\mathscr{F}}_{k\hdec}$ as follows.

\begin{itemize}
\item We have $\tilde{F} \in \tilde{\mathscr{F}}_{\comp}$ because
\begin{equation*}
\PREF^2_{\tilde{F}, \{00\}} = \{00\}, \quad \PREF^2_{\tilde{F}, \{00, 10\}} = \{00, 10\}, \quad \PREF^2_{\tilde{F}, \{00, 01, 10, 11\}} = \{00, 01, 10, 11\}
\end{equation*}
as shown in Table \ref{tab:reduced-pref}.

\item We have $\tilde{F} \in \tilde{\mathscr{F}}_{\reg}$ because the direct realization $F$ of $\tilde{F}$, shown in Table \ref{tab:direct}, is regular by $\kernel_F = \{\{00\}, \{00, 10\}, \allowbreak \{00, 01, 10, 11\}\} \neq \emptyset$.

\item We have $\tilde{F} \in \tilde{\mathscr{F}}_{\ext}$ because $[\tilde{F}] = \{\{00\}, \{00, 10\}, \{00, 01, 10, 11\}\} \not \owns \emptyset$.

\item We have $\tilde{F} \in \tilde{\mathscr{F}}_{k\hdec}$ by Table \ref{tab:reduced-pref} because
Definition \ref{def:k-bitdelay} (a) is satisfied as
\begin{itemize}
\item $\tilde{\trans}_{\{00\}}(\mathrm{a}) \cap \bar{\PREF}^2_{\tilde{F}, \{00\}}(\tilde{f}_{\{00\}}(\mathrm{a})) = \{01, 10\} \cap \emptyset = \emptyset$,
\item $\tilde{\trans}_{\{00\}}(\mathrm{b}) \cap \bar{\PREF}^2_{\tilde{F}, \{00\}}(\tilde{f}_{\{00\}}(\mathrm{b})) = \{00\} \cap \emptyset = \emptyset$,
\item $\tilde{\trans}_{\{00\}}(\mathrm{c}) \cap \bar{\PREF}^2_{\tilde{F}, \{00\}}(\tilde{f}_{\{00\}}(\mathrm{c})) = \{01\} \cap \{00, 10, 11\} = \emptyset$,
\item $\tilde{\trans}_{\{00\}}(\mathrm{d}) \cap \bar{\PREF}^2_{\tilde{F}, \{00\}}(\tilde{f}_{\{00\}}(\mathrm{d})) = \{00, 11\} \cap \emptyset = \emptyset$,
\item $\tilde{\trans}_{\{00, 10\}}(\mathrm{a}) \cap \bar{\PREF}^2_{\tilde{F}, \{00, 10\}}(\tilde{f}_{\{00, 10\}}(\mathrm{a})) = \{01\} \cap \{00\} = \emptyset$,
\item $\tilde{\trans}_{\{00, 10\}}(\mathrm{b}) \cap \bar{\PREF}^2_{\tilde{F}, \{00, 10\}}(\tilde{f}_{\{00, 10\}}(\mathrm{b})) = \{00, 01, 10, 11\} \cap \emptyset = \emptyset$,
\item $\tilde{\trans}_{\{00, 10\}}(\mathrm{c}) \cap \bar{\PREF}^2_{\tilde{F}, \{00, 10\}}(\tilde{f}_{\{00, 10\}}(\mathrm{c})) = \{00\} \cap \{10\} = \emptyset$,
\item $\tilde{\trans}_{\{00, 10\}}(\mathrm{d}) \cap \bar{\PREF}^2_{\tilde{F}, \{00, 10\}}(\tilde{f}_{\{00, 10\}}(\mathrm{d})) = \{00, 01, 10, 11\} \cap \emptyset = \emptyset$,
\item $\tilde{\trans}_{\{00,  01, 10, 11\}}(\mathrm{a}) \cap \bar{\PREF}^2_{\tilde{F}, \{00,  01, 10, 11\}}(\tilde{f}_{\{00, 01, 10, 11\}}(\mathrm{a})) = \{01, 10\} \cap \{00\} = \emptyset$,
\item $\tilde{\trans}_{\{00,  01, 10, 11\}}(\mathrm{b}) \cap \bar{\PREF}^2_{\tilde{F}, \{00,  01, 10, 11\}}(\tilde{f}_{\{00, 01, 10, 11\}}(\mathrm{b})) = \{11\} \cap \{00\} = \emptyset$,
\item $\tilde{\trans}_{\{00,  01, 10, 11\}}(\mathrm{c}) \cap \bar{\PREF}^2_{\tilde{F}, \{00,  01, 10, 11\}}(\tilde{f}_{\{00, 01, 10, 11\}}(\mathrm{c})) = \{00, 01, 10, 11\} \cap \emptyset = \emptyset$,
\item $\tilde{\trans}_{\{00,  01, 10, 11\}}(\mathrm{d}) \cap \bar{\PREF}^2_{\tilde{F}, \{00,  01, 10, 11\}}(\tilde{f}_{\{00, 01, 10, 11\}}(\mathrm{d})) = \{00, 01, 10, 11\} \cap \emptyset = \emptyset$,
\end{itemize}
and Definition \ref{def:k-bitdelay} (b) is satisfied as
\begin{itemize}
\item $\tilde{\trans}_{\{00\}}(\mathrm{a}) \cap \tilde{\trans}_{\{00\}}(\mathrm{d}) = \{01, 10\} \cap \{00, 11\} = \emptyset$,
\item $\tilde{\trans}_{\{00, 01, 10, 11\}}(\mathrm{a}) \cap \tilde{\trans}_{\{00, 01, 10, 11\}}(\mathrm{b}) = \{01, 10\} \cap \{11\} = \emptyset$.
\end{itemize}
\end{itemize}
\end{example}

The following theorem is one of our main results.

\begin{theorem}
\label{thm:equiv}
The following two statements (i) and (ii) hold.
\begin{enumerate}[(i)]
\item For any $F \in \mathscr{F}_{\reg} \cap \mathscr{F}_{\ext} \cap \mathscr{F}_{k\hdec}$, there exists $\tilde{F} \in \tilde{\mathscr{F}}_{\comp} \cap \tilde{\mathscr{F}}_{\reg} \cap \tilde{\mathscr{F}}_{\ext} \cap \tilde{\mathscr{F}}_{k\hdec}$ such that $\tilde{L}(\tilde{F}) \leq L(F)$.
\item For any $\tilde{F} \in \tilde{\mathscr{F}}_{\comp} \cap \tilde{\mathscr{F}}_{\reg} \cap \tilde{\mathscr{F}}_{\ext} \cap \tilde{\mathscr{F}}_{k\hdec}$, there exists $F \in \mathscr{F}_{\reg} \cap \mathscr{F}_{\ext} \cap \mathscr{F}_{k\hdec}$ such that $L(F) = \tilde{L}(\tilde{F})$.
\end{enumerate}
\end{theorem}

Theorem \ref{thm:equiv} implies that for any $\tilde{F} \in \tilde{\mathscr{F}}_{k\hopt}$, there exists $F \in \mathscr{F}_{k\hopt}$ and vice versa.
Therefore, the problem of finding $F \in \mathscr{F}_{k\hopt}$ can be reduced to the problem of finding $\tilde{F} \in \tilde{\mathscr{F}}_{k\hopt}$,
in which the number of code tables to be considered is extremely smaller as shown in Table \ref{tab:orbit-num}.

We defer the proof of Theorem \ref{thm:equiv} (i) to Appendix \ref{subsec:proof-equiv1},
and we prove Theorem \ref{thm:equiv} (ii) in the next subsection by constructing a concrete code-tuple $F \in \mathscr{F}_{\reg} \cap \mathscr{F}_{\ext} \cap \mathscr{F}_{k\hdec}$ such that $L(F) = \tilde{L}(\tilde{F})$ for a given $\tilde{F} \in \tilde{\mathscr{F}}_{\comp} \cap \tilde{\mathscr{F}}_{\reg} \cap \tilde{\mathscr{F}}_{\ext} \cap \tilde{\mathscr{F}}_{k\hdec}$.

\subsection{The Coding Procedure with a Reduced-Code-Tuple}
\label{subsec:expand}

Assume that we are given $\tilde{F} = (\tilde{f}, \tilde{\trans}) \in \tilde{\mathscr{F}}_{\comp} \cap \tilde{\mathscr{F}}_{\reg} \cap \tilde{\mathscr{F}}_{\ext} \cap \tilde{\mathscr{F}}_{k\hdec}$.
For every $A \in [\tilde{F}]$ and $s \in \mathcal{S}$, we choose and fix a mapping $\psi_{A, s} \in \Phi_k$ such that
\begin{equation}
\label{eq:ks8m7ci0bnrg}
\psi_{A, s}(\overline{\tilde{\trans}_A(s)}) = \tilde{\trans}_A(s)
\end{equation} 
and suppose that the information of $\tilde{\trans}_A(s)$ is given as the pair $( \overline{\tilde{\trans}_A(s)}, \psi_{A, s} ) \in [\tilde{F}] \times \Phi_k$.
Note that $\psi_{A, s} \in \Phi_k$ satisfying (\ref{eq:ks8m7ci0bnrg}) is not unique, 
and we fix arbitrarily chosen one for each $A \in [\tilde{F}]$ and $s \in \mathcal{S}$.

We define a code-tuple $\hat{F}=(\hat{f}, \hat{\trans})$ as follows:
the domain of $\tilde{F}$ is defined as
\begin{eqnarray}
\label{eq:079ii8hhglqh}
[\hat{F}] &\coloneqq& [\tilde{F}] \times \Phi_k, 
\end{eqnarray}
that is, each index is a pair $(A, \phi) \in [\tilde{F}] \times \Phi_k$;
the mappings $f_{(A, \phi)}$ and $\trans_{(A, \phi)}$ are defined as
\begin{eqnarray}
\hat{f}_{(A, \phi)}(s) &\coloneqq& \phi(\tilde{f}_{A}(s)), \label{eq:atqlr6s1yeoz}\\
\hat{\trans}_{(A, \phi)}(s) &\coloneqq& \left( \overline{\tilde{\trans}_A(s)}, \quot{\phi}{\tilde{f}_{A}(s)} \circ \psi_{A, s}\right) \label{eq:1aj44hktwb5i}
\end{eqnarray}
for $(A, \phi) \in [\hat{F}] = [\tilde{F}] \times \Phi_k$ and $s \in \mathcal{S}$.
Note that (\ref{eq:079ii8hhglqh})--(\ref{eq:1aj44hktwb5i}) do indeed define a code-tuple:
since $[\tilde{F}]$ and $\Phi_k$ are finite sets, the domain $[\hat{F}]$ is also a finite set;
since $\tilde{\trans}_A(s) \in \langle [\tilde{F}] \rangle$ by Definition \ref{def:semi} (iii), it holds indeed that $\overline{\tilde{\trans}_A(s)} \in [\tilde{F}]$ in (\ref{eq:1aj44hktwb5i});
the mapping $\quot{\phi}{\tilde{f}_{A}(s)} \circ \psi_{A, s}$ of (\ref{eq:1aj44hktwb5i}) is indeed in $\Phi_k$ by $\phi, \psi_{A, s} \in \Phi_k$, Lemma \ref{lem:phi-quot} (ii) and Lemma \ref{lem:Phi-group}.

\begin{table}
\caption{The RCT $\tilde{F}$ in Table \ref{tab:reduced}, where $\tilde{\trans}_A(s)$ is given in the form $( \overline{\tilde{\trans}_A(s)}, \psi_{A, s} ) \in [\tilde{F}] \times \Phi_k$}
\label{tab:reduced-kai}
\centering
\begin{tabular}{c | lclclc}
\hline
$s \in \mathcal{S}$ & $\tilde{f}_{\{00\}}$ & $\tilde{\trans}_{\{00\}}$ & $\tilde{f}_{\{00, 10\}}$ & $\tilde{\trans}_{\{00, 10\}}$ & $\tilde{f}_{\{00, 01, 10, 11\}}$ & $\tilde{\trans}_{\{00, 01, 10, 11\}}$\\
\hline
a & 001 & $(\{00, 10\}, \phi_{010})$ & 1 & $(\{00\}, \phi_{010})$ & 1 & $(\{00, 10\}, \phi_{010})$\\
b & 000 & $(\{00\}, \phi_{000})$ & 1001 & $(\{00, 01, 10, 11\}, \phi_{000})$ & 1 & $(\{00\}, \phi_{110})$\\
c & 00 & $(\{00\}, \phi_{010})$ & $\lambda$ & $(\{00\}, \phi_{000})$ & 100 & $(\{00, 01, 10, 11\}, \phi_{000})$\\
d & 001 & $(\{00, 10\}, \phi_{001})$ & 1000 & $(\{00, 01, 10, 11\}, \phi_{000})$ & 0 & $(\{00, 01, 10, 11\}, \phi_{000})$\\
\hline
\end{tabular}
\end{table}

\begin{example}
We consider the code-tuple $\hat{F} = (\hat{f}, \hat{\trans})$ defined by (\ref{eq:079ii8hhglqh})--(\ref{eq:1aj44hktwb5i}) for $\tilde{F}$ in Table \ref{tab:reduced}.
Suppose that $\tilde{\trans}_A(s)$ of $\tilde{F}$ is given in the form of the pair $( \overline{\tilde{\trans}_A(s)}, \psi_{A, s} ) \in [\tilde{F}] \times \Phi_2$ as Table \ref{tab:reduced-kai},
where $\phi_{b_1b_2b_3}$ for $b_1b_2b_3 \in \mathcal{C}^3$ denotes the mapping $\phi \in \Phi_2$ determined by $\phi^{\ast}(\lambda) = b_1, \phi^{\ast}(0) = b_2, \phi^{\ast}(1) = b_3$.
Then $\hat{F}$ consists of $24$ code tables and has the domain 
$[F] = [\tilde{F}] \times \Phi_2 = \{\{00\}, \{00, 10\}, \{00, 01, 10, 11\}\} 
\times \{\phi_{000}, \phi_{001}, \phi_{010}, \phi_{011}, \phi_{100}, \phi_{101}, \phi_{110}, \phi_{111}\}$.

Since $\hat{F}$ has too many mappings to display all here, we show the mappings only for $(A, \phi) = {(\{00, 10\}, \phi_{101})}$:
the mappings $\hat{f}_{(\{00, 10\}, \phi_{101})}$ and $\hat{\trans}_{(\{00, 10\}, \phi_{101})}$ are as follows.
\begin{eqnarray*}
\hat{f}_{(\{00, 10\}, \phi_{101})}(\mathrm{a})
&=& \phi_{101}(1) = 0,\\
\hat{f}_{(\{00, 10\}, \phi_{101})}(\mathrm{b})
&=& \phi_{101}(1001) = 0101,\\
\hat{f}_{(\{00, 10\}, \phi_{101})}(\mathrm{c})
&=& \phi_{101}(\lambda) = \lambda,\\
\hat{f}_{(\{00, 10\}, \phi_{101})}(\mathrm{d})
&=& \phi_{101}(1000) = 0100,
\end{eqnarray*}
\begin{eqnarray*}
\hat{\trans}_{(\{00, 10\}, \phi_{101})}(\mathrm{a})
&=& \left( \{00\}, \quot{\phi_{101}}{1} \circ \phi_{010}\right)
= \left( \{00\}, \phi_{110}\right),\\
\hat{\trans}_{(\{00, 10\}, \phi_{101})}(\mathrm{b})
&=& \left( \{00, 01, 10, 11\}, \quot{\phi_{101}}{1001} \circ \phi_{000}\right)
= \left( \{00, 01, 10, 11\}, \phi_{000}\right),\\
\hat{\trans}_{(\{00, 10\}, \phi_{101})}(\mathrm{c})
&=& \left( \{00\}, \quot{\phi_{101}}{\lambda} \circ \phi_{000}\right)
= \left( \{00\}, \phi_{101}\right),\\
\hat{\trans}_{(\{00, 10\}, \phi_{101})}(\mathrm{d})
&=& \left( \{00, 01, 10, 11\}, \quot{\phi_{101}}{1000} \circ \phi_{000}\right)
= \left( \{00, 01, 10, 11\}, \phi_{000}\right).
\end{eqnarray*}
\end{example}

Next, let $\mathcal{I}$ be a minimal non-empty closed set of $\hat{F}$, defined in Definition \ref{def:closed},
and we define a code-tuple $F=(f,\trans)$ as
\begin{eqnarray}
[F] &\coloneqq& \mathcal{I},  \label{eq:wh6g7yi11crb}
\end{eqnarray}
and for $(A, \phi) \in [F] = \mathcal{I}$ and $s \in \mathcal{S}$,
\begin{eqnarray}
f_{(A, \phi)}(s) &\coloneqq& \hat{f}_{(A, \phi)}(s) = \phi(\tilde{f}_{A}(s)), \label{eq:9tr7le1wumbs} \\
\trans_{(A, \phi)}(s) &\coloneqq& \hat{\trans}_{(A, \phi)}(s) = \left( \overline{\tilde{\trans}_A(s)}, \quot{\phi}{\tilde{f}_{A}(s)} \circ \psi_{A, s}\right). \label{eq:otmlflag72l3}
\end{eqnarray}

Then the code-tuple $F$ defined by (\ref{eq:wh6g7yi11crb})--(\ref{eq:otmlflag72l3}) is a desired code-tuple as shown as the following Lemma \ref{lem:expand}, which completes the proof of Theorem \ref{thm:equiv} (ii). 

\begin{lemma}
\label{lem:expand}
The code-tuple $F$ defined by (\ref{eq:wh6g7yi11crb})--(\ref{eq:otmlflag72l3}) satisfies $F \in \mathscr{F}_{\reg} \cap \mathscr{F}_{\ext} \cap \mathscr{F}_{k\hdec}$ and $L(F) = \tilde{L}(\tilde{F})$.
\end{lemma}
We can see $F \in \mathscr{F}_{\reg}$ of Lemma \ref{lem:expand} directly from (\ref{eq:wh6g7yi11crb}) and Lemma \ref{lem:closed}.
The proofs of $L(F) = \tilde{L}(\tilde{F})$, $F \in \mathscr{F}_{k\hdec}$, and $F \in \mathscr{F}_{\ext}$ are deferred to Appendices \ref{subsec:proof-expand1}--\ref{subsec:proof-expand3}, respectively.

By Theorem \ref{thm:equiv}, given a $k$-bit delay optimal RCT $\tilde{F}$,
we can obtain a $k$-bit delay optimal code-tuple as $F$ defined by (\ref{eq:wh6g7yi11crb})--(\ref{eq:otmlflag72l3}) and use it for coding.
In practical coding process, we can use $\hat{F}$ defined by (\ref{eq:079ii8hhglqh})--(\ref{eq:1aj44hktwb5i}) instead of $F$ defined by (\ref{eq:wh6g7yi11crb})--(\ref{eq:otmlflag72l3}):
$\hat{F}$ is generally not regular and thus $L(\hat{F})$ is not necessarily defined; however, $\hat{F}$ achieves the equivalent efficiency to the average codeword length $\tilde{L}(\tilde{F})$ of the direct realization $F' = (f', \trans')$ of $\tilde{F}$ in the following sense.

\begin{lemma}
\label{lem:expand-coding}
For any $(A, \phi) \in [\hat{F}]$ and $\pmb{x} \in \mathcal{S}^{\ast}$, we have $|\hat{f}^{\ast}_{(A, \phi)}(\pmb{x})| = |f'^{\ast}_{A}(\pmb{x})|$.
\end{lemma}
See Appendix \ref{subsec:proof-expand-coding} for the proof of Lemma \ref{lem:expand-coding}.

Further, we can perform the coding process using only $\tilde{F}$ without constructing $\hat{F}$ itself explicitly.
Algorithm \ref{alg:encoding} shows the encoding procedure using only $\tilde{F}$: we start with an arbitrarily chosen index $(A, \phi) \in [\hat{F}] = [\tilde{F}] \times \Phi_k$ and then repeat to encode a source symbol and update the index of the current code table according to (\ref{eq:qwshe4o3q759}).

 \begin{algorithm}
\caption{Encoding procedure without explicit construction of $\hat{F}$.}
\begin{algorithmic}[1]
\label{alg:encoding}
 \REQUIRE{The used RCT $\tilde{F} = (\tilde{f}, \tilde{\trans}) \in \tilde{\mathscr{F}}$, the index of the initial code table $(A, \phi) \in [\tilde{F}] \times \Phi_k$, and the source sequence $\pmb{x} \in \mathcal{S}^{\ast}$.}
 \ENSURE{The codeword sequence $\hat{f}^{\ast}_{(A, \phi)}(\pmb{x})$ of $\pmb{x}$ encoded with $\hat{F}$ corresponding to $\tilde{F}$.}
 \STATE {$\pmb{b} \gets \lambda$.}
 \FOR {$i \gets 1$ \TO $|\pmb{x}|$}
 \STATE{$\pmb{b} \gets \pmb{b}\phi(\tilde{f}_A(x_i))$.}
 \STATE{$(A, \phi) \gets \left( \overline{\tilde{\trans}_A(x_i)}, \quot{\phi}{\tilde{f}_{A}(x_i)} \circ \psi_{A, x_i}\right)$.}
\ENDFOR 
 \RETURN $\pmb{b}$
 \end{algorithmic}
\end{algorithm}

Also, Algorithm \ref{alg:decoding} shows the decoding procedure using only $\tilde{F}$:
we start with the initial index $(A, \phi) \in [\hat{F}] = [\tilde{F}] \times \Phi_k$ shared with the encoder, and
as long as there exists $s \in \mathcal{S}$ satisfying (\ref{eq:dqovdwqo208r}), we recover $s$ and update the current index $(A, \phi)$ by (\ref{eq:qwshe4o3q759}) repeatedly.
It is guaranteed that at most one $s \in \mathcal{S}$ satisfies (\ref{eq:dqovdwqo208r}) because of the $k$-bit delay decodability of $\tilde{F}$.

 \begin{algorithm}
\caption{Decoding procedure without explicit construction of $\hat{F}$.}
\begin{algorithmic}[1]
\label{alg:decoding}
 \REQUIRE{The used RCT $\tilde{F} \in \tilde{\mathscr{F}}$, the index of the initial code table $(A, \phi) \in [\tilde{F}] \times \Phi_k$, and the codeword sequence $\pmb{b} \in \mathcal{C}^{\ast}$.}
 \ENSURE{The source sequence $\pmb{x}$ obtained by decoding $\pmb{b}$ with $\hat{F}$ corresponding to $\tilde{F}$.}
 \STATE {$\pmb{x} \gets \lambda$.}
 \STATE \textbf{do}
 \STATE {\quad $\pmb{b}' \gets \phi^{-1}(\pmb{b})$.}
 \STATE {\quad Let $s \gets \mathcal{S}$ be a symbol such that
 \begin{equation}
 \label{eq:dqovdwqo208r}
 \exists \pmb{c} \in \mathcal{C}^{\ast} \,\,\mathtt{s.t.}\,\, \tilde{f}_A(s)\tilde{\trans}_{A}(s) \owns \pmb{c} \preceq \pmb{b}'.
 \end{equation}
 \quad If there exists no $s \in \mathcal{S}$ satisfying (\ref{eq:dqovdwqo208r}), then {\bf break}.}
 \STATE{\quad $\pmb{x} \gets \pmb{x}s, \,\, \pmb{b} \gets \phi(\tilde{f}_A(s))^{-1}\pmb{b}$.}
 \STATE{\quad $(A, \phi) \gets \left( \overline{\tilde{\trans}_A(s)}, \quot{\phi}{\tilde{f}_{A}(s)} \circ \psi_{A, s}\right)$.}
\STATE \textbf{loop}
\RETURN $\pmb{x}$
 \end{algorithmic}
\end{algorithm}

In the coding process, each index $A \in [\tilde{F}]$ can be represented as the bit sequence of length $2^k$ that indicates whether $\pmb{b} \in A$ or $\pmb{b} \not\in A$ for each $\pmb{b} \in \mathcal{C}^k$.
Also, as stated after Lemma \ref{lem:phi-represent}, a mapping $\phi \in \Phi_k$ can be encoded to the bit sequence of length $|\mathcal{C}^{\leq k-1}| = 2^{k}-1$ that represents the values of $\phi^{\ast}(\pmb{b})$ for each $\pmb{b} \in \mathcal{C}^{\leq k-1}$.
Lemma \ref{lem:phi-represent} and the following Lemma \ref{lem:quot-calc} can be used for the calculations in the coding process.
The proof of Lemma \ref{lem:quot-calc} is in Appendix \ref{subsec:proof-quot-calc}.

\begin{lemma}
\label{lem:quot-calc}
For any integer $k \geq 0$ , the following statements (i)--(iii) hold.
\begin{enumerate}[(i)]
\item For any $\phi \in \Phi_k$ and $\pmb{d}, \pmb{b} \in \mathcal{C}^{\ast}$, 
we have $\left(\quot{\phi}{\pmb{d}}\right)^{\ast}(\pmb{b}) = \phi^{\ast}(\pmb{db})$.

\item For any $\phi, \psi \in \Phi_k$ and $\pmb{b} \in \mathcal{C}^{\ast}$, 
we have $(\phi \circ \psi)^{\ast}(\pmb{b}) = \phi^{\ast}(\psi(\pmb{b})) \oplus \psi^{\ast}(\pmb{b})$.

\item For any $\phi \in \Phi_k$ and $\pmb{b} \in \mathcal{C}^{\ast}$,
we have $(\phi^{-1})^{\ast}(\phi(\pmb{b})) = \phi^{\ast}(\pmb{b})$.
\end{enumerate}
\end{lemma}

\begin{example}
\label{ex:reduced-encode}
We show the encoding process of $\pmb{x} = \mathrm{acdb}$ with the RCT $\tilde{F}$ in Table \ref{tab:reduced-kai}.
We choose $(\{00, 01, 10, 11\}, \phi_{000})$ as the index of the first code table.
Then the encoding process is as follows.

\begin{itemize}
\item $x_1 = \mathrm{a}$ is encoded by $\hat{f}_{(\{00, 01, 10, 11\}, \phi_{000})}$ to
\begin{equation*}
\hat{f}_{(\{00, 01, 10, 11\}, \phi_{000})}(\mathrm{a})
= \phi_{000}(\tilde{f}_{\{00, 01, 10, 11\}}(\mathrm{a}))
= \phi_{000}(1)
= 1
\end{equation*}
and the next index is 
\begin{equation*}
\hat{\trans}_{(\{00, 01, 10, 11\}, \phi_{000})}(\mathrm{a})
= \left( \{00, 10\}, \quot{\phi_{000}}{1} \circ \phi_{010}\right)
= \left( \{00, 10\}, \phi_{010}\right).
\end{equation*}

\item $x_2 = \mathrm{c}$ is encoded by $\hat{f}_{(\{00, 10\}, \phi_{010})}$ to
\begin{equation*}
\hat{f}_{(\{00, 10\}, \phi_{010})}(\mathrm{c})
= \phi_{010}(\tilde{f}_{\{00, 10\}}(\mathrm{c}))
= \phi_{010}(\lambda)
= \lambda
\end{equation*}
and the next index is 
\begin{equation*}
\hat{\trans}_{(\{00, 10\}, \phi_{010})}(\mathrm{c})
= \left( \{00\}, \quot{\phi_{010}}{\lambda} \circ \phi_{000}\right)\
= \left( \{00\}, \phi_{010}\right).
\end{equation*}

\item $x_3 = \mathrm{d}$ is encoded by $\hat{f}_{(\{00\}, \phi_{010})}$ to
\begin{equation*}
\hat{f}_{(\{00\}, \phi_{010})}(\mathrm{d})
= \phi_{010}(\tilde{f}_0(\mathrm{d}))
= \phi_{010}(001)
= 011
\end{equation*}
and the next index is 
\begin{equation*}
\hat{\trans}_{(\{00\}, \phi_{010})}(\mathrm{d})
= \left( \{00, 10\}, \quot{\phi_{010}}{001} \circ \phi_{001}\right)
= \left( \{00, 10\}, \phi_{001}\right).
\end{equation*}

\item $x_4 = \mathrm{b}$ is encoded by $\hat{f}_{(\{00, 10\}, \phi_{001})}$ to
\begin{equation*}
\hat{f}_{(\{00, 10\}, \phi_{001})}(\mathrm{b})
= \phi_{001}(\tilde{f}_{\{00, 10\}}(\mathrm{b}))
= \phi_{001}(1001)
= 1101
\end{equation*}
and the next index is 
\begin{equation*}
\hat{\trans}_{(\{00, 10\}, \phi_{001})}(\mathrm{b})
= \left( \{00, 01, 10, 11\}, \quot{\phi_{001}}{1001} \circ \phi_{000}\right)
= \left( \{00, 01, 10, 11\}, \phi_{000}\right).
\end{equation*}
\end{itemize}

Consequently, the codeword sequence is
\begin{equation*}
\hat{f}^{\ast}_{(\{00, 01, 10, 11\}, \phi_{000})}(\mathrm{acdb})
= \hat{f}_{(\{00, 01, 10, 11\}, \phi_{000})}(\mathrm{a})\hat{f}_{(\{00, 10\}, \phi_{010})}(\mathrm{c})\hat{f}_{(\{00\}, \phi_{010})}(\mathrm{d})\hat{f}_{(\{00, 10\}, \phi_{001})}(\mathrm{b})
= 1 011 1101.
\end{equation*}
\end{example}

\begin{example}
We show the decoding process of the codeword sequence $10111101$ obtained in Example \ref{ex:reduced-encode}.
We suppose that the given codeword sequence is $\pmb{b} = \hat{f}^{\ast}_{(\{00, 01, 10, 11\}, \phi_{000})}(\mathrm{acdb})00 = 1011110100$, that is, the additional $2$ bits $00$ noted in Remark \ref{rem:not-unique} is appended to the codeword sequence.
The decoder is also shared $(\{00, 01, 10, 11\}, \phi_{000})$, the initial index of encoding, with the encoder.

\begin{itemize}
\item At the beginning of the first iteration, the initial configuration is $(A, \phi) = (\{00, 01, 10, 11\}, \phi_{000}), \pmb{b} = 1011110100$.
Then
\begin{equation*}
\pmb{b}' = \phi_{000}^{-1}(\pmb{b}) = \phi_{000}(\pmb{b}) = 1011\ldots
\end{equation*}
and 
\begin{equation*}
\tilde{f}_{\{00, 01, 10, 11\}}(\mathrm{a}) \tilde{\trans}_{\{00, 01, 10, 11\}}(\mathrm{a}) = 1\{01, 10\} \owns 101 \preceq \pmb{b}'.
\end{equation*}
Therefore, we recover $x_1 = \mathrm{a}$ and the next index is 
\begin{equation*}
\hat{\trans}_{(\{00, 01, 10, 11\}, \phi_{000})}(\mathrm{a})
= \left( \{00, 10\}, \quot{\phi_{000}}{1} \circ \phi_{010}\right)
= \left( \{00, 10\}, \phi_{010}\right).
\end{equation*}

\item At the beginning of the second iteration, the current configuration is $(A, \phi) = \left( \{00, 10\}, \phi_{010}\right), \pmb{b} = 011110100$.
Then
\begin{equation*}
\pmb{b}' = \phi_{010}^{-1}(\pmb{b}) = \phi_{010}(\pmb{b}) = 0011\ldots
\end{equation*}
and 
\begin{equation*}
\tilde{f}_{\{00, 10\}}(\mathrm{c}) \tilde{\trans}_{\{00, 10\}}(\mathrm{c}) = \{00\} \owns 00 \preceq \pmb{b}'.
\end{equation*}
Therefore, we recover $x_2 = \mathrm{c}$ and the next index is
\begin{equation*}
\hat{\trans}_{(\{00, 10\}, \phi_{010})}(\mathrm{c})
= \left( \{00\}, \quot{\phi_{010}}{\lambda} \circ \phi_{000}\right)\
= \left( \{00\}, \phi_{010}\right).
\end{equation*}

\item At the beginning of the third iteration, the current configuration is $(A, \phi) = \left( \{00\}, \phi_{010}\right), \pmb{b} = 011110100$.
Then
\begin{equation*}
\pmb{b}' = \phi_{010}^{-1}(\pmb{b}) = \phi_{010}(\pmb{b}) = 00111\ldots
\end{equation*}
and 
\begin{equation*}
\tilde{f}_{\{00\}}(\mathrm{d}) \tilde{\trans}_{\{00\}}(\mathrm{d}) = 001\{00, 11\} \owns 00111 \preceq \pmb{b}'.
\end{equation*}
Therefore, we recover $x_3 = \mathrm{d}$ and the next index is
\begin{equation*}
\hat{\trans}_{(\{00\}, \phi_{010})}(\mathrm{d})
= \left( \{00, 10\}, \quot{\phi_{010}}{001} \circ \phi_{001}\right)
= \left( \{00, 10\}, \phi_{001}\right).
\end{equation*}

\item At the beginning of the fourth iteration, the current configuration is $(A, \phi) = \left( \{00, 10\}, \phi_{001}\right)), \pmb{b} = 110100$.
Then
\begin{equation*}
\pmb{b}' = \phi_{001}^{-1}(\pmb{b}) = \phi_{001}(\pmb{b}) = 100100
\end{equation*}
and 
\begin{equation*}
\tilde{f}_{\{00, 10\}}(\mathrm{b}) \tilde{\trans}_{\{00, 10\}}(\mathrm{b}) = 1001\{00, 01, 10, 11\} \owns 110100 \preceq \pmb{b}'.
\end{equation*}
Therefore, we recover $x_4 = \mathrm{b}$ and the next index is
\begin{equation*}
\hat{\trans}_{(\{00, 10\}, \phi_{001})}(\mathrm{b})
= \left( \{00, 01, 10, 11\}, \quot{\phi_{001}}{1001} \circ \phi_{000}\right)
= \left( \{00, 01, 10, 11\}, \phi_{000}\right).
\end{equation*}

\item At the beginning of the fifth iteration, the current configuration is $(A, \phi) = \left( \{00, 01, 10, 11\}, \phi_{000}\right), \pmb{b} = 00$.
Now, since there exists no $s \in \mathcal{S}$ satisfying (\ref{eq:dqovdwqo208r}), we terminate the decoding process.
 \end{itemize}

Consequently, we recover the original source sequence $\pmb{x} = \mathrm{acdb}$ correctly.

Note that in each step of the decoding process, it is not necessary to calculate the entire sequence $\pmb{b}' = \phi^{-1}(\pmb{b})$, but only a sufficiently long prefix of $\pmb{b}'$.
\end{example}

\section{Conclusion}
\label{sec:conclusion}

In this paper, we proposed a method to reduce the number of code tables to be considered to discuss $k$-bit delay optimal code-tuples.

We first defined the set $\Phi_k$ of mappings and divided $\mathscr{C}^k$ into equivalence classes using $\Phi_k$.
Then we introduced a reduced-code-tuple (RCT), in which each equivalence class is associated with at most one code table.
We next proved Theorem \ref{thm:equiv} that for any $\tilde{F} \in \tilde{\mathscr{F}}_{k\hopt}$, there exists $F \in \mathscr{F}_{k\hopt}$ and vice versa. 
This theorem implies that the construction of a $k$-bit delay optimal code-tuple can be reduced to the construction of a $k$-bit delay optimal RCT.
We also presented a coding procedure using an RCT equivalent to the coding procedure using a code-tuple.

Since the number of the equivalence classes is dramatically less than the upper bound $2^{(2^k)}$ in \cite{IEICE2023} as shown by Theorem \ref{thm:orbit-num} (cf.~Table \ref{tab:orbit-num}),
we can dramatically reduce the number of code tables to be considered in the theoretical analysis, code construction, and coding process by replacing code-tuples with RCTs.

As mentioned in Remark \ref{rem:differ}, the upper bound $2^{(2^k)}$ can be slightly improved to $2^{(2^k)} - 2^{(2^{k-1} + 1)} + 1$ by \cite{JSAIT2022}.
Applying this fact to the results of this paper also slightly improves the upper bounds $a_k$ in Table \ref{tab:orbit-num} to $a'_k \coloneqq  a_{k-1}(a_{k-1}-1)/2$ in Table \ref{tab:orbit-num2}.
The upper bound $a'_2 = 3$ in Table \ref{tab:orbit-num2} can be improved further to $2$: it is sufficient to consider at most two code tables because it is proven by  \cite{Hashimoto2023} that there exists a $2$-bit delay optimal code-tuple that is also an AIFV-$2$ code.

  \begin{table*}
  \caption{The upper bound $a_k$ in Table \ref{tab:orbit-num} and the improved upper bound $a'_k$}
 \label{tab:orbit-num2}
  \centering
\begin{tabular}{c|rrrrrrr}
 \multicolumn{1}{c|}{$k$} & \multicolumn{1}{c}{1} & \multicolumn{1}{c}{2} & \multicolumn{1}{c}{3} & \multicolumn{1}{c}{4} & \multicolumn{1}{c}{5} & \multicolumn{1}{c}{6} & \multicolumn{1}{c}{7}  \\
\hline
$a_k$ & 3 & 6 & 21 & 231 & 26796 & 359026206  & 64449908476890321\\
$a'_k$ & 1 & 3 & 15 & 210 & 26565 & 358999410  & 64449908117864115
 \end{tabular}
  \end{table*}

\appendix

\section{The Proofs}
\label{sec:proofs}

\subsection{Proof of Lemma \ref{lem:pref}}
\label{subsec:proof-pref}

\begin{proof}[Proof of Lemma \ref{lem:pref}]
(Proof of (i))
For any $\pmb{c} \in \mathcal{C}^k$, we have
\begin{eqnarray*}
\pmb{c} \in \PREF^k_{F, i}
&\eqlab{A}{\iff}& {\exists}\pmb{x} \in \mathcal{S}^+ \,\,\mathtt{s.t.}\,\, f^{\ast}_i(\pmb{x}) \succeq \pmb{c}\\
&\iff& {\exists}s \in \mathcal{S}, \pmb{x} \in \mathcal{S}^{\ast} \,\,\mathtt{s.t.}\,\, f^{\ast}_i(s\pmb{x}) \succeq \pmb{c}\\
&\eqlab{B}\iff& {\exists}s \in \mathcal{S}, \pmb{x} \in \mathcal{S}^{\ast} \,\,\mathtt{s.t.}\,\, f_i(s)f^{\ast}_{\trans_i(s)}(\pmb{x}) \succeq \pmb{c}\\
&\eqlab{C}{\iff}& {\exists}s \in \mathcal{S} \,\,\mathtt{s.t.}\,\, \pmb{c} \in \left[ f_i(s)\PREF^{k}_{F, \trans_i(s)} \right]_k\\
&\iff& \pmb{c} \in \bigcup_{s \in \mathcal{S}} \left[f_i(s)\PREF^{k}_{F, \trans_i(s)} \right]_k,
\end{eqnarray*}
where
(A) follows from (\ref{eq:pref3}),
(B) follows from (\ref{eq:fstar}),
and (C) follows from (\ref{eq:pref3}).

(Proof of (ii))
For any $\pmb{c} \in \mathcal{C}^k$, we have
\begin{eqnarray*}
\pmb{c} \in \bar{\PREF}^k_{F, i}(\pmb{b})
&\eqlab{A}{\iff}& {\exists}\pmb{x} \in \mathcal{S}^+ \,\,\mathtt{s.t.}\,\, (f_i(x_1) \succ \pmb{b}, f^{\ast}_i(\pmb{x}) \succeq \pmb{bc})\\
&\iff& {\exists}s \in \mathcal{S}, \pmb{x} \in \mathcal{S}^{\ast} \,\,\mathtt{s.t.}\,\, (f_i(s) \succ \pmb{b}, f^{\ast}_i(s\pmb{x}) \succeq \pmb{bc})\\
&\eqlab{B}\iff& {\exists}s \in \mathcal{S}, \pmb{x} \in \mathcal{S}^{\ast} \,\,\mathtt{s.t.}\,\, (f_i(s) \succ \pmb{b}, f_i(s)f^{\ast}_{\trans_i(s)}(\pmb{x}) \succeq \pmb{b}\pmb{c})\\
&\iff& {\exists}s \in \mathcal{S}, \pmb{x} \in \mathcal{S}^{\ast} \,\,\mathtt{s.t.}\,\, (f_i(s) \succ \pmb{b}, \pmb{b}^{-1}f_i(s)f^{\ast}_{\trans_i(s)}(\pmb{x}) \succeq \pmb{c})\\
&\eqlab{C}{\iff}& {\exists}s \in \mathcal{S} \,\,\mathtt{s.t.}\,\, \left(f_i(s) \succ \pmb{b}, \pmb{c} \in \left[ \pmb{b}^{-1}f_i(s)\PREF^{k}_{F, \trans_i(s)} \right]_k \right)\\
&\iff& \pmb{c} \in \bigcup_{\substack{s \in \mathcal{S},\\ f_i(s) \succ \pmb{b}}} \left[ \pmb{b}^{-1}f_i(s)\PREF^{k}_{F, \trans_i(s)} \right]_k,
\end{eqnarray*}
where
(A) follows from (\ref{eq:pref2}),
(B) follows from (\ref{eq:fstar}),
and (C) follows from (\ref{eq:pref3}).
\end{proof}

\subsection{Proof of Lemma \ref{lem:closed}}
\label{subsec:proof-closed}

\begin{proof}[Proof of Lemma \ref{lem:closed}]
We choose and fix $p \in [F]$ arbitrarily and define a set $\mathcal{I} \subseteq [F]$ as
\begin{equation}
\label{eq:y5srs1ad8npa}
\mathcal{I} \coloneqq \{i \in [F] : {\forall}\pmb{x}' \in \mathcal{S}^{\ast}, \trans^{\ast}_i(\pmb{x}') \neq  p\}.
\end{equation}

Then $\mathcal{I}$ is a closed set of $F$: we have $\trans^{\ast}_{i}(\pmb{x}) \in \mathcal{I}$ for any $i \in \mathcal{I}$ and $\pmb{x} \in \mathcal{S}^{\ast}$ because 
\begin{equation*}
\forall\pmb{x}' \in \mathcal{S}^{\ast}, \trans^{\ast}_{\trans^{\ast}_i(\pmb{x})}(\pmb{x}')
\eqlab{A}{=} \trans^{\ast}_i(\pmb{xx}')
\eqlab{B}{\neq} p,
\end{equation*}
where
(A) follows from Lemma \ref{lem:f_T} (ii),
and (B) follows from $i \in \mathcal{I}$ and (\ref{eq:y5srs1ad8npa}).

Also, $\mathcal{I}$ is a proper subset of $[F]$:
\begin{equation*}
\mathcal{I}
\eqlab{A}{\subseteq} [F] \setminus \{p\}
\subsetneq [F],
\end{equation*}
where
(A) follows since $p \not\in \mathcal{I}$ from (\ref{eq:y5srs1ad8npa}).

Therefore, $\mathcal{I}$ is a closed proper subset of $[F]$.
Since $[F]$ is a minimal non-empty closed set of $F$, it must hold that $\mathcal{I} = \emptyset$.
Consequently, for any $i \in [F]$, there exists $\pmb{x}' \in \mathcal{S}^{\ast}$ such that $\trans^{\ast}_{i}(\pmb{x}') = p$.
Namely, we have $p \in \kernel_{F}$.

Since $p$ is arbitrarily chosen from $[F]$, we obtain $[F] = \kernel_F$ as desired.
\end{proof}

\subsection{Proof of Lemma \ref{lem:phi-bijection}}
\label{subsec:proof-phi-bijection}

\begin{proof}[proof of Lemma \ref{lem:phi-bijection}]
(Injectivity)
For any $\pmb{b}, \pmb{b}' \in \mathcal{C}^{\ast}$, we have
\begin{eqnarray*}
\phi(\pmb{b}) = \phi(\pmb{b}')
&\iff& \phi(\pmb{b}) \preceq \phi(\pmb{b}'), \phi(\pmb{b}) \succeq \phi(\pmb{b}')\\
&\eqlab{A}{\iff}& \pmb{b} \preceq \pmb{b}', \pmb{b} \succeq \pmb{b}'\\
&\iff& \pmb{b} = \pmb{b}'
\end{eqnarray*}
as desired,
where (A) follows from Definition \ref{def:Phi} (b).

(Surjectivity)
By Definition \ref{def:Phi} (a), we have
\begin{equation}
\label{eq:v6r2ot892s3h}
\phi(\mathcal{C}^l) \subseteq \mathcal{C}^l
\end{equation}
for any integer $l \geq 0$.
In (\ref{eq:v6r2ot892s3h}), the equality holds because
\begin{eqnarray*}
|\phi(\mathcal{C}^l)| = |\mathcal{C}^l|
\end{eqnarray*}
by the injectivity of $\phi$.
Namely, 
\begin{equation}
\label{eq:ll36goxxdomg}
\phi(\mathcal{C}^l) = \mathcal{C}^l
\end{equation}
holds for any integer $l \geq 0$.
Therefore, for any $\pmb{b}' \in \mathcal{C}^{\ast}$, there exists $\pmb{b} \in \mathcal{C}^{l}$ such that $\phi(\pmb{b}) = \pmb{b}'$, where $l = |\pmb{b}'|$.
\end{proof}

\subsection{Proof of Lemma \ref{lem:phi-quot}}
\label{subsec:proof-phi-quot}

\begin{proof}[Proof of Lemma \ref{lem:phi-quot}]
(Proof of (i))
We have
\begin{equation*}
\phi(\pmb{bb}')
\eqlab{A}{=} \phi(\pmb{b})\phi(\pmb{b})^{-1}\phi(\pmb{bb}')
\eqlab{B}{=}  \phi(\pmb{b})\quot{\phi}{\pmb{b}}(\pmb{b}'),
\end{equation*}
where
(A) follows since $\phi(\pmb{b}) \preceq \phi(\pmb{bb}')$ is implied by $\pmb{b} \preceq \pmb{bb}'$ and Definition \ref{def:Phi} (b),
and (B) follows from (\ref{eq:r6xuob4upqw5}).

(Proof of (ii))
We show $\quot{\phi}{\pmb{d}}$ satisfies Definition \ref{def:Phi} (a)--(c).
\begin{itemize}
\item (Definition \ref{def:Phi} (a)) For any $\pmb{b} \in \mathcal{C}^{\ast}$, we have
\begin{equation*}
\left|\quot{\phi}{\pmb{d}}(\pmb{b})\right|
\eqlab{A}{=}  |\phi(\pmb{d})^{-1}\phi(\pmb{d}\pmb{b})|
=  -|\phi(\pmb{d})| + |\phi(\pmb{d}\pmb{b})|
\eqlab{B}{=}  -|\pmb{d}| + |\pmb{d}\pmb{b}|
= |\pmb{b}|,
\end{equation*}
where
(A) follows from (\ref{eq:r6xuob4upqw5}),
and (B) follows from Definition \ref{def:Phi} (a).

\item (Definition \ref{def:Phi} (b))
For any $\pmb{b}, \pmb{b}' \in \mathcal{C}^{\ast}$, we have
\begin{eqnarray*}
\pmb{b} \preceq \pmb{b}'
&\iff& \pmb{d}\pmb{b} \preceq \pmb{d}\pmb{b}' \\
&\eqlab{A}{\iff}& \phi(\pmb{d}\pmb{b}) \preceq \phi(\pmb{d}\pmb{b}') \\
&\eqlab{B}{\iff}& \phi(\pmb{d})^{-1}\phi(\pmb{d}\pmb{b}) \preceq \phi(\pmb{d})^{-1}\phi(\pmb{d}\pmb{b}') \\
&\eqlab{C}{\iff}& \quot{\phi}{\pmb{d}}(\pmb{b}) \preceq \quot{\phi}{\pmb{d}}(\pmb{b}'),
\end{eqnarray*}
where
(A) follows from Definition \ref{def:Phi} (b),
(B) follows since $\phi(\pmb{d}) \preceq \phi(\pmb{db})$ and  $\phi(\pmb{d}) \preceq \phi(\pmb{db}')$ hold by Definition \ref{def:Phi} (b),
and (C) follows from (\ref{eq:r6xuob4upqw5}).

\item (Definition \ref{def:Phi} (c))
For any $\pmb{b} \in \mathcal{C}^{\geq K}$, we have
\begin{eqnarray*}
\suff^{K}\left(\quot{\phi}{\pmb{d}}(\pmb{b})\right)
&\eqlab{A}{=}&  \suff^{K}(\phi(\pmb{d})^{-1}\phi(\pmb{d}\pmb{b}))\\
&=&  \suff^{K+|\phi(\pmb{d})|}(\phi(\pmb{d}\pmb{b}))\\
&\eqlab{B}{=}&  \suff^{K+|\pmb{d}|}(\phi(\pmb{d}\pmb{b}))\\
&\eqlab{C}{=}&  \suff^{K+|\pmb{d}|-k}(\suff^k(\phi(\pmb{d}\pmb{b})))\\
&\eqlab{D}{=}&  \suff^{K+|\pmb{d}|-k}(\suff^k(\pmb{d}\pmb{b}))\\
&=&  \suff^{K+|\pmb{d}|}(\pmb{d}\pmb{b})\\
&=&  \suff^{K}(\pmb{b}),
\end{eqnarray*}
where
(A) follows from (\ref{eq:r6xuob4upqw5}),
(B) follows from Definition \ref{def:Phi} (a),
(C) follows from $K + |\pmb{d}| = \max\{0, k-|\pmb{d}|\} + |\pmb{d}| = \max\{|\pmb{d}|, k\} \geq k$,
and (D) follows from Definition \ref{def:Phi} (c).
\end{itemize}

Since $K = \max\{0, k - |\pmb{d}|\} \leq k$, we also have $\quot{\phi}{\pmb{d}} \in \Phi_K \subseteq \Phi_k$.
\end{proof}

\subsection{Proof of Lemma \ref{lem:phi-represent}}
\label{subsec:proof-phi-represent}

\begin{proof}[Proof of Lemma \ref{lem:phi-represent}]
(Proof of (i))
In general, for any $\psi \in \Phi_k$ and $c \in \mathcal{C}$, we have
\begin{equation}
\label{eq:3ggv1jl6ft06}
\psi(c) = c \oplus \psi(0),
\end{equation}
that is, $\psi(0) = 0 \oplus \psi(0)$ and $\psi(1) = 1 \oplus \psi(0)$
because $\psi(\{0, 1\}) = \{0, 1\}$ holds by (\ref{eq:ll36goxxdomg}).

Now, we prove the desired assertion (\ref{eq:uheypssietwo}) by induction on $|\pmb{b}|$.
The base case $|\pmb{b}| = 0$ is trivial.
For the induction step for $|\pmb{b}| \geq 1$, we have
\begin{eqnarray*}
\phi(\pmb{b}) &=& \phi([\pmb{b}]_{|\pmb{b}|-1} b_{|\pmb{b}|}) \\
&\eqlab{A}=& \phi([\pmb{b}]_{|\pmb{b}|-1}) \quot{\phi}{[\pmb{b}]_{|\pmb{b}|-1}}(b_{|\pmb{b}|}) \\
&\eqlab{B}=& (b_1\oplus\phi^{\ast}([\pmb{b}]_0))(b_2\oplus\phi^{\ast}([\pmb{b}]_1))\ldots (b_{|\pmb{b}|-1}\oplus\phi^{\ast}([\pmb{b}]_{|\pmb{b}|-2})) \quot{\phi}{[\pmb{b}]_{|\pmb{b}|-1}}(b_{|\pmb{b}|})\\
&\eqlab{C}=& (b_1\oplus\phi^{\ast}([\pmb{b}]_0))(b_2\oplus\phi^{\ast}([\pmb{b}]_1))\ldots (b_{|\pmb{b}|-1}\oplus\phi^{\ast}([\pmb{b}]_{|\pmb{b}|-2})) \left(b_{|\pmb{b}|}\oplus\quot{\phi}{[\pmb{b}]_{|\pmb{b}|-1}}(0)\right)\\
&\eqlab{D}=& (b_1\oplus\phi^{\ast}([\pmb{b}]_0))(b_2\oplus\phi^{\ast}([\pmb{b}]_1))\ldots (b_{|\pmb{b}|-1}\oplus\phi^{\ast}([\pmb{b}]_{|\pmb{b}|-2})) (b_{|\pmb{b}|}\oplus\phi^{\ast}([\pmb{b}]_{|\pmb{b}|-1}))
\end{eqnarray*}
as desired, where
(A) follows from Lemma \ref{lem:phi-quot} (i),
(B) follows from the induction hypothesis,
(C) follows from (\ref{eq:3ggv1jl6ft06}),
and (D) follows from (\ref{eq:phi-star}).

(Proof of (ii))
We have
\begin{eqnarray*}
\phi^{\ast}(\pmb{d})
&\eqlab{A}=&\quot{\phi}{\pmb{d}}(0)\\
&=& \suff^{|\phi(\pmb{d})|}\left(\phi(\pmb{d}) \quot{\phi}{\pmb{d}}(0)\right)\\
&\eqlab{B}{=}& \suff^{|\phi(\pmb{d})|}(\phi(\pmb{d}0)) \\
&\eqlab{C}{=}& \suff^{|\pmb{d}|}(\phi(\pmb{d}0)) \\
&\eqlab{D}{=}& \suff^{|\pmb{d}|-k}(\suff^k(\phi(\pmb{d}0))) \\
&\eqlab{E}{=}& \suff^{|\pmb{d}|-k}(\suff^k(\pmb{d}0)) \\
&=& \suff^{|\pmb{d}|}(\pmb{d}0) \\
&=& 0,
\end{eqnarray*}
where
(A) follows from (\ref{eq:phi-star}),
(B) follows from Lemma \ref{lem:phi-quot} (i),
(C) follows from Definition \ref{def:Phi} (a),
(D) follows since $|\pmb{d}|-k \geq 0$ by $\pmb{d} \in \mathcal{C}^{\geq k}$,
and (E) follows from Definition \ref{def:Phi} (c).
\end{proof}

\subsection{Proof of Theorem \ref{thm:orbit-num}}
\label{subsec:proof-orbit-num}

The proof of Theorem \ref{thm:orbit-num} relies on the following Lemma \ref{lem:obrit-num}.

\begin{lemma}
\label{lem:obrit-num}
Let $k \geq 1$ be an integer and let $A,  B \in \mathscr{C}^k$ be represented as $A = 0 A_0 \cup 1 A_1$ and $ B = 0 B_0 \cup 1 B_1$ by $A_0, A_1,  B_0,  B_1 \in \mathscr{C}^{k-1}$.
Then $A \sim_k  B$ if and only if
there exists $c \in \mathcal{C}$ such that $A_0 \sim_{k-1}  B_c$ and $A_1 \sim_{k-1}  B_{\bar{c}}$,
where $\bar{c}$ denotes the negation of $c \in \mathcal{C}$, that is, $\bar{0} \coloneqq 1$ and $\bar{1} \coloneqq 0$.
\end{lemma}

\begin{proof}[Proof of Lemma \ref{lem:obrit-num}]

(Necessity)
Assume $A \sim_k  B$.
Then there exists $\phi \in \Phi_k$ such that
\begin{equation}
\label{eq:mnek0mjij7au}
\phi(A) =  B.
\end{equation}
Putting $c \coloneqq \phi(0)$, we have $\bar{c} = \phi(1)$; and thus
\begin{eqnarray*}
c\quot{\phi}{0}(A_0) \cup \bar{c}\quot{\phi}{1}(A_1)
&=& \phi(0)\quot{\phi}{0}(A_0) \cup \phi(1)\quot{\phi}{1}(A_1)\\
&\eqlab{A}{=}&  \phi(0 A_0) \cup \phi(1 A_1)\\
&\eqlab{B}{=}&  \phi(0 A_0 \cup 1 A_1)\\
&=& \phi(A)\\
&\eqlab{C}=&  B\\
&=& c B_c \cup \bar{c} B_{\bar{c}},
\end{eqnarray*}
where
(A) follows from Lemma \ref{lem:phi-quot} (i),
(B) follows from Lemma \ref{lem:phi-bijection},
and (C) follows from (\ref{eq:mnek0mjij7au}).
Comparing both sides, we obtain $B_c = \quot{\phi}{0}(A_0)$ and $B_{\bar{c}} = \quot{\phi}{1}(A_1)$.
Since $\quot{\phi}{0} , \quot{\phi}{1} \in \Phi_{k-1}$ by Lemma \ref{lem:phi-quot} (ii), it follows that 
$A_0 \sim_{k-1}  B_c$ and $A_1 \sim_{k-1}  B_{\bar{c}}$  as desired.

(Sufficiency)
Assume that there exists $c \in \mathcal{C}$ such that $A_0 \sim_{k-1}  B_c$ and $A_1 \sim_{k-1}  B_{\bar{c}}$.
Then there exist $\phi_0, \phi_1 \in \Phi_{k-1}$ such that
\begin{equation}
\label{eq:96gxg4g80e6c}
\phi_0(A_0) =  B_c, \quad \phi_1(A_1) =  B_{\bar{c}}.
\end{equation}
We define $\phi \colon \mathcal{C}^{\ast} \to \mathcal{C}^{\ast}$ as
\begin{equation}
\label{eq:lfc0hqyuauug}
\phi(\pmb{b}) = 
\begin{cases}
\lambda &\,\,\text{if}\,\, \pmb{b} = \lambda,\\
c\phi_0(\suff(\pmb{b})) &\,\,\text{if}\,\, \pmb{b} \succeq 0,\\
\bar{c}\phi_1(\suff(\pmb{b})) &\,\,\text{if}\,\, \pmb{b} \succeq 1
\end{cases}
\end{equation}
for $\pmb{b} \in \mathcal{C}^{\ast}$.
We show $\phi \in \Phi$, that is, $\phi$ satisfies Definition \ref{def:Phi} (a)--(c).

\begin{itemize}
\item (Definition \ref{def:Phi} (a)) We have
\begin{eqnarray*}
|\phi(\pmb{b})|
&\eqlab{A}{=}&\begin{cases}
|\lambda| &\,\,\text{if}\,\, \pmb{b} = \lambda,\\
|c\phi_0(\suff(\pmb{b}))| &\,\,\text{if}\,\, \pmb{b} \succeq 0,\\
|\bar{c}\phi_1(\suff(\pmb{b}))| &\,\,\text{if}\,\, \pmb{b} \succeq 1
\end{cases}\\
&=& \begin{cases}
|\pmb{b}| &\,\,\text{if}\,\, \pmb{b} = \lambda,\\
1+|\phi_{b_1}(\suff(\pmb{b}))| &\,\,\text{if}\,\, \pmb{b} \neq \lambda,\\
\end{cases}\\
&\eqlab{B}{=}& \begin{cases}
|\pmb{b}| &\,\,\text{if}\,\, \pmb{b} = \lambda,\\
1+|\suff(\pmb{b})| &\,\,\text{if}\,\, \pmb{b} \neq \lambda,\\
\end{cases}\\
&=& |\pmb{b}|,
\end{eqnarray*}
where
$b_1$ denotes the first letter of $\pmb{b}$, and
(A) follows from (\ref{eq:lfc0hqyuauug}),
and (B) follows since $\phi_{b_1}$ satisfies Definition \ref{def:Phi} (a).

\item(Definition \ref{def:Phi} (b))
($\implies$)
Assume $\pmb{b} \preceq \pmb{b}'$.
If $\pmb{b} = \lambda$, then
\begin{equation*}
\phi(\pmb{b}) = \phi(\lambda) \eqlab{A}{=} \lambda \preceq \phi(\pmb{b}')
\end{equation*}
as desired, where
(A) follows from the first case of (\ref{eq:lfc0hqyuauug}).
Now, we assume $\pmb{b} \neq \lambda$.
Then by $\pmb{b} \preceq \pmb{b}'$, it follows that $\suff(\pmb{b}) \preceq \suff(\pmb{b}')$,
which yields
\begin{equation}
\label{eq:hvqkex16o0fn}
\phi_0(\suff(\pmb{b})) \preceq \phi_0(\suff(\pmb{b}'))
\end{equation}
since $\phi_0 \in \Phi$ satisfies Definition \ref{def:Phi} (b).

We consider the case $b_1 = 0$ and the case $b_1 = 1$ separately.
In the case $b_1 = 0$, we have
\begin{equation*}
\phi(\pmb{b})
\eqlab{A}{=} c\phi_0(\suff(\pmb{b}))
\eqlab{B}{\preceq} c\phi_0(\suff(\pmb{b}'))
\eqlab{C}{=} \phi(\pmb{b}'),
\end{equation*}
where
(A) follows from the second case of (\ref{eq:lfc0hqyuauug}),
(B) follows from (\ref{eq:hvqkex16o0fn}),
and (C) follows from the second case of (\ref{eq:lfc0hqyuauug}).
For the other case $b_1 = 1$, the symmetric argument can be applied.

($\impliedby$)
Assume $\phi(\pmb{b}) \preceq \phi(\pmb{b}')$.
If $\phi(\pmb{b}) = \lambda$, then
\begin{equation*}
\pmb{b} \eqlab{A}{=} \lambda \preceq \pmb{b}'
\end{equation*}
as desired, where
(A) follows from (\ref{eq:lfc0hqyuauug}) and $\phi(\pmb{b}) = \lambda$.
Now, we assume $\phi(\pmb{b}) \neq \lambda$.
Then by $\phi(\pmb{b}) \preceq \phi(\pmb{b}')$, the first letters of $\phi(\pmb{b})$ and $\phi(\pmb{b}')$ must be equal:
we have exactly one of $c \preceq \phi(\pmb{b}) \preceq \phi(\pmb{b}')$ and $\bar{c} \preceq \phi(\pmb{b}) \preceq \phi(\pmb{b}')$.

We now consider the case $c \preceq \phi(\pmb{b}) \preceq \phi(\pmb{b}')$, which is possible only if
\begin{equation}
\label{eq:onyo2e2duoah}
b_1 = b'_1 = 0
\end{equation}
and accordingly
\begin{equation*}
\phi(\pmb{b}) = c\phi_0(\suff(\pmb{b})), \quad \phi(\pmb{b}') = c\phi_0(\suff(\pmb{b}'))
\end{equation*}
by (\ref{eq:lfc0hqyuauug}).
Hence, we have
\begin{equation*}
c\phi_0(\suff(\pmb{b})) = \phi(\pmb{b}) \preceq \phi(\pmb{b}') = c\phi_0(\suff(\pmb{b}'))
\end{equation*}
and thus
\begin{equation*}
\phi_0(\suff(\pmb{b})) \preceq \phi_0(\suff(\pmb{b'})),
\end{equation*}
which yields
\begin{equation}
\label{eq:72vprqs7olk2}
\suff(\pmb{b}) \preceq \suff(\pmb{b'})
\end{equation}
since $\phi_0 \in \Phi$ satisfies Definition \ref{def:Phi} (b).
Therefore, we obtain
\begin{equation*}
\pmb{b} = b_1\suff(\pmb{b})
\eqlab{A}{=} b'_1\suff(\pmb{b})
\eqlab{B}{\preceq} b'_1\suff(\pmb{b'}) = \pmb{b}'
\end{equation*}
as desired, where 
(A) follows from (\ref{eq:onyo2e2duoah}),
and (B) follows from (\ref{eq:72vprqs7olk2}).
In the other case $\bar{c} \preceq \phi(\pmb{b}) \preceq \phi(\pmb{b}')$, the symmetric argument can be applied.

\item (Definition \ref{def:Phi} (c))
For any $\pmb{b} \in \mathcal{C}^{\geq k}$,
we have
\begin{eqnarray*}
\suff^k(\phi(\pmb{b}))
&\eqlab{A}{=}&\begin{cases}
\suff^k(c\phi_0(\suff(\pmb{b}))) &\,\,\text{if}\,\, \pmb{b} \succeq 0,\\
\suff^k(\bar{c}\phi_1(\suff(\pmb{b}))) &\,\,\text{if}\,\, \pmb{b} \succeq 1
\end{cases}\\
&=&\begin{cases}
\suff^{k-1}(\phi_0(\suff(\pmb{b}))) &\,\,\text{if}\,\, \pmb{b} \succeq 0,\\
\suff^{k-1}(\phi_1(\suff(\pmb{b}))) &\,\,\text{if}\,\, \pmb{b} \succeq 1
\end{cases}\\
&\eqlab{B}{=}& \suff^{k-1}(\suff(\pmb{b}))\\
&=& \suff^k(\pmb{b}),
\end{eqnarray*}
where
(A) follows from (\ref{eq:lfc0hqyuauug}) since $\pmb{b} \neq \lambda$ by the assumption $k \geq 1$,
and (B) follows from $\phi_0, \phi_1 \in \Phi_{k-1}$ and Definition \ref{def:Phi} (c).
\end{itemize}

Therefore, we conclude that $\phi \in \Phi_k$, and we obtain 
\begin{equation*}
\phi(A)
= \phi(0 A_0 \cup 1 A_1)
\eqlab{A}{=} \phi(0 A_0) \cup \phi(1 A_1)
\eqlab{B}{=} c\phi_0(A_0) \cup \bar{c}\phi_1(A_1)
\eqlab{C}{=} c B_c \cup \bar{c} B_{\bar{c}}
=  B,
\end{equation*}
so that $A \sim  B$,
where
(A) follows from Lemma \ref{lem:phi-bijection},
(B) follows from the second and third cases of (\ref{eq:lfc0hqyuauug}),
and (C) follows from (\ref{eq:96gxg4g80e6c}).
\end{proof}

\begin{proof}[Proof of Theorem \ref{thm:orbit-num}]
For $k = 0$, we have $\mathscr{C}^0 / {\sim_0} = \{\emptyset, \{\lambda\}\}$, and thus $a_0 = 2$.
For $k \geq 1$, two elements $A, B \in \mathscr{C}^k$ satisfy $A \sim_k  B$ if and only if unordered pairs $\{\langle  A_0 \rangle, \langle  A_1 \rangle\}$ and $\{\langle  B_0 \rangle, \langle  B_1 \rangle\}$ are identical by Lemma \ref{lem:obrit-num}, 
where the unordered pairs are considered as multisets, that is, $\langle  A_0 \rangle$ and $\langle  A_1 \rangle$ (resp.~$\langle  B_0 \rangle$ and $\langle  B_1 \rangle$) can be the same elements of $\mathscr{C}^{k-1} / {\sim}_{k-1}$.
Therefore, $a_k$ is equal to the number of ways to choose an unordered pair of equivalence classes in $\mathscr{C}^{k-1} / {\sim}_{k-1}$, that is, $a_k = \binom{a_{k-1}+1}{2}$.
\end{proof}

\subsection{Proof of Theorem \ref{thm:equiv} (i)}
\label{subsec:proof-equiv1}

The proof of Theorem \ref{thm:equiv} (i) relies on the following Lemmas \ref{lem:irr}--\ref{lem:to-semi}.
See \cite{IEICE2023} for the proofs of Lemmas \ref{lem:irr}--\ref{lem:improve}.

\begin{lemma}[{\cite[Lemmas 9, 10]{IEICE2023}}]
\label{lem:irr}
For any $F=(f, \trans) \in \mathscr{F}_{\reg}$, there exists $\bar{F} = (\bar{f}, \bar{\trans}) \in \mathscr{F}_{\irr}$ satisfying the following conditions (a)--(d).
\begin{enumerate}[(a)]
\item $L(\bar{F}) = L(F)$.
\item $[\bar{F}] = \kernel_F \subseteq [F]$.
\item For any $A \in [\bar{F}]$, it holds that $\bar{f}_{A} = f_{A}$.
\item For any $A \in [\bar{F}]$, it holds that $\bar{\trans}_{A} = \trans_{A}$.
\end{enumerate}
\end{lemma}

\begin{lemma}[{\cite[Lemma 13]{IEICE2023}}]
\label{lem:chooseone}
For any $F=(f, \trans) \in \mathscr{F}_{\irr}$, $\mathcal{I} \subseteq [F]$, and $P \in \mathcal{I}$,
the code-tuple $F'=(f', \trans') \in \mathscr{F}$ defined as (\ref{eq:v2572xoo0ye0})--(\ref{eq:hbm12zjixhzy}) satisfies $F' \in \mathscr{F}_{\reg}$:
\begin{align}
\label{eq:v2572xoo0ye0}
[F'] &\coloneqq [F]
\end{align}
and for $A \in [F']$ and $s \in \mathcal{S}$,
\begin{align}
f'_A(s) &\coloneqq f_A(s), \label{eq:cp8qdsjrwe14}\\
\trans'_A(s) &\coloneqq 
\begin{cases}
P & \,\,\text{if}\,\, \trans_A(s) \in \mathcal{I},\\
\trans_A(s) & \,\,\text{if}\,\, \trans_A(s) \not\in \mathcal{I}.
\end{cases}
\label{eq:hbm12zjixhzy}
\end{align}
\end{lemma}

\begin{lemma}[{\cite[Lemma 14]{IEICE2023}}]
\label{lem:potential}
For any $F \in \mathscr{F}_{\reg}$, there exists a vector $(h_A)_{A \in [F]} \in \mathbb{R}^{|F|}$ satisfying
\begin{equation}
\label{eq:potential}
{\forall}A \in [F], L(F) = L_A(F) + \sum_{A' \in [F]}(h_{A'} - h_A)Q_{A, A'}(F).
\end{equation}
\end{lemma}
A vector $(h_A)_{A \in [F]} \in \mathbb{R}^{|F|}$ is satisfying (\ref{eq:potential}) is not unique.
For $F \in \mathscr{F}_{\reg}$, let $h(F)$ denote an arbitrarily chosen and fixed vector $(h_A)_{A \in [F]} \in \mathbb{R}^{|F|}$ satisfying (\ref{eq:potential}).

\begin{lemma}[{\cite[Lemma 15]{IEICE2023}}]
\label{lem:improve}
For any $F=(f, \trans), F'=(f', \trans') \in \mathscr{F}_{\reg}$ such that $[F] = [F']$,
if the following conditions (a) and (b) hold, then $L(F') \leq L(F)$.
\begin{enumerate}[(a)]
\item $L_A(F) = L_A(F')$ for any $A \in [F]$.
\item $h_{\trans_A(s)}(F) \geq h_{\trans'_A(s)}(F)$ for any $A \in [F]$ and $s \in \mathcal{S}$.
\end{enumerate}
\end{lemma}

\begin{lemma}
\label{lem:to-semi}
For any $F \in \mathscr{F}_{\reg}$, 
there exists $F^{\dag} = (f^{\dag}, \trans^{\dag}) \in \mathscr{F}_{\irr}$ satisfying the following conditions $(\mathrm{a})$--$(\mathrm{e})$.
\begin{itemize}
\item[$(\mathrm{a})$] $L(F^{\dag}) \leq L(F)$.
\item[$(\mathrm{b})$] $[F^{\dag}] \subseteq [F]$.
\item[$(\mathrm{c})$] for any $A \in [F^{\dag}]$, it holds that $f^{\dag}_{A} = f_{A}$.
\item[$(\mathrm{d})$] for any $A \in [F^{\dag}]$ and $s \in \mathcal{S}$, it holds that $\trans^{\dag}_{A}(s) \sim_k \trans_{A}(s)$.
\item[$(\mathrm{e})$] for any $A, A' \in [F^{\dag}]$, if $A \sim_k  A'$, then $A =  A'$.
\end{itemize}
\end{lemma}

The proof of Lemma \ref{lem:to-semi} is analogous to the proof of {\cite[Theorem 1]{IEICE2023}}.

\begin{proof}[Proof of Lemma \ref{lem:to-semi}]
We prove by induction on $|F|$.
By Lemma \ref{lem:irr}, there exists $\bar{F} = (\bar{f}, \bar{\trans}) \in \mathscr{F}_{\irr}$ satisfying the following conditions $(\bar{\mathrm{a}})$--$(\bar{\mathrm{d}})$.
\begin{itemize}
\item[$(\bar{\mathrm{a}})$] $L(\bar{F}) = L(F)$.
\item[$(\bar{\mathrm{b}})$] $[\bar{F}]  \subseteq [F]$.
\item[$(\bar{\mathrm{c}})$] For any $A \in [\bar{F}]$, it holds that $\bar{f}_{A} = f_{A}$.
\item[$(\bar{\mathrm{d}})$] For any $A \in [\bar{F}]$, it holds that $\bar{\trans}_{A} = \trans_{A}$.
\end{itemize}
Because $(\bar{\mathrm{a}})$--$(\bar{\mathrm{d}})$ imply that $\bar{F}$ satisfies Lemma \ref{lem:to-semi} (a)--(d),
if $\bar{F}$ also satisfies Lemma \ref{lem:to-semi} (e), then $F^{\dagger} \coloneqq \bar{F}$ is a desired code-tuple.
Thus, we now assume that $\bar{F}$ does not satisfy Lemma \ref{lem:to-semi} (e).
Then we can choose $B,  B' \in [\bar{F}]$ such that $ B \neq  B'$ and $ B \sim_k  B'$.
Also, by Lemma \ref{lem:potential}, there exists a real-valued vector $h(\bar{F})$ satisfying (\ref{eq:potential}) for $\bar{F}$.
We define $F'=(f', \trans') \in \mathscr{F}$ as
\begin{align}
[F'] &\coloneqq [\bar{F}], \label{eq:d5mcfcrebtl1}
\end{align}
and for $A \in [F'] = [\bar{F}]$ and $s \in \mathcal{S}$,
\begin{align}
f'_{A}(s) &\coloneqq \bar{f}_{A}(s), \label{eq:21l9xdppl0h1}\\
\trans'_{A}(s) &\coloneqq 
\begin{cases}
P  & \,\,\text{if}\,\, \bar{\trans}_ A(s) \in \mathcal{I},\\
\bar{\trans}_{A}(s) & \,\,\text{if}\,\, \bar{\trans}_ A(s) \not\in \mathcal{I},
\end{cases}
\label{eq:ynxygdkzsrk1}
\end{align}
where
\begin{equation}
\label{eq:pa234xqjzgvp}
\mathcal{I} \coloneqq \langle  B \rangle \cap [\bar{F}] = \langle  B' \rangle \cap [\bar{F}]
\end{equation}
and we choose 
\begin{equation}
\label{eq:tuhir8am5say}
P \in \arg\min_{A \in \mathcal{I}} h_A(\bar{F})
\end{equation}
arbitrarily.
Then we obtain $F' \in \mathscr{F}_{\reg}$ by applying Lemma \ref{lem:chooseone} since $\bar{F} \in \mathscr{F}_{\irr}$.
Also, we can see
\begin{equation}
\label{eq:hpvmxtw3pohj}
L(F') \leq L(\bar{F})
\end{equation}
by applying Lemma \ref{lem:improve} because $F'$ satisfies (a) (resp.~(b)) of Lemma \ref{lem:improve} by (\ref{eq:21l9xdppl0h1}) (resp.~(\ref{eq:ynxygdkzsrk1})--(\ref{eq:tuhir8am5say})).

By Lemma \ref{lem:irr}, there exists $\bar{F'} = (\bar{f'}, \bar{\trans'}) \in \mathscr{F}_{\irr}$ satisfying the following conditions $(\bar{\mathrm{a}}')$--$(\bar{\mathrm{d}}')$.
\begin{enumerate}
\item[$(\bar{\mathrm{a}}')$] $L(\bar{F'}) = L(F')$.
\item[$(\bar{\mathrm{b}}')$] $[\bar{F'}] = \kernel_{F'} \subseteq [F']$.
\item[$(\bar{\mathrm{c}}')$] For any $A \in [\bar{F'}]$, it holds that $\bar{f'}_{A} = f'_{A}$.
\item[$(\bar{\mathrm{d}}')$] For any $A \in [\bar{F'}]$, it holds that $\bar{\trans'}_{A} = \trans'_{A}$.
\end{enumerate}

We have
\begin{equation}
\label{eq:jmckcr9oz3s3}
\kernel_{F'} \subseteq [F'] \eqlab{A}\setminus (\mathcal{I} \setminus \{P\}) \eqlab{B}  \subseteq [F'] \setminus (\{B, B'\} \setminus \{P\}) \subsetneq [F'],
\end{equation}
where
(A) follows because $A \not\in \kernel_{F'}$ holds for any $A \in \mathcal{I} \setminus \{P\}$ by (\ref{eq:ynxygdkzsrk1}),
(B) follows from $\{B, B'\} \subseteq \mathcal{I}$ by (\ref{eq:pa234xqjzgvp}),
and (C) follows because $(\{B, B'\} \setminus \{P\}) \neq \emptyset$ by $B \neq B'$.
Hence, it follows that
\begin{equation}
[\bar{F'}]
\eqlab{A}= \kernel_{F'}
\eqlab{B}{\subsetneq} [F']
\eqlab{C}= [\bar{F}]
\eqlab{D}{\subseteq} [F], \label{eq:kel9gxkuamvj}
\end{equation}
where
(A) follows from $(\bar{\mathrm{b}}')$,
(B) follows from (\ref{eq:jmckcr9oz3s3}),
(C) follows from (\ref{eq:d5mcfcrebtl1}),
and (D) follows from $(\bar{\mathrm{b}})$.
In particular, $|\bar{F'}| < |F|$ holds, so that we can apply the induction hypothesis to $\bar{F'}$ to see that
 there exists $F^{\dagger} = (f^{\dag}, \trans^{\dag}) \in \mathscr{F}_{\irr}$ satisfying the following conditions ($\mathrm{a}^{\dagger}$)--($\mathrm{e}^{\dagger}$).
\begin{itemize}
\item[$(\mathrm{a}^{\dag})$] $L(F^{\dag}) \leq L(\bar{F'})$.
\item[$(\mathrm{b}^{\dag})$] $[F^{\dag}] \subseteq [\bar{F'}]$.
\item[$(\mathrm{c}^{\dag})$] for any $A \in [F^{\dag}]$, it holds that $f^{\dag}_{A} = \bar{f'}_{A}$.
\item[$(\mathrm{d}^{\dag})$] for any $A \in [F^{\dag}]$ and $s \in \mathcal{S}$, it holds that $\trans^{\dag}_{A}(s) \sim_k \bar{\trans'}_{A}(s)$.
\item[$(\mathrm{e}^{\dag})$] for any $A, A' \in [F^{\dag}]$, if $A \sim_k A'$, then $A = A'$.
\end{itemize}
We can see that $F^{\dagger}$ is a desired code-tuple, that is, $F^{\dagger}$ satisfies Lemma \ref{lem:to-semi} (a)--(e) as follows.

\begin{itemize}
\item (Lemma \ref{lem:to-semi} (a))
We have
\begin{equation*}
L(F^{\dagger}) \eqlab{A}{\leq} L(\bar{F'})
\eqlab{B}{=} L(F')
\eqlab{C}{\leq} L(\bar{F})
\eqlab{D}{=} L(F),
\end{equation*}
where
(A) follows from ($\mathrm{a}^{\dagger}$),
(B) follows from $(\bar{\mathrm{a}}')$,
(C) follows from (\ref{eq:hpvmxtw3pohj}),
and (D) follows from $(\bar{\mathrm{a}})$.

\item (Lemma \ref{lem:to-semi} (b))
We have
\begin{align*}
[F^{\dagger}]
\eqlab{A}\subseteq [\bar{F'}]
\eqlab{B}\subseteq [F],
\end{align*}
where
(A) follows from ($\mathrm{b}^{\dagger}$),
and (B) follows from (\ref{eq:kel9gxkuamvj}).

\item (Lemma \ref{lem:to-semi} (c))
For any $A \in [F^{\dag}]$, we have
\begin{equation*}
f^{\dag}_{A}
\eqlab{A}= \bar{f}'_{A}
\eqlab{B}= f'_{A}
\eqlab{C}= \bar{f}_{A}
\eqlab{D}= f_{A},
\end{equation*}
where
(A) follows from $(\mathrm{c}^{\dag})$,
(B) follows from $(\bar{\mathrm{c}}')$ and $[F^{\dag}] \subseteq [\bar{F}']$,
(C) follows from (\ref{eq:21l9xdppl0h1}) and $[F^{\dag}] \subseteq [F']$,
and (D) follows from $(\bar{\mathrm{c}})$ and $[F^{\dag}] \subseteq [\bar{F}]$.

\item (Lemma \ref{lem:to-semi} (d))
For any $A \in [F^{\dag}]$ and $s \in \mathcal{S}$, we have
\begin{equation*}
\trans^{\dag}_{A}(s)
\eqlab{A}{\sim_k} \bar{\trans}'_{A}(s)
\eqlab{B}= \trans'_{A}(s)
\eqlab{C}{\sim_k} \bar{\trans}_{A}(s)
\eqlab{D}= \trans_{A}(s),
\end{equation*}
where
(A) follows from $(\mathrm{d}^{\dag})$,
(B) follows from $(\bar{\mathrm{d}}')$ and $[F^{\dag}] \subseteq [\bar{F}']$,
(C) follows from (\ref{eq:ynxygdkzsrk1}) and $[F^{\dag}] \subseteq [F']$,
and (D) follows from $(\bar{\mathrm{d}})$ and $[F^{\dag}] \subseteq [\bar{F}]$.

\item (Lemma \ref{lem:to-semi} (e)) Directly from ($\mathrm{e}^{\dagger}$).
\end{itemize}
\end{proof}

\begin{proof}[Proof of Theorem \ref{thm:equiv} (i)]
By Lemma \ref{thm:differ}, we may assume that if $A \neq A'$, then $\PREF^k_{F, A} \neq \PREF^k_{F, A'}$ without loss of generality.
Thus, we can reassign indices to the code tables so that
\begin{equation}
\label{eq:bild3bgtj5ik}
\forall A \in [F], \PREF^k_{F, A} =  A
\end{equation}
by changing the domain $[F] \subseteq \mathscr{C}^k$ if necessary.

By Lemma \ref{lem:to-semi}, there exists $F^{\dag} \in \mathscr{F}_{\irr}$ satisfying Lemma \ref{lem:to-semi} (a)--(e).
By Lemma \ref{lem:to-semi} (e), we may assume that all $A \in [F^{\dag}]$ are representatives of the equivalence classes,
that is, 
\begin{equation}
\label{eq:0rq8izx3ph42}
\forall A \in [F^{\dag}], \bar{A} = A
\end{equation}
by changing the choice of the representatives if necessary.
Then for any $A \in [F^{\dag}]$ and $s \in \mathcal{S}$, we have
\begin{equation}
\label{eq:ygvo0pm55gdx}
\trans^{\dag}_{A}(s) \eqlab{A}{=} \overline{\trans^{\dag}_{A}(s)} \eqlab{B}{=} \overline{\trans_{A}(s)},
\end{equation}
where 
(A) follows from (\ref{eq:0rq8izx3ph42}),
and (B) follows from Lemma \ref{lem:to-semi} (d).

Define an RCT $\tilde{F}=(\tilde{f}, \tilde{\trans}) \in \tilde{\mathscr{F}}$ as
\begin{eqnarray}
[\tilde{F}] &\coloneqq& [F^{\dag}] \label{eq:4fhbihitwkg5} 
\end{eqnarray}
and for $A \in [\tilde{F}] = [F^{\dag}]$ and $s \in \mathcal{S}$,
\begin{eqnarray}
\tilde{f}_{A}(s) &\coloneqq& f^{\dag}_{A}(s) \text{ }( = f_{A}(s) \text{ by Lemma \ref{lem:to-semi} (c)}), \label{eq:6ije1hv7t3po} \\
\tilde{\trans}_{A}(s) &\coloneqq& \trans_{A}(s). \label{eq:3r0fsh7bq6yp}
\end{eqnarray}
Note that $\tilde{\trans}_{A}$ is indeed a mapping conforming to Definition \ref{def:semi} (iii) because
for any $A \in [\tilde{F}]$ and $s \in \mathcal{S}$, we have
\begin{equation*}
\tilde{\trans}_{A}(s)
\eqlab{A}{=} \trans_{A}(s)
\in \langle \overline{\trans_{A}(s)} \rangle
\eqlab{B}{=} \langle \trans^{\dag}_{A}(s) \rangle
\subseteq \langle [F^{\dag}] \rangle
\eqlab{C}{=} \langle [\tilde{F}] \rangle,
\end{equation*}
where
(A) follows from (\ref{eq:3r0fsh7bq6yp}),
(B) follows from (\ref{eq:ygvo0pm55gdx}),
and (C) follows from (\ref{eq:4fhbihitwkg5}).

We show that $\tilde{F}$ is a desired RCT, that is, $\tilde{F} \in \tilde{\mathscr{F}}_{\comp} \cap \tilde{\mathscr{F}}_{\reg} \cap \tilde{\mathscr{F}}_{\ext} \cap \tilde{\mathscr{F}}_{k\hdec}$ and $\tilde{L}(\tilde{F}) \leq L(F)$.

\begin{itemize}
\item (Proof of $\tilde{F} \in \tilde{\mathscr{F}}_{\reg}$ and $\tilde{L}(\tilde{F}) \leq L(F)$)
The direct realization $F' = (f', \trans')$ of $\tilde{F}$ is identical to $F^{\dag}$ as shown as
\begin{equation}
f'_A(s) \eqlab{A}{=} \tilde{f}_{A}(s) \eqlab{B}{=} f^{\dag}_{A}(s), \quad
\trans'_A(s) \eqlab{C}{=} \overline{\tilde{\trans}_{A}(s)}
\eqlab{D}{=} \overline{\trans_{A}(s)}
\eqlab{E}{=} \trans^{\dag}_{A}(s)
\end{equation}
for any $A \in [\tilde{F}] = [F^{\dag}]$ and $s \in \mathcal{S}$,
where
(A) follows from (\ref{eq:z3m658av3wl7}),
(B) follows from (\ref{eq:6ije1hv7t3po}),
(C) follows from (\ref{eq:qwshe4o3q759}),
(D) follows from (\ref{eq:3r0fsh7bq6yp}),
and (E) follows from (\ref{eq:ygvo0pm55gdx}).
Since $F' = F^{\dag} \in \mathscr{F}_{\reg}$, we have $\tilde{F} \in \tilde{\mathscr{F}}_{\reg}$
and
\begin{equation*}
\tilde{L}(\tilde{F}) = L(F') = L(F^{\dag}) \eqlab{A}\leq L(F)
\end{equation*}
as desired, where (A) follows from Lemma \ref{lem:to-semi} (a).

\item (Proof of $\tilde{F} \in \tilde{\mathscr{F}}_{\ext}$)
It holds that
\begin{equation}
\label{eq:7id2oaebdaeh}
\forall A \in [\tilde{F}], 
 A \in [\tilde{F}]
\eqlab{A}{=} [F^{\dag}]
\eqlab{B}{\subseteq} [F],
\end{equation}
where
(A) follows from (\ref{eq:4fhbihitwkg5}),
and (B) follows from Lemma \ref{lem:to-semi} (b).
Hence, for any $A \in [\tilde{F}]$, we have
\begin{equation*}
 A \eqlab{A}{=} \PREF^k_{F, A}
\eqlab{B}{\neq} \emptyset,
\end{equation*}
where
(A) follows from (\ref{eq:bild3bgtj5ik}) and (\ref{eq:7id2oaebdaeh}),
and (B) follows because there exists $\pmb{x} \in \mathcal{S}^{\ast}$ such that $|f^{\ast}_A(\pmb{x})| \geq k$ by $F \in \mathscr{F}_{\ext}$ and Lemma \ref{lem:F_ext}.
This concludes that $[\tilde{F}] \not\owns \emptyset$.

\item (Proof of $\tilde{F} \in \tilde{\mathscr{F}}_{\comp}$)
For any $A \in [\tilde{F}]$, we have
\begin{equation*}
\label{eq:burwu8k4nh3r}
\PREF^k_{\tilde{F}, A}
\eqlab{A}{=} \bigcup_{s \in \mathcal{S}} \left[ \tilde{f}_{A}(s) \tilde{\trans}_{A}(s) \right]_k
\eqlab{B}{=} \bigcup_{s \in \mathcal{S}} \left[ f_{A}(s) \trans_{A}(s) \right]_k
\eqlab{C}{=} \bigcup_{s \in \mathcal{S}} \left[ f_{A}(s) \PREF^k_{F, \trans_{A}(s)} \right]_k
\eqlab{D}{=} \PREF^k_{F, A}
\eqlab{E}{=}  A,
\end{equation*}
where
(A) follows from (\ref{eq:x425gfz5autj}),
(B) follows from (\ref{eq:6ije1hv7t3po}) and (\ref{eq:3r0fsh7bq6yp}),
(C) follows from (\ref{eq:bild3bgtj5ik}) and (\ref{eq:7id2oaebdaeh}),
(D) follows from (\ref{eq:9fhqriquk4ht}),
and (E) follows from (\ref{eq:bild3bgtj5ik}) and (\ref{eq:7id2oaebdaeh}).

\item (Proof of $\tilde{F} \in \tilde{\mathscr{F}}_{k\hdec}$)
We show that $\tilde{F}$ satisfies the conditions (a) and (b) of Definition \ref{def:semi-class} (iv).

\begin{itemize}
\item (Definition \ref{def:semi-class} (iv) (a))
For any $A \in [\tilde{F}]$ and $\pmb{b} \in \mathcal{C}^{\ast}$, we have
\begin{eqnarray}
\bar{\PREF}^k_{\tilde{F}, A}(\pmb{b})
&\eqlab{A}{=}& \bigcup_{\substack{s \in \mathcal{S},\\ \tilde{f}_{A}(s) \succ \tilde{f}_{A}(s)}} \left[ \pmb{b}^{-1}\tilde{f}_{A}(s) \tilde{\trans}_{A}(s) \right]_k \nonumber\\
&\eqlab{B}{=}& \bigcup_{\substack{s \in \mathcal{S},\\ f_{A}(s) \succ f_{A}(s)}} \left[ \pmb{b}^{-1}f_{A}(s) \trans_{A}(s) \right]_k \nonumber\\
&\eqlab{C}{=}& \bigcup_{\substack{s \in \mathcal{S},\\ f_{A}(s) \succ f_{A}(s)}} \left[ \pmb{b}^{-1}f_{A}(s) \PREF^k_{F, \trans_{A}(s)} \right]_k \nonumber\\
&\eqlab{D}{=}& \bar{\PREF}^k_{F, A}(\pmb{b}), \label{eq:p90pugcr4jbq}
\end{eqnarray}
where
(A) follows from (\ref{eq:hh4v6kmbk61n}),
(B) follows from (\ref{eq:6ije1hv7t3po}) and (\ref{eq:3r0fsh7bq6yp}),
(C) follows from (\ref{eq:bild3bgtj5ik}) and (\ref{eq:7id2oaebdaeh}),
and (D) follows from (\ref{eq:gfo3oig06zin}).
Hence, for any $A \in [\tilde{F}]$ and $s \in \mathcal{S}$, we have
\begin{equation*}
\tilde{\trans}_{A}(s) \cap \bar{\PREF}^k_{\tilde{F}, A}(\tilde{f}_{A}(s))
\eqlab{A}{=} \tilde{\trans}_{A}(s) \cap \bar{\PREF}^k_{F, A}(\tilde{f}_{A}(s))
\eqlab{B}{=} \trans_{A}(s) \cap \bar{\PREF}^k_{F, A}(f_{A}(s))
\eqlab{C}{=} \PREF^k_{F, \trans_{A}(s)} \cap\bar{\PREF}^k_{F, A}(f_{A}(s))
\eqlab{D}{=} \emptyset,
\end{equation*}
where
(A) follows from (\ref{eq:p90pugcr4jbq}),
(B) follows from (\ref{eq:6ije1hv7t3po}) and (\ref{eq:3r0fsh7bq6yp}),
(C) follows from (\ref{eq:bild3bgtj5ik}) and (\ref{eq:7id2oaebdaeh}),
and (D) follows from $F \in \mathscr{F}_{k\hdec}$.

\item
(Definition \ref{def:semi-class} (iv) (b))
Choose $A \in [\tilde{F}]$ and $s, s' \in \mathcal{S}$ such that $\tilde{f}_{A}(s) = \tilde{f}_{A}(s')$ arbitrarily.
Then we have
\begin{equation}
\label{eq:y6cbq22ti6sg}
f_{A}(s)
\eqlab{A}{=} \tilde{f}_{A}(s)
= \tilde{f}_{A}(s')
\eqlab{B}{=} f_{A}(s'),
\end{equation}
where
(A) follows from (\ref{eq:6ije1hv7t3po}),
and (B) follows from (\ref{eq:6ije1hv7t3po}).
Therefore, we obtain
\begin{equation*}
\tilde{\trans}_{A}(s) \cap \tilde{\trans}_{A}(s')
\eqlab{A}{=} \trans_{A}(s) \cap \trans_{A}(s')
\eqlab{B}{=} \PREF^k_{F, \trans_{A}(s)} \cap \PREF^k_{F, \trans_{A}(s')}
\eqlab{C}{=} \emptyset,
\end{equation*}
where
(A) follows from (\ref{eq:3r0fsh7bq6yp}),
(B) follows from (\ref{eq:bild3bgtj5ik}),
and (C) follows from $F \in \mathscr{F}_{k\hdec}$ and (\ref{eq:y6cbq22ti6sg}).
\end{itemize}
\end{itemize}
\end{proof}

\subsection{Proof of Lemma \ref{lem:expand}: $L(F) = \tilde{L}(\tilde{F})$}
\label{subsec:proof-expand1}

The proof relies on \cite[Lemma 7]{IEICE2023}.

Let $F'=(f', \trans')$ be the direct realization of $\tilde{F}$, and we define a mapping $\varphi \colon [\hat{F}] \to [F']$ as
\begin{equation}
\label{eq:rg9zmfsn1btj}
\varphi((A, \phi)) = A
\end{equation}
for $(A, \phi) \in [\hat{F}]$.
Then for any $(A, \phi) \in [\hat{F}]$ and $s \in \mathcal{S}$, we have
\begin{equation}
\label{eq:12yxaf48at5p}
\varphi(\trans_{(A, \phi)}(s))
\eqlab{A}{=} \varphi \left( \left(\overline{\tilde{\trans}_{A}(s)},  \quot{\phi}{\tilde{f}_{A}(s)} \circ \psi_{A, s} \right) \right)
\eqlab{B}{=} \overline{\tilde{\trans}_{A}(s)}
\eqlab{C}{=} \trans'_{A}(s)
\eqlab{D}{=} \trans'_{\varphi((A, \phi))}(s),
\end{equation}
where
(A) follows from (\ref{eq:otmlflag72l3}),
(B) follows from (\ref{eq:rg9zmfsn1btj}),
(C) follows from (\ref{eq:qwshe4o3q759}),
and (D) follows from (\ref{eq:rg9zmfsn1btj}).
Also, for any $(A, \phi) \in [\hat{F}]$ and $s \in \mathcal{S}$, we have
\begin{equation}
\label{eq:ohfzbcgbftvx}
|\hat{f}_{(A, \phi)}(s)|
\eqlab{A}{=} |\phi(\tilde{f}_{A}(s))|
\eqlab{B}{=} |\tilde{f}_{A}(s)|
\eqlab{C}{=} |f'_{A}(s)|
\eqlab{D}{=} |f'_{\varphi((A, \phi))}(s)|,
\end{equation}
where
(A) follows from (\ref{eq:9tr7le1wumbs}),
(B) follows from Definition \ref{def:Phi} (a),
(C) follows from (\ref{eq:z3m658av3wl7}),
and (D) follows from (\ref{eq:rg9zmfsn1btj}).

Equations (\ref{eq:12yxaf48at5p}) and (\ref{eq:ohfzbcgbftvx}) are analogous to  \cite[Eqs.~(23) and (24)]{IEICE2023}.
Noting that the restriction $\varphi|_{[F]} \colon [F] \to [F']$ of $\varphi$ to $[F]$ also satisfies (\ref{eq:12yxaf48at5p}) and (\ref{eq:ohfzbcgbftvx}),
we can show $L(F) = L(F')$ by the almost identical argument to the proofs of \cite[Lemma 7 (iii) and (iv)]{IEICE2023} from \cite[Eqs.~(23) and (24)]{IEICE2023}.

\subsection{Proof of Lemma \ref{lem:expand}: $F \in \mathscr{F}_{k\hdec}$}
\label{subsec:proof-expand2}

We first prove the following Lemma \ref{lem:expanded-pref}.

\begin{lemma}
\label{lem:expanded-pref}
For any $(A, \phi) \in [F]$, we have $\PREF^k_{F, (A, \phi)} \subseteq \phi(A)$.
\end{lemma}

\begin{proof}[Proof of Lemma \ref{lem:expanded-pref}]
It suffices to prove
\begin{equation}
\label{eq:ifd5otujq7t0}
{\forall}(\pmb{x}, (A, \phi), k') \in \mathcal{S}^{\ast} \times [F] \times \{0, 1, 2, \ldots, k\},
\left( |f^{\ast}_{(A, \phi)}(\pmb{x})| \geq k' \implies \left[ f^{\ast}_{(A, \phi)}(\pmb{x}) \right]_{k'} \in \left[ \phi(A) \right]_{k'} \right)
\end{equation}
because
\begin{eqnarray*}
\PREF^k_{F, (A, \phi)}
\eqlab{A}{=} \{c \in \mathcal{C}^k : \pmb{x} \in \mathcal{S}^{\ast}, f^{\ast}_{(A, \phi)}(\pmb{x}) \succeq \pmb{c}\}
= \left\{\left[f^{\ast}_{(A, \phi)}(\pmb{x}) \right]_k \colon \pmb{x} \in \mathcal{S}^{\ast}, |f^{\ast}_{(A, \phi)}(\pmb{x})| \geq k \right\}
\eqlab{B}{\subseteq} \left[ \phi(A) \right]_k
= \phi(A)
\end{eqnarray*}
as desired, where
(A) follows from (\ref{eq:pref3}),
and (B) follows from (\ref{eq:ifd5otujq7t0}).
Thus, we prove (\ref{eq:ifd5otujq7t0}) by induction on $|\pmb{x}|$.
Choose  $(\pmb{x}, (A, \phi), k') \in \mathcal{S}^{\ast} \times [F] \times \{0, 1, 2, \ldots, k\}$ such that
\begin{equation}
\label{eq:suol5hppo8wa}
|f^{\ast}_{(A, \phi)}(\pmb{x})| \geq k'
\end{equation}
arbitrarily.

We first consider the base case $|\pmb{x}| = 0$.
This case is possible only if $|f^{\ast}_{(A, \phi)}(\pmb{x})| = k' = 0$ by (\ref{eq:suol5hppo8wa}).
Therefore, we have
\begin{equation*}
\left[ f^{\ast}_{(A, \phi)}(\pmb{x}) \right]_{k'}
= \left[ f^{\ast}_{(A, \phi)}(\pmb{x}) \right]_0
= \lambda
\in \{\lambda\}
\eqlab{A}= \left[\phi(A)\right]_0
= \left[\phi(A)\right]_{k'},
\end{equation*}
where
(A) is justified as follows:
by $\tilde{F} \in \tilde{\mathscr{F}}_{\ext}$, we have $\emptyset \not\in [\tilde{F}]$, which implies $\emptyset \not\in \langle [\tilde{F}] \rangle$;
on the other hand, $A \in [\tilde{F}]$ implies $\phi(A) \in \langle [\tilde{F}] \rangle$; hence, we have $\phi(A) \neq \emptyset$ and thus $\left[\phi(A)\right]_0 \neq \emptyset$.
This concludes that the assertion (\ref{eq:ifd5otujq7t0}) holds for the case $|\pmb{x}| = 0$.

We next consider the induction step for $|\pmb{x}| \geq 1$.
We first show 
\begin{equation}
\label{eq:ngwm3wseb9wy}
\left[ f^{\ast}_{(A, \phi)}(\pmb{x}) \right]_{k'} \in \left[ f_{(A, \phi)}(x_1)\quot{\phi}{\tilde{f}_{A}(x_1)}(\tilde{\trans}_{A}(x_1)) \right]_{k'}
\end{equation}
dividing into the following two cases: the case $|f_{(A, \phi)}(x_1)| \geq k'$ and the case $|f_{(A, \phi)}(x_1)| < k'$.

\begin{itemize}
\item The case $|f_{(A, \phi)}(x_1)| \geq k'$:
We have
\begin{equation*}
\left[f^{\ast}_{(A, \phi)}(\pmb{x}) \right]_{k'}
= \left[ f_{(A, \phi)}(x_1) \right]_{k'}
\eqlab{A}{\in} \left[f_{(A, \phi)}(x_1)\quot{\phi}{\tilde{f}_{A}(x_1)}(\tilde{\trans}_{A}(x_1)) \right]_{k'},
\end{equation*}
where
(A) follows from $\tilde{\trans}_{A}(x_1) \neq \emptyset$ by $\tilde{F} \in \mathscr{F}_{\ext}$.

\item The case $|f_{(A, \phi)}(x_1)| < k'$:
We have
\begin{eqnarray*}
\left[f^{\ast}_{(A, \phi)}(\pmb{x}) \right]_{k'}
&\eqlab{A}=& \left[ f_{(A, \phi)}(x_1)f^{\ast}_{\trans_{(A, \phi)}(x_1)}(\suff(\pmb{x})) \right]_{k'} \\
&\eqlab{B}{=} & f_{(A, \phi)}(x_1)\left[ f^{\ast}_{\trans_{(A, \phi)}(x_1)}(\suff(\pmb{x})) \right]_{k'-|f_{(A, \phi)}(x_1)|}\\
&\eqlab{C}{=} & f_{(A, \phi)}(x_1)\left[ f^{\ast}_{\left(\overline{\tilde{\trans}_{A}(x_1)}, \quot{\phi}{\tilde{f}_{A}(x_1)} \circ \psi_{A, x_1}\right)}(\suff(\pmb{x})) \right]_{k'-|f_{(A, \phi)}(x_1)|}\\
&\eqlab{D}{\in}& f_{(A, \phi)}(x_1)\left[ \quot{\phi}{\tilde{f}_{A}(x_1)} \circ \psi_{A, x_1}(\overline{\tilde{\trans}_{A}(x_1)}) \right]_{k'-|f_{(A, \phi)}(x_1)|}\\
&\eqlab{E}=& f_{(A, \phi)}(x_1)\left[ \quot{\phi}{\tilde{f}_{A}(x_1)} (\tilde{\trans}_{A}(x_1)) \right]_{k'-|f_{(A, \phi)}(x_1)|}\\
&=& \left[ f_{(A, \phi)}(x_1)\quot{\phi}{\tilde{f}_{A}(x_1)} (\tilde{\trans}_{A}(x_1)) \right]_{k'},\\
\end{eqnarray*}
where
(A) follows from (\ref{eq:fstar}),
(B) follows from $|f_{(A, \phi)}(x_1)| < k'$,
(C) follows from (\ref{eq:otmlflag72l3}),
(D) follows from the induction hypothesis since
\begin{equation*}
\left|f^{\ast}_{\left(\overline{\tilde{\trans}_{A}(x_1)}, \quot{\phi}{\tilde{f}_{A}(x_1)} \circ \psi_{A, x_1}\right)}(\suff(\pmb{x}))\right| \geq k' - |f_{(A, \phi)}(x_1)|,
\end{equation*}
and (E) follows from (\ref{eq:ks8m7ci0bnrg}).
\end{itemize}

From the above, (\ref{eq:ngwm3wseb9wy}) has been proven.
Therefore, we obtain
\begin{eqnarray*}
\left[ f^{\ast}_{(A, \phi)}(\pmb{x}) \right]_{k'}
&\eqlab{A}{\in}& \left[ f_{(A, \phi)}(x_1)\quot{\phi}{\tilde{f}_{A}(x_1)}(\tilde{\trans}_{A}(x_1)) \right]_{k'}\\
&\eqlab{B}{=}& \left[ \phi(\tilde{f}_{A}(x_1)) \quot{\phi}{\tilde{f}_{A}(x_1)}(\tilde{\trans}_{A}(x_1))  \right]_{k'} \\
&\eqlab{C}{=}&  \left[ \phi(\tilde{f}_{A}(x_1)\tilde{\trans}_{A}(x_1))  \right]_{k'} \\
&\subseteq& \left[ \phi \left( \bigcup_{s \in \mathcal{S}} \tilde{f}_{A}(s)\tilde{\trans}_{A}(s) \right)  \right]_{k'} \\
&\eqlab{D}{=}&  \phi \left(\left[ \bigcup_{s \in \mathcal{S}} \tilde{f}_{A}(s)\tilde{\trans}_{A}(s) \right]_{k'}  \right)  \\
&=& \phi \left( \left[ \bigcup_{s \in \mathcal{S}} \left[ \tilde{f}_{A}(s)\tilde{\trans}_{A}(s) \right]_{k} \right]_{k'} \right)   \\
&\eqlab{E}{=}& \phi \left( \left[ \PREF^{k}_{\tilde{F}, A} \right]_{k'} \right)   \\
&\eqlab{F}{=}& \left[ \phi \left( \PREF^{k}_{\tilde{F}, A} \right) \right]_{k'}  \\
&\eqlab{G}{=}& \left[ \phi (A) \right]_{k'}
\end{eqnarray*}
as desired, where
(A) follows from (\ref{eq:ngwm3wseb9wy}),
(B) follows from (\ref{eq:9tr7le1wumbs}),
(C) follows from Lemma \ref{lem:phi-quot} (i),
(D) follows from (\ref{eq:s9qdkaq7m7ap}),
(E) follows from (\ref{eq:x425gfz5autj}),
(F) follows from (\ref{eq:s9qdkaq7m7ap}),
and (G) follows from $\tilde{F} \in \tilde{\mathscr{F}}_{\comp}$.
\end{proof}

\begin{corollary}
\label{cor:expanded-pref}
For any $(A, \phi) \in [F]$ and $s \in \mathcal{S}$,
we have $\PREF^k_{F, \trans_{(A, \phi)}(s)} \subseteq  \quot{\phi}{\tilde{f}_{A}(s')}(\tilde{\trans}_{A}(s'))$.
\end{corollary}

\begin{proof}[Proof of Corollary \ref{cor:expanded-pref}]
We have
\begin{equation*}
\PREF^k_{F, \trans_{(A, \phi)}(s)}
\eqlab{A}= \PREF^k_{F, \left( \overline{\tilde{\trans}_A(s)}, \quot{\phi}{\tilde{f}_{A}(s)} \circ \psi_{A, s}\right)}
\eqlab{B}\subseteq \quot{\phi}{\tilde{f}_{A}(s)} \circ \psi_{A, s}(\overline{\tilde{\trans}_{A}(s)})
\eqlab{C}= \quot{\phi}{\tilde{f}_{A}(s)}(\tilde{\trans}_{A}(s)),
\end{equation*}
where
(A) follows from (\ref{eq:1aj44hktwb5i}),
(B) follows from Lemma \ref{lem:expanded-pref},
and (C) follows from (\ref{eq:ks8m7ci0bnrg}).
\end{proof}

Now, to prove $F \in \mathscr{F}_{k\hdec}$,
we show that $F$ satisfies Definition \ref{def:k-bitdelay} (a) and (b) below.

(Definition \ref{def:k-bitdelay} (a))
Choose $(A, \phi) \in [F]$ and $s \in \mathcal{S}$ arbitrarily.
We have
\begin{eqnarray}
\bar{\PREF}^k_{F, \trans_{(A, \phi)}}(f_{(A, \phi)}(s))
&\eqlab{A}{=}& \bigcup_{\substack{s' \in \mathcal{S},\\ f_{(A, \phi)}(s') \succ f_{(A, \phi)}(s)}} \left[ f_{(A, \phi)}(s)^{-1} f_{(A, \phi)}(s') \PREF^k_{F, \trans_{(A, \phi)}(s')} \right]_k \nonumber\\
&\eqlab{B}{=}& \bigcup_{\substack{s' \in \mathcal{S},\\ \phi(\tilde{f}_{A}(s')) \succ \phi(\tilde{f}_{A}(s))}} \left[ \phi(\tilde{f}_{A}(s))^{-1} \phi(\tilde{f}_{A}(s')) \PREF^k_{F, \trans_{(A, \phi)}(s')} \right]_k \nonumber\\
&\eqlab{C}{=}& \bigcup_{\substack{s' \in \mathcal{S},\\ \tilde{f}_{A}(s')\succ \tilde{f}_{A}(s)}} \left[ \phi(\tilde{f}_{A}(s))^{-1} \phi(\tilde{f}_{A}(s')) \PREF^k_{F, \trans_{(A, \phi)}(s')}  \right]_k \nonumber\\
&\eqlab{D}{\subseteq}& \bigcup_{\substack{s' \in \mathcal{S},\\ \tilde{f}_{A}(s')\succ \tilde{f}_{A}(s)}} \left[ \phi(\tilde{f}_{A}(s))^{-1} \phi(\tilde{f}_{A}(s')) \quot{\phi}{\tilde{f}_{A}(s')}(\tilde{\trans}_{A}(s')) \right]_k \nonumber\\
&\eqlab{E}{=}& \bigcup_{\substack{s' \in \mathcal{S},\\ \tilde{f}_{A}(s')\succ \tilde{f}_{A}(s)}} \left[ \phi(\tilde{f}_{A}(s))^{-1}  \phi(\tilde{f}_{A}(s')\tilde{\trans}_{A}(s')) \right]_k \nonumber\\
&\eqlab{F}{=}& \bigcup_{\substack{s' \in \mathcal{S},\\ \tilde{f}_{A}(s')\succ \tilde{f}_{A}(s)}} \left[ \phi(\tilde{f}_{A}(s))^{-1}  \phi(\tilde{f}_{A}(s)\tilde{f}_{A}(s)^{-1}\tilde{f}_{A}(s')\tilde{\trans}_{A}(s')) \right]_k \nonumber\\
&\eqlab{G}{=}& \bigcup_{\substack{s' \in \mathcal{S},\\ \tilde{f}_{A}(s')\succ \tilde{f}_{A}(s)}} \left[ \quot{\phi}{\tilde{f}_{A}(s)} (\tilde{f}_{A}(s)^{-1}\tilde{f}_{A}(s')\tilde{\trans}_{A}(s')) \right]_k \nonumber\\
&\eqlab{H}{=}& \bigcup_{\substack{s' \in \mathcal{S},\\ \tilde{f}_{A}(s')\succ \tilde{f}_{A}(s)}} \quot{\phi}{\tilde{f}_{A}(s)} \left(\left[ \tilde{f}_{A}(s)^{-1}\tilde{f}_{A}(s')\tilde{\trans}_{A}(s') \right]_k \right) \nonumber\\
&\eqlab{I}{=}& \quot{\phi}{\tilde{f}_{A}(s)} \left( \bigcup_{\substack{s' \in \mathcal{S},\\ \tilde{f}_{A}(s')\succ \tilde{f}_{A}(s)}} \left[ \tilde{f}_{A}(s)^{-1}\tilde{f}_{A}(s')\tilde{\trans}_{A}(s') \right]_k \right) \nonumber\\
&\eqlab{J}{=}& \quot{\phi}{\tilde{f}_{A}(s)} \tilde{\PREF}^k_{\tilde{F}, A}(\tilde{f}_{A}(s)),  \label{eq:fr12poy890jk}
\end{eqnarray}
where
(A) follows from (\ref{eq:gfo3oig06zin}),
(B) follows from (\ref{eq:9tr7le1wumbs}),
(C) follows from Definition \ref{def:Phi} (b),
(D) follows from Corollary \ref{cor:expanded-pref},
(E) follows from Lemma \ref{lem:phi-quot} (i),
(F) follows from $\tilde{f}_{A}(s')\succ \tilde{f}_{A}(s)$,
(G) follows from (\ref{eq:r6xuob4upqw5}),
(H) follows from (\ref{eq:s9qdkaq7m7ap}),
(I) follows from Lemma \ref{lem:phi-bijection},
and (J) follows from (\ref{eq:hh4v6kmbk61n}).
Hence, we obtain
\begin{eqnarray*}
\PREF^k_{F, \trans_{(A, \phi)}(s)} \cap \bar{\PREF}^k_{F, \trans_{(A, \phi)}}(f_{(A, \phi)}(s))
&\eqlab{A}{\subseteq}& \quot{\phi}{\tilde{f}_{A}(s)}(\tilde{\trans}_{A}(s)) \cap \bar{\PREF}^k_{F, \trans_{(A, \phi)}}(f_{(A, \phi)}(s)) \\
&\eqlab{B}{\subseteq}& \quot{\phi}{\tilde{f}_{A}(s)}(\tilde{\trans}_{A}(s)) \cap \quot{\phi}{\tilde{f}_{A}(s)} \bar{\PREF}^k_{\tilde{F}, A}(\tilde{f}_{A}(s))\\
&\eqlab{C}{=}& \quot{\phi}{\tilde{f}_{A}(s)}(\tilde{\trans}_{A}(s) \cap \bar{\PREF}^k_{\tilde{F}, A}(\tilde{f}_{A}(s)))\\
&\eqlab{D}{=}& \quot{\phi}{\tilde{f}_{A}(s)}(\emptyset)\\
&=& \emptyset
 \end{eqnarray*}
as desired, where
(A) follows from Corollary \ref{cor:expanded-pref},
(B) follows from (\ref{eq:fr12poy890jk}),
(C) follows from Lemma \ref{lem:phi-bijection},
and (D) follows from Definition \ref{def:semi-class} (iv) (a).

(Definition \ref{def:k-bitdelay} (b))
Choose $(A, \phi) \in [F]$ and $s, s' \in \mathcal{S}$ such that $s \neq s'$ and $f_{(A, \phi)}(s) = f_{(A, \phi)}(s')$ arbitrarily.
Then we have
\begin{equation*}
\phi(\tilde{f}_{A}(s))
\eqlab{A}{=} f_{(A, \phi)}(s) 
= f_{(A, \phi)}(s')
\eqlab{B}{=} \phi(\tilde{f}_{A}(s')),
\end{equation*}
where
(A) follows from (\ref{eq:9tr7le1wumbs}),
and (B) follows from (\ref{eq:9tr7le1wumbs}).
Hence, we have
\begin{equation}
\label{eq:8he9ypp74cyb}
\tilde{f}_{A}(s) = \tilde{f}_{A}(s')
\end{equation}
by Lemma \ref{lem:phi-bijection}.
Therefore, we obtain
\begin{eqnarray*}
\PREF^k_{F, \trans_{(A, \phi)}(s)} \cap \PREF^k_{F, \trans_{(A, \phi)}(s')}  
&\eqlab{A}{\subseteq}& \quot{\phi}{\tilde{f}_{A}(s)}(\tilde{\trans}_{A}(s)) \cap \quot{\phi}{\tilde{f}_{A}(s')}(\tilde{\trans}_{A}(s')) \\
&\eqlab{B}{=}& \quot{\phi}{\tilde{f}_{A}(s)}(\tilde{\trans}_{A}(s)) \cap \quot{\phi}{\tilde{f}_{A}(s)}(\tilde{\trans}_{A}(s')) \\
&\eqlab{C}{=}& \quot{\phi}{\tilde{f}_{A}(s)}(\tilde{\trans}_{A}(s) \cap \tilde{\trans}_{A}(s')) \\
&\eqlab{D}{=}&\quot{\phi}{\tilde{f}_{A}(s)}(\emptyset) \\
&=& \emptyset,
\end{eqnarray*}
as desired, where
(A) follows from Corollary \ref{cor:expanded-pref},
(B) follows from (\ref{eq:8he9ypp74cyb}),
(C) follows Lemma \ref{lem:phi-bijection},
and (D) follows from Definition \ref{def:semi-class} (iv) (b) and (\ref{eq:8he9ypp74cyb}).

\subsection{Proof of Lemma \ref{lem:expand}: $F \in \mathscr{F}_{\ext}$}
\label{subsec:proof-expand3}

We choose
\begin{equation}
\label{eq:plg10theuj7z}
(A, \phi) \in \argmin_{(A', \phi') \in \kernel_F} |A'|
\end{equation}
arbitrarily.
We first show that there exists at most one $s \in \mathcal{S}$ such that $f_{(A, \phi)}(s) = \lambda$ by contradiction.
Assume that there exist $s, s' \in \mathcal{S}$ such that $s \neq s'$ and $f_{(A, \phi)}(s) = f_{(A, \phi)}(s') = \lambda$.
Then we have
\begin{equation}
\label{eq:vgeb48eh7caw}
\tilde{f}_{A}(s) = \tilde{f}_{A}(s') = \lambda
\end{equation}
because
\begin{equation*}
|\tilde{f}_{A}(s)| \eqlab{A}{=} |\phi(\tilde{f}_{A}(s))| \eqlab{B}{=} |f_{(A, \phi)}(s)| = 0,\quad
|\tilde{f}_{A}(s')| \eqlab{A}{=} |\phi(\tilde{f}_{A}(s'))| \eqlab{B}{=} |f_{(A, \phi)}(s')| = 0,
\end{equation*}
where
(A) follows from Definition \ref{def:Phi} (a),
and (B) follows from (\ref{eq:9tr7le1wumbs}).
Thus, we have
\begin{eqnarray}
|\overline{\tilde{\trans}_{A}(s)}| + |\overline{\tilde{\trans}_{A}(s')}|
&\eqlab{A}=& |\tilde{\trans}_{A}(s)| + |\tilde{\trans}_{A}(s')| \nonumber\\
&\eqlab{B}{=}& |\tilde{\trans}_{A}(s) \cup \tilde{\trans}_{A}(s')| \nonumber\\
&=& \left|\left[\tilde{\trans}_{A}(s)\right]_k \cup \left[\tilde{\trans}_{A}(s')\right]_k\right| \nonumber\\
&\eqlab{C}{=}& \left|\left[\tilde{f}_{A}(s)\tilde{\trans}_{A}(s)\right]_k \cup \left[\tilde{f}_{A}(s')\tilde{\trans}_{A}(s')\right]_k\right| \nonumber\\
&\leq& \left| \bigcup_{s'' \in \mathcal{S}} \left[\tilde{f}_{A}(s'')\tilde{\trans}_{A}(s'')\right]_k \right| \nonumber\\
&\eqlab{D}{=}& |\PREF^k_{\tilde{F}, A}| \nonumber\\
&\eqlab{E}{=}& |A|, \label{eq:3vxzhkr8quvj}
\end{eqnarray}
where
(A) follows from (\ref{eq:9k4gm8wmk5qb}),
(B) follows from Definition \ref{def:semi-class} (iv) (b) and (\ref{eq:vgeb48eh7caw}),
(C) follows from (\ref{eq:vgeb48eh7caw}),
(D) follows from (\ref{eq:x425gfz5autj}),
and (E) follows from $\tilde{F} \in \tilde{\mathscr{F}}_{\comp}$.
Therefore, the pair
\begin{equation}
\label{eq:pap8g02r1wmk}
\left(\overline{\tilde{\trans}_{A}(s)}, \quot{\phi}{\tilde{f}_{A}(s)} \circ \psi_{A, s}\right)
\eqlab{A}{=} \trans_{(A, \phi)}(s)
\eqlab{B}{\in} \kernel_F
\end{equation}
satisfies
\begin{equation}
\label{eq:aruv6d6t6ij1}
|\overline{\tilde{\trans}_{A}(s)}| \eqlab{C}{\leq} |A| - |\overline{\tilde{\trans}_{A}(s')}| \eqlab{D}{\leq} |A| - 1,
\end{equation}
where 
(A) follows from (\ref{eq:otmlflag72l3}),
(B) follows from $(A, \phi) \in \kernel_F$ by (\ref{eq:plg10theuj7z}),
(C) follows from (\ref{eq:3vxzhkr8quvj}),
and (D) follows from $\overline{\tilde{\trans}_{A}(s')} \in [\tilde{F}]$ and $\tilde{F} \in \mathscr{F}_{\ext}$.
Equations (\ref{eq:pap8g02r1wmk}) and (\ref{eq:aruv6d6t6ij1}) conflict with (\ref{eq:plg10theuj7z}).
Consequently, there exists at most one $s \in \mathcal{S}$ such that $f_{(A, \phi)}(s) = \lambda$.

Therefore, there exists $\hat{s} \in \mathcal{S}$ such that $|f_{(A, \phi)}(\hat{s})| \geq 1$ since $|\mathcal{S}| \geq 2$.
By $(A, \phi) \in \kernel_F$, for any $(A', \phi') \in [F]$, there exists $\pmb{x} \in \mathcal{S}^{\ast}$ such that 
$\trans^{\ast}_{(A', \phi')}(\pmb{x}) =  (A, \phi)$, so that
\begin{equation*}
|f^{\ast}_{(A', \phi')}(\pmb{x}\hat{s})| \eqlab{A}= |f^{\ast}_{(A', \phi')}(\pmb{x})f_{(A, \phi)}(\hat{s})| = |f^{\ast}_{(A', \phi')}(\pmb{x})| + |f_{(A, \phi)}(\hat{s})| \geq 1,
\end{equation*}
where
(A) follows from Lemma \ref{lem:f_T} (i).
This shows $F \in \mathscr{F}_{\ext}$.

\subsection{Proof of Lemma \ref{lem:expand-coding}}
\label{subsec:proof-expand-coding}

\begin{proof}[Proof of Lemma \ref{lem:expand-coding}]
The proof relies on \cite[Lemma 7]{IEICE2023}.

Using the mapping $\varphi \colon [\hat{F}] \to [F']$ defined as (\ref{eq:rg9zmfsn1btj}) and the analogous equations (\ref{eq:12yxaf48at5p}) and (\ref{eq:ohfzbcgbftvx}) to \cite[Eqs.~(23) and (24)]{IEICE2023},
we can show that for any $(A, \phi) \in [\hat{F}]$ and $\pmb{x} \in \mathcal{S}^{\ast}$, we have $|\hat{f}^{\ast}_{(A, \phi)}(\pmb{x})| = |f'^{\ast}_{\varphi((A, \phi))}(\pmb{x})| = |f'^{\ast}_A(\pmb{x})|$ by the almost identical argument to the proofs of \cite[Lemma 7 (i)]{IEICE2023} from \cite[Eqs.~(23) and (24)]{IEICE2023}.
\end{proof}

\subsection{Proof of Lemma \ref{lem:quot-calc}}
\label{subsec:proof-quot-calc}

\begin{proof}[Proof of Lemma \ref{lem:quot-calc}]
(Proof of (i))
We have
\begin{eqnarray*}
\left(\quot{\phi}{\pmb{d}}\right)^{\ast}(\pmb{b})
&\eqlab{A}=&  \quot{\left(\quot{\phi}{\pmb{d}}\right)}{\pmb{b}}(0) \\
&\eqlab{B}=& \quot{\phi}{\pmb{d}}(\pmb{b})^{-1} \quot{\phi}{\pmb{d}}(\pmb{b}0)\\
&\eqlab{C}=& \left( \phi(\pmb{d})^{-1} \phi(\pmb{db}) \right)^{-1} \phi(\pmb{d})^{-1} \phi(\pmb{db}0) \\
&=& \phi(\pmb{db})^{-1} \phi(\pmb{d}) \phi(\pmb{d})^{-1} \phi(\pmb{db}0)\\
&=& \phi(\pmb{db})^{-1} \phi(\pmb{db}0)\\
&\eqlab{D}=& \quot{\phi}{\pmb{db}}(0)\\
&\eqlab{E}=& \phi^{\ast}(\pmb{db}),
\end{eqnarray*}
where
(A) follows from (\ref{eq:phi-star}),
(B) follows from (\ref{eq:r6xuob4upqw5}),
(C) follows from (\ref{eq:r6xuob4upqw5}),
(D) follows from (\ref{eq:r6xuob4upqw5}),
and (E) follows from (\ref{eq:phi-star}).

(Proof of (ii))
We have
\begin{eqnarray*}
(\phi \circ \psi)^{\ast}(\pmb{b})
&\eqlab{A}=& \quot{\phi \circ \psi}{\pmb{b}}(0) \\
&\eqlab{B}=& (\phi \circ \psi)(\pmb{b})^{-1} (\phi \circ \psi)(\pmb{b}0) \\
&=& \phi(\psi(\pmb{b}))^{-1} \phi(\psi(\pmb{b}0)) \\
&\eqlab{C}=& \phi(\psi(\pmb{b}))^{-1} \phi \left( \psi(\pmb{b})\quot{\psi}{\pmb{b}}(0) \right) \\
&\eqlab{D}=& \phi(\psi(\pmb{b}))^{-1} \phi( \psi(\pmb{b}) ) \quot{\phi}{\psi(\pmb{b})}\left(\quot{\psi}{\pmb{b}}(0)\right) \\
&=& \quot{\phi}{\psi(\pmb{b})} \left(\quot{\psi}{\pmb{b}}(0)\right) \\
&\eqlab{E}=& \quot{\phi}{\psi(\pmb{b})} (\psi^{\ast}(\pmb{b})) \\
&\eqlab{F}=& \quot{\phi}{\psi(\pmb{b})}(0) \oplus  \psi^{\ast}(\pmb{b}) \\
&\eqlab{G}=& \phi^{\ast}(\psi(\pmb{b})) \oplus  \psi^{\ast}(\pmb{b}),
\end{eqnarray*}
where
(A) follows from (\ref{eq:phi-star}),
(B) follows from (\ref{eq:r6xuob4upqw5}),
(C) follows from Lemma \ref{lem:phi-quot} (i),
(D) follows from Lemma \ref{lem:phi-quot} (i),
(E) follows from (\ref{eq:phi-star}),
(F) follows from (\ref{eq:3ggv1jl6ft06}),
and (G) follows from (\ref{eq:phi-star}).

(Proof of (iii))
We have
\begin{eqnarray*}
(\phi^{-1})^{\ast}(\phi(\pmb{b}))
&=& (\phi^{-1})^{\ast}(\phi(\pmb{b})) \oplus 0\\
&=& (\phi^{-1})^{\ast}(\phi(\pmb{b})) \oplus \phi^{\ast}(\pmb{b}) \oplus \phi^{\ast}(\pmb{b})\\
&\eqlab{A}=& (\phi^{-1} \circ\phi)^{\ast}(\pmb{b}) \oplus \phi^{\ast}(\pmb{b}) \\
&=& (\mathrm{id}_{\mathcal{C}^{\ast}})^{\ast}(\pmb{b}) \oplus \phi^{\ast}(\pmb{b})\\
&\eqlab{B}=& \quot{\mathrm{id}_{\mathcal{C}^{\ast}}}{\pmb{b}}(0) \oplus \phi^{\ast}(\pmb{b})\\
&\eqlab{C}=& \left( \mathrm{id}_{\mathcal{C}^{\ast}}(\pmb{b})^{-1} \mathrm{id}_{\mathcal{C}^{\ast}}(\pmb{b}0) \right) \oplus \phi^{\ast}(\pmb{b})\\
&=& \left( \pmb{b}^{-1} \pmb{b}0  \right) \oplus \phi^{\ast}(\pmb{b})\\
&=& 0 \oplus \phi^{\ast}(\pmb{b})\\
&=& \phi^{\ast}(\pmb{b}),
\end{eqnarray*}
where
(A) follows from (ii) of this lemma,
(B) follows from (\ref{eq:phi-star}),
and (C) follows from (\ref{eq:r6xuob4upqw5}).
\end{proof}

\subsection{List of Notation}
\label{sec:notation}

\begin{longtable}{lp{0.8\textwidth}}
  $|\mathcal{A}|$ & the cardinality of a set $\mathcal{A}$, defined at the beginning of Section \ref{sec:preliminary}. \\
  $\mathcal{A}^k$ & the set of all sequences of length $k$ over a set $\mathcal{A}$, defined at the beginning of Section \ref{sec:preliminary}. \\
 $\mathcal{A}^{\geq k}$ & the set of all sequences of length greater than or equal to $k$ over a set $\mathcal{A}$, defined at the beginning of Section \ref{sec:preliminary}. \\
 $\mathcal{A}^{\leq k}$ & the set of all sequences of length less than or equal to  $k$ over a set $\mathcal{A}$, defined at the beginning of Section \ref{sec:preliminary}. \\
  $\mathcal{A}^{\ast}$ & the set of all sequences of finite length over a set $\mathcal{A}$, defined at the beginning of Section \ref{sec:preliminary}.\\
  $\mathcal{A}^{+}$ & the set of all sequences of finite positive length over a set $\mathcal{A}$, defined at the beginning of Section \ref{sec:preliminary}.\\
  $\bar{A}$ & the representative of $\langle A \rangle$, defined at the beginning of Subsection \ref{subsec:semi}.\\
  $\langle A \rangle$ & the equivalence class $A$ belongs to, defined above Lemma \ref{lem:Phi-group}.\\
   $\left[\mathcal{A}\right]_k$ & $\{[\pmb{x}]_k : \pmb{x} \in \mathcal{A}, |\pmb{x}| \geq k\}$ for a set $\mathcal{A}$ of sequences and an integer $k \geq 0$, defined at the beginning of Section \ref{sec:preliminary}.\\  
  $\mathcal{A} \times \mathcal{B}$ & the Cartesian product of sets $\mathcal{A}$ and $\mathcal{B}$, that is, $\{(a, b) : a \in \mathcal{A}, b \in \mathcal{B}\}$, defined at the beginning of Section \ref{sec:preliminary}. \\
  $\mathcal{C}$ & the coding alphabet $\mathcal{C} = \{0, 1\}$, at the beginning of Section \ref{sec:preliminary}.\\
  $\bar{c}$ & the negation of $c \in \mathcal{C}$, that is, $\bar{0} = 1, \bar{1} = 0$, defined 
  at the beginning of the proof of Theorem \ref{thm:orbit-num}.\\
  $\mathscr{C}^k$ & the power set of $\mathcal{C}^k$, defined after Lemma \ref{thm:differ}.\\
  $\tilde{\mathscr{C}}^k$ & the set of all representatives chosen and fixed at the beginning of Subsection \ref{subsec:semi}.\\
  $\mathscr{C}^k / {\sim_k}$ & the set of all the equivalence classes of $\mathscr{C}^k$ by $\sim_k$, defined after Lemma \ref{lem:Phi-group}.\\
  $f^{\ast}_i$ & defined in Definition \ref{def:f_T}. \\
  $|F|$ & the number of code tables of $F$, defined after Definition \ref{def:treepair} and at the beginning of Subsection \ref{subsec:semi}. \\
  $[F]$ & the domain of $F$, defined after Definition \ref{def:treepair} and at the beginning of Subsection \ref{subsec:semi}.\\
  $\mathscr{F}^{(m)}$ & the set of all $m$-code-tuples, defined after Definition \ref{def:treepair}.\\
  $\mathscr{F}$ & the set of all code-tuples, defined after Definition \ref{def:treepair}.\\
  $\mathscr{F}_{\ext}$ & the set of all extendable code-tuples, defined in Definition \ref{def:F_ext}. \\
  $\mathscr{F}_{k\hdec}$ & the set of all $k$-bit delay decodable code-tuples, defined in Definition \ref{def:k-bitdelay}. \\
  $\mathscr{F}_{k\hopt}$ & the set of all $k$-bit delay optimal code-tuples, defined in Definition \ref{def:optimalset}. \\
  $\mathscr{F}_{\irr}$ & the set of all irreducible code-tuples, defined in Definition \ref{def:F_irr}. \\
  $\mathscr{F}_{\reg}$ & the set of all regular code-tuples, defined in Definition \ref{def:regular}. \\
  $|\tilde{F}|$ & the number of code tables of an RCT $\tilde{F}$, defined after Definition \ref{def:semi}. \\
  $[\tilde{F}]$ & the domain of an RCT $\tilde{F}$, defined after Definition \ref{def:semi}\\
  $\tilde{\mathscr{F}}$ & the set of all RCTs, defined after Definition \ref{def:semi}.\\
  $\tilde{\mathscr{F}}_{\comp}$ & the set of all compliant RCTs, defined in Definition \ref{def:semi-class}.\\
  $\tilde{\mathscr{F}}_{\ext}$ & the set of all extendable RCTs, defined in Definition \ref{def:semi-class}.\\
  $\tilde{\mathscr{F}}_{k\hdec}$ & the set of all $k$-bit delay decodable RCTs, defined in Definition \ref{def:semi-class}.\\
  $\tilde{\mathscr{F}}_{k\hopt}$ & the set of all $k$-bit delay optimal RCTs, defined in Definition \ref{def:semi-class}.\\
  $\tilde{\mathscr{F}}_{\reg}$ & the set of all regular RCTs, defined in Definition \ref{def:semi-class}.\\
  $h(F)$ & defined after Lemma \ref{lem:potential}.\\
  $\langle \mathcal{I} \rangle$ & defined by (\ref{eq:1k1d4tzz0kr5}) after Lemma \ref{lem:Phi-group}.\\
  $L(F)$ & the average codeword length of a code-tuple $F$, defined in Definition \ref{def:evaluation}. \\
  $L_i(F)$ & the average codeword length of the $i$-th code table of $F$, defined in Definition \ref{def:evaluation}. \\
  $\tilde{L}(\tilde{F})$ & the average codeword length of a regular RCT $\tilde{F}$, defined in Definition \ref{def:semi-class}.\\
  $\mathcal{P}^k_{F, i}(\pmb{b})$ & defined in Definition \ref{def:pref}.\\
  $\bar{\mathcal{P}}^k_{F, i}(\pmb{b})$ & defined in Definition \ref{def:pref}.\\
  $\PREF^k_{\tilde{F}, A}$ & defined in Definition \ref{def:pref-tilde}.\\
  $\bar{\PREF}^k_{\tilde{F}, A}(\pmb{b}) $ & defined in Definition \ref{def:pref-tilde}.\\
  $Q(F)$ & the transition probability matrix, defined in Definition \ref{def:transprobability}.\\
  $Q_{i, j}(F)$ & the transition probability, defined in Definition \ref{def:transprobability}.\\
  $\kernel_F$ & defined in Lemma \ref{lem:kernel}.\\
  $\mathcal{S}$ & the source alphabet, defined at the beginning of Section \ref{sec:preliminary}.\\
  $\suff(\pmb{x})$ & the sequence obtained by deleting the first letter of $\pmb{x}$, defined at the beginning of Section \ref{sec:preliminary}. \\
  $\suff^k(\pmb{x})$ & the sequence obtained by deleting the first $k$ letters of $\pmb{x}$, defined at the beginning of Section \ref{sec:preliminary}. \\
  $x_i$ & the $i$-th letter of a sequence $\pmb{x}$, defined at the beginning of Section \ref{sec:preliminary}.\\
  $|\pmb{x}|$ & the length of a sequence $\pmb{x}$, defined at the beginning of Section \ref{sec:preliminary}.\\
    $\pmb{x} \preceq \pmb{y}$ & \pmb{x} is a prefix of \pmb{y}, defined at the beginning of Section \ref{sec:preliminary}.\\
    $\pmb{x} \prec \pmb{y}$ & $\pmb{x} \preceq \pmb{y}$ and $\pmb{x} \neq \pmb{y}$, defined at the beginning of Section \ref{sec:preliminary}. \\
  $\pmb{x}^{-1}\pmb{y}$ & the sequence $\pmb{z}$ such that $\pmb{x}\pmb{z} = \pmb{y}$, defined at the beginning of Section \ref{sec:preliminary}.\\
   $\pmb{x}\mathcal{A}$ & $\{\pmb{x}\pmb{y} : \pmb{y} \in \mathcal{A}\}$ for a sequence $\pmb{x}$ and a set $\mathcal{A}$ of sequences, defined at the beginning of Section \ref{sec:preliminary}.\\
 $[\pmb{x}]_k$ & the prefix of length $k$ of $\pmb{x}$, defined at the beginning of Section \ref{sec:preliminary}.\\
  $\lambda$ & the empty sequence, defined at the beginning of Section \ref{sec:preliminary}.\\
  $\mu$ & the source distribution $\mu \colon \mathcal{S} \to (0, 1) \subseteq \mathbb{R}$, defined at the beginning of Section \ref{sec:preliminary}. \\
  $\pmb{\pi}(F)$ & the unique stationary distribution of $F \in \mathscr{F}_{\reg}$, defined in Definition \ref{def:regular}.\\
  $\trans^{\ast}_i$ & defined in Definition \ref{def:f_T}. \\
  $\phi^{\ast}$ & defined in Definition \ref{def:phi-star}. \\
  $\quot{\phi}{\pmb{d}}$ & defined in Definition \ref{def:phi-quot}.\\
  $\Phi_k$ & defined in Definition \ref{def:Phi}. \\
  $\phi(\mathcal{X})$ & the image $\phi(\mathcal{X}) \coloneqq \{\phi(a) : a \in \mathcal{X}\}$ of a set $\mathcal{X}$ under a mapping $\phi$, defined at the beginning of Section \ref{sec:preliminary}.\\
  $\phi \circ \psi$ & the composite $\phi(\psi(\cdot))$ of two mappings $\phi$ and $\psi$, defined at the beginning of Section \ref{sec:preliminary}.\\
  $\oplus$ & the addition modulo $2$, that is, $0 \oplus 0 = 1 \oplus 1 = 0, 0 \oplus 1 = 1 \oplus 0 = 1$, defined in Lemma \ref{lem:phi-represent}.\\
  $\sim_k$ & defined above Lemma \ref{lem:Phi-group}.\\
  \end{longtable}
\vspace{9pt}

\section*{Acknowledgment}
This work was supported in part by JSPS KAKENHI Grant Numbers JP24K23859 and JP24K14818.

\nocite{*}
\bibliographystyle{IEEE}

\begin{thebibliography}{99}
\bibitem{Huffman1952}
	 D.~A.~Huffman,
	 ``A Method for the Construction of Minimum-Redundancy Codes,''
	 \emph{Proc. I.R.E.},
	 vol.~40, no.~9, pp.~1098--1102, 1952.

\bibitem{JSAIT2022}
	K.~Hashimoto and K.~Iwata,
	 "Optimality of Huffman Code in the Class of $1$-bit Delay Decodable Codes,"
	 in \emph{IEEE Journal on Selected Areas in Information Theory},
	 doi: 10.1109/JSAIT.2022.3230745,
	 \emph{arXiv:2209.08874}. [Online].~Available: http://arxiv.org/abs/2209.08874

\bibitem{IEICE2023}
	K.~Hashimoto and K.~Iwata, 
	``Properties of $k$-bit Delay Decodable Codes,''
	in \emph{The IEICE Transactions on 	Fundamentals of Electronics, Communications and Computer Sciences},
	vol.~E107-A, no.~3, pp.~417--447, Mar.~2024,
	doi:10.1587/transfun.2023TAP0016.

\bibitem{Hashimoto2021}
	K.~Hashimoto and K.~Iwata, 
	``On the Optimality of Binary AIFV Codes with Two Code Trees,''
	in \emph{Proc.~IEEE International Symposium on Information Theory} (ISIT),
	Melbourne, Victoria, Australia (Virtual Conference), Jul.~2021, pp. 3173--3178.

\bibitem{Hashimoto2023}
	K.~Hashimoto and K.~Iwata, 
	``The Optimality of AIFV Codes in the Class of 2-bit Delay Decodable Codes,''
	 \emph{arXiv:2306.09671}.\newline
	  [Online].~Available: https://arxiv.org/abs/2306.09671

\bibitem{Yamamoto2015}
	 H.~Yamamoto, M.~Tsuchihashi, and J.~Honda, 
	 ``Almost Instantaneous Fixed-to-Variable Length Codes,''
	 \emph{IEEE Trans.\ on Inf.\ Theory},
	 vol.~61, no.~12, pp.~6432--6443, Dec.~2015.
\bibitem{Hu2017}
	W.~Hu, H.~Yamamoto, and J.~Honda, 
	``Worst-case Redundancy of Optimal Binary AIFV Codes and Their Extended Codes,''
	 \emph{IEEE Transactions on Information Theory},
	vol.~63, no.~8, pp.~5074--5086, Aug.~2017.


 \bibitem{GolinH23}
	 M.~J.~Golin and E.~Y.~Harb,
	 ``A Polynomial Time Algorithm for Constructing Optimal Binary AIFV-$2$ Codes,''
	 \emph{IEEE Transactions on Information Theory},
	 vol.~69, no.~10, pp.~6269--6278, Oct.~2023.
\bibitem{Golin2021}
	 M.~J.~Golin and E.~Y.~Harb,
	``Speeding up the AIFV-2 dynamic programs by two orders of magnitude using Range Minimum Queries,''
	\emph{Theoretical Computer Science},
	vol.\ 865, no.\ 14, pp.~99--118, Apr.~2021.

\bibitem{Fujita2018}
	R.~Fujita, K.~Iwata and H.~Yamamoto,
	``An Optimality Proof of the Iterative Algorithm for AIFV-m Codes,''
	\emph{2018 IEEE International Symposium on Information Theory (ISIT)},
	Vail, CO, USA, 2018, pp.~2187--2191,
	doi: 10.1109/ISIT.2018.8437861.
 \bibitem{Golin2022}
	 M.~J.~Golin and A.~J.~L.~Parupat,
	 ``Speeding Up AIFV-$m$ Dynamic Programs by $m-1$ Orders of Magnitude,''
	 in \emph{Proc.~IEEE International Symposium on Information Theory} (ISIT),	 
	 Espoo, Finland, Jun.--Jul.~2022, pp.~282--287. 		


\bibitem{Fujita2019}
	R.~Fujita, K.~Iwata, and H.~Yamamoto,
	``An Iterative Algorithm to Optimize the Average Performance of Markov Chains with Finite States,''
	 in \emph{Proc.~IEEE International Symposium on Information Theory} (ISIT),
	Paris, France,  Jul.~2019, pp.~1902--1906.
\bibitem{Golin2024}
	 R.~H.~Dolatabadi, M.~J.~Golin, and A.~Zamani,
	 ``Better Algorithms for Constructing Minimum Cost Markov Chains and AIFV Codes,''
	 arXiv:2405.06831, May.~10, 2024.		

	
 \bibitem{Sugiura2022}
	 R.~Sugiura, Y.~Kamamoto, and T.~Moriya,
	 ``General Form of Almost Instantaneous Fixed-to-Variable-Length Codes,''
	 \emph{IEEE Transactions on Information Theory},
	 vol.~69, no.~12, pp.~7672--7690, Dec.~2023.
\bibitem{Sugiura2023}
	 R.~Sugiura, M.~Nishino, N.~Yasuda, Y.~Kamamoto, and T.~Moriya,
	 ``Optimal Construction of $N$-bit-delay Almost Instantaneous Fixed-to-Variable-Length Codes,''
	 arXiv:2311.02797v1, Nov.~5, 2023.	 
	 
	 	 
\bibitem{OEIS}
	\emph{The On-Line Encyclopedia of Integer Sequences}, 
	A007501,
	https://oeis.org/A007501.
\end{thebibliography}

%

\end{document}